\documentclass[10pt, oneside,reqno]{amsart}

\input ./math_headers_amsart.sty

\bibliography{Twist} 

\title{Higher Deformation Quantization for Kapustin--Witten Theories}
\author{Chris Elliott \and Owen Gwilliam \and Brian R. Williams}
\date{\today}

\begin{document}

\begin{abstract}
We pursue a uniform quantization of all twists of 4-dimensional $\mathcal N=4$ supersymmetric Yang--Mills theory, using the BV formalism, and we explore consequences for factorization algebras of observables.  
Our central result is the construction of a one-loop exact quantization on $\RR^4$ for all such twists and for every point in a moduli of vacua. 
When an action of the group $\mathrm{SO}(4)$ can be defined --- for instance, for Kapustin and Witten's family of twists --- the associated framing anomaly vanishes.  
It follows that the local observables in such theories can be canonically described by a family of framed $\mathbb E_4$ algebras; 
this structure allows one to take the factorization homology of observables on any oriented 4-manifold.  
In this way, each Kapustin-Witten theory yields a fully extended, oriented 4-dimensional topological field theory {\em \`a la} Lurie and Scheimbauer.
\end{abstract}

\maketitle

\renewcommand{\baselinestretch}{0.5}\normalsize
\setcounter{tocdepth}{2}
\tableofcontents
\renewcommand{\baselinestretch}{1.0}\normalsize

\section{Introduction}

In \cite{KapustinWitten} Kapustin and Witten offered a view from quantum field theory
on the geometric Langlands correspondence,
relating S-duality for 4-dimensional $\mc{N} =4$ supersymmetric gauge theories to Langlands duality.
Subsequently, there has been an enormous amount of work
at this crossroads between quantum field theory, representation theory,  topology, and higher category theory.  To give just a few examples, we refer to the research program of Ben-Zvi, Gunningham, Nadler and collaborators \cite{BZNCharacter, BZNBetti, BZGN}, the work of Ben-Zvi, Brochier and Jordan \cite{BBJ1, BBJ2} and the results on skein algebras that have followed them \cite{Cooke, GunninghamJordanSafronov}, and the work of Frenkel and Gaiotto \cite{GaiottoBC, FrenkelGaiotto}. 
Much of the mathematical work has used methods involving higher categories and derived geometry,
in much the same spirit as homological mirror symmetry reworks the physicists' view on duality for $\mc{N} =(2,2)$ supersymmetric sigma models.
In this paper, we explore another path,
using Lagrangian field theory as physicists conventionally would but then exploiting recent mathematical progress \cite{CostelloBook, Book1, Book2} to extract higher algebraic structures in a new way from the proposal of Kapustin and Witten.

More precisely, the aim of this paper is to give an explicit construction of an important family of 4-dimensional quantum gauge theories, 
and to explore some of the higher algebraic structures that arise from the construction.  
Each quantum gauge theory we study arises from a classical field theory, 
by which we mean a Lagrangian field theory;
all our theories will be associated with twists of 4-dimensional $\mc{N} =4$ supersymmetric Yang-Mills theory on $\RR^4$.
Our family lives over the space $\CC^3 \times [\gg^*/G]$, where $G$ is a reductive group and $[\gg^*/G]$ is the coadjoint quotient stack of $G$.
The points of $\CC^3$ parametrize twists.
Over generic points in $\CC^3$, the theories we construct will be topological;
the theories emphasized by Kapustin and Witten appear among these twists.
The points of $[\gg^*/G]$ can be understood as classical vacua that make sense for all twists simultaneously.
We construct quantizations via Feynman diagrams and the BV formalism,
and we analyze the observables of these perturbative quantizations.
We show here, among other results, that the observables are $\bb E_4$-algebras,
or algebras over the little 4-disks operad,
a higher algebraic structure introduced and developed by topologists.
Such algebras have intrinsic interest in topology: 
one can obtain interesting geometric invariants by computing the \emph{factorization homology} of an $\bb E_4$ algebra on an appropriate 4-manifold.
It is an interesting question to ask how this kind of higher algebraic structure relates to prior work on Kapustin-Witten theories that uses higher algebra.

\begin{remark} \label{statesobservables_remark}
There are several distinct notions for what it means to ``construct'' a topological quantum field theory. 
Loosely speaking, one can attempt to model either the \emph{states} of the theory, or the \emph{observables} of the theory.  
In the language of (extended) functorial field theories (as developed by Atiyah \cite{AtiyahTQFT}, Baez--Dolan \cite{BaezDolan}, Lurie \cite{LurieCobordism} and many others), 
this choice manifests itself in the choice of target $(\infty,n)$-category for the functor.  
For example, for two-dimensional field theories, one might choose to construct a theory valued in an $(\infty,2)$-category of dg-categories --- this would be an example of a ``theory of states'' --- or one might choose to construct a theory valued in a suitable version of a higher Morita category of $\bb E_2$-algebras --- this would be an example of a ``theory of observables.''  
In this paper we will take the latter view. 
The relevant higher Morita categories and their relationship with factorization algebras have been developed in detail by Scheimbauer~\cite{ScheimbauerThesis},
precisely to give a home to such theories of observables.

We would like to emphasize an appealing aspect of our work here: 
the theory of observables can be constructed in full starting from a Lagrangian description of a field theory, and applying a suitable version of deformation quantization (here, using the Batalin--Vilkovisky formalism).  
While the theories of states play a crucial role in the geometric Langlands correspondence, for example, there does not yet exist a systematic procedure for their construction starting from a Lagrangian field theory; 
instead, various informal constructions and ans\"{a}tze must be used along the way.~\hfill$\Diamond$
\end{remark}

Let us contextualize the theories we will be constructing.  There is a rich family of four-dimensional gauge theories that can be obtained by the procedure known as twisting.  One starts with a Yang--Mills gauge theory coupled to scalar and spinorial matter, with the special property that the action of the Poincar\'e symmetry group can be extended to an action of a $\ZZ/2\ZZ$-graded extension known as a super Poincar\'e group; such theories are known as super Yang--Mills theories.  If one chooses an odd element $Q$ of the Lie algebra of the super Poincar\'e group with the property that $[Q,Q]=0$, the \emph{twist} by $Q$ is a new gauge theory obtained, roughly, by deforming the action functional using the infinitesimal action of $Q$.

By applying this procedure, one obtains theories that are more mathematically tractable but still deeply interesting: \emph{holomorphic} and \emph{topological field theories}.  Twists of supersymmetric Yang--Mills theories in dimension 4 include the famous Donaldson--Witten theory \cite{WittenTQFT}, which was the main motivating example for Witten's introduction of the twisting procedure and which lead to a physical origin for the Donaldson invariants of 4-manifolds.  More recently, Kapustin and Witten \cite{KapustinWitten} studied a family of twists to give a gauge-theoretic origin story for the categorical geometric Langlands conjecture of Beilinson and Drinfeld.

In this paper we will study the twists of maximally supersymmetric Yang--Mills theories in dimension 4, known as $\mc N=4$ super Yang--Mills theories.  These twists include the Kapustin--Witten family, as well as other examples, such as the mixed holomorphic-topological twist first studied by Kapustin \cite{KapustinHolo}.  A description of all twists of Yang--Mills theory at the classical level was given in \cite{ESW} using the Batalin--Vilkovisky (BV) formalism (see in particular Section 10.3 of loc. cit. for twists of 4d $\mc N=4$ theories).  In this paper we will analyze these twisted theories at the quantum level.

\begin{remark}
As we mentioned in Remark \ref{statesobservables_remark} above, we will construct quantum field theories modelling the observables for the Kapustin--Witten twisted gauge theories.  
One could also attempt to construct theories of states.  
Upon compactification on a Riemann surface $\Sigma$, 
one obtains a two-dimensional field theory; 
the literature often refers to geometric Langlands topological quantum field theories to mean two-dimensional functorial theories valued in dg-categories, 
assigning the categories of interest in the geometric Langlands theory to the point (see, for instance \cite{BZNBetti, BZGN, EY3}).  
~\hfill$\Diamond$\end{remark}

\begin{remark}
Our results do not rely on any aspect of S-duality, nor do we attempt in this paper to understand how S-duality interacts with our results.  Our construction models a family of perturbative quantum field theories varying over a version of the classical moduli of vacua; 
it would be very interesting to investigate what additional data is required in order to realize an S-duality relationship between constructions of this type.~\hfill$\Diamond$
\end{remark}

For the rest of this section, let us fix a complex reductive group $G$, viewed as the gauge group of super Yang--Mills theory.  We will suppress $G$ in the statements below.

Our first main result is as follows.

\begin{theorem}[See Theorem \ref{anomaly_theorem}, Theorem \ref{quantum_vacua_family_thm}.]
All twists of $\mc N=4$ Yang--Mills theory on $\RR^4$ admit an exact one-loop quantization. 
These quantizations extend to families of quantum field theories living over the quotient stack~$[\gg^*/G]$.
\end{theorem}

In Section~\ref{modvac} we interpret this quotient stack~$[\gg^*/G]$ as a moduli of vacua that works for all twists simultaneously.
There are some subtleties here, however,
as the view from derived geometry indicates that the moduli of vacua looks like $[\gg^*[2]/G]$, where the dual of the Lie algebra is shifted down in cohomological degree by~2. 
As we are working with the theories as $\ZZ/2$-graded objects,
this difference in the shift of $\gg^*$ is not visible in our construction,
and it is fruitful to work with the ordinary (underived) stack as we do.

\begin{remark}
It is possible to eliminate this even grading shift by careful incorporation of the action of R-symmetries (as discussed in \cite[Section 3.4]{EY2}), 
but in this paper we take a simpler approach, and forget the $\ZZ$-grading on our algebras down to a $\ZZ/2\ZZ$-grading.
~\hfill$\Diamond$\end{remark}

This result can be thought of as consisting of two steps.  First, we verify the existence of a \emph{prequantization} of the classical twisted Yang--Mills theories: a collection of effective theories depending on a scale parameter, satisfying the renormalization group flow condition.  Twisted field theories in general permit particularly nice renormalizations, using general results about holomorphic field theories (for example, no counter-terms occur).

Next, we compute the anomaly associated to the prequantization: the obstruction to solving the quantum master equation.  It is obtained as the weight of Feynman diagrams.  We show that this anomaly vanishes because the weight is zero for one-loop Feynman diagrams and no Feynman diagrams with more than one-loop can occur.

\begin{remark} 
In the appendix of \cite{CostelloSUSY}, Costello showed that, if it exists, there is a unique quantization of the holomorphic twist (up to equivalence of BV theories) that preserves various natural symmetries of the classical theory, like translation invariance and $R$-symmetry (see Remark \ref{uniqueness_remark} for further discussion of Costello's result).  Here we produce such a quantization, and quantizations of the other twists, and analyze some of its features.
Costello was certainly aware that results of this flavor existed, 
as he originally suggested the Landau gauge fixing we use and encouraged efforts in this direction!
~\hfill$\Diamond$\end{remark}

Given our family of quantum field theories, we can then study the structure of their local observables.  These observables can be modelled using the machinery of \emph{factorization algebras}, as in \cite{Book1, Book2}.  For those twists which are topological, these factorization algebras admit an alternative model, familiar to homotopy theorists, as $\bb E_4$-algebras: algebras over the operad of little $4$-disks.  In conjunction with results from \cite{ElliottSafronov}, we prove the following.

\begin{theorem}
Let $(\CC^3 \bs (\CC \cup_{\{0\}} \CC)) \times [\gg^*/G]$ denote the subspace of pairs $(Q, [x])$, where $Q$ is a topological supercharge and $[x]$ is a choice of vacuum.
There is a sheaf of $\bb E_4$-algebras over this space, 
given by the factorization algebras of local observables in the topological twists of $\mc N=4$ super Yang--Mills theory.
\end{theorem}

There is a useful interpretation of this sheaf, exposing an analogy with deformation quantization.
Here the moduli of vacua $[\gg^*/G]$ parametrizes a space of translation-invariant solutions to the equations of motion,
and we can consider its formal neighborhood inside all solutions.
(Loosely speaking, think about the tubular neighborhood.)
This neighborhood can be encoded by a family over $[\gg^*/G]$ of classical perturbative theories on $\RR^4$.
(Loosely speaking, identify the tubular neighborhood with the normal bundle.)
There is a natural shifted symplectic structure on each theory,
and when we quantize, we deform the dg commutative algebra of each theory to an $\bb E_4$-algebra.
In short, we produce a family of $\bb E_4$-stacks over $[\gg^*/G]$.
Note the analogy to deformation quantization, where one deforms the structure sheaf from a commutative to an associative (or $\bb E_1$-)algebra.

Having proven the existence of the quantization, we show that it can be equipped with rich additional structure. There are two main types of structure that we consider.
\begin{enumerate}
\item First, a \emph{framing structure}.  The method of factorization homology (as introduced in \cite[Section 5.5]{LurieHA} \cite{FrancisEn}) allows one to take an $\bb E_n$ algebra $\mc A$ and ``integrate'' it over an $n$-manifold $M$ with trivialized tangent bundle ({\it aka} a framed manifold).  From the point of view of field theory, we can think of this as defining a theory whose algebra of local observables is $\mc A$ and whose global observables on $M$ are the factorization homology.  To extend this method to {\em oriented} $n$-manifolds, one needs additional data: an action of the group $\SO(n)$ on $\mc A$ compatible with the $\bb E_n$-structure, with a trivialization up to homotopy.  An algebra with this structure is known as a \emph{framed $\bb E_n$ algebra}.
 
In the world of twisted supersymmetric field theories on $\RR^n$, one can obtain framings in this sense in a natural way.  One must start with the action of the group $\mr{ISO}(n)$ of isometries of $\RR^n$ on the theory, and find a \emph{homotopical trivialization}.  Roughly speaking, one must show that the conserved currents for the $\mr{ISO}(n)$-action are exact, with a family of potentials compatible with the Lie bracket on the Lie algebra of $\mr{ISO}(n)$.
 
For a subclass of the theories at hand -- the two-dimensional family of twists studied by Kapustin and Witten -- we construct such a homotopy trivialization for the action of $\mr{ISO}(4)$ on the classical twisted super Yang--Mills theories, and then go on to show that the quantum anomaly for this action (an example of a ``framing anomaly'') vanishes to all orders.  In particular, we obtain the following result.
 
\begin{theorem}[See Corollary \ref{framed E4}] 
There is a sheaf of framed $\bb E_4$-algebras over the space $(\CC^2 \bs \{0\}) \times [\gg^*/G]$ corresponding to the local observables for the Kapustin-Witten family of twisted 4d $\mc N=4$ theories.
\end{theorem}
 
\item We additionally consider a natural \emph{filtration} on the local observables. Our $\bb E_4$-algebras, when considered in the absence of this filtration, are not very interesting: they are merely commutative.  However, when we incorporate the data of the filtration, we obtain genuine non-trivial filtered $\bb E_4$-algebras. 
 
This natural filtration on the local observables in any quantum field theory is called the \emph{free-to-interacting} filtration (see Section \ref{F2I_section}).  The associated graded of this filtration can be identified with the quantum observables of the underlying free theory, and at higher terms in the corresponding spectral sequence, we see both the higher order terms in the classical interaction, as well as quantum correction terms, occurring at prescribed levels.  We exploit this filtration when we study our twisted theories.
 
\begin{prop}
The family of factorization algebras of quantum observables over the space $\CC^3 \times [\gg^*/G]$ can be equipped with a natural filtration.  In particular, the subfamilies with the structure of $\bb E_4$- or framed $\bb E_4$-algebras can be promoted to families of filtered $\bb E_4$- or filtered framed $\bb E_4$-algebras.
\end{prop}
 
\end{enumerate}

Having constructed these families of (framed, filtered) $\bb E_4$ algebras, we can begin to compute the factorization homology on interesting curved 4-manifolds $M$.  In this paper we prove some results about the factorization homology where $M$ is compact, and where $M = N \times \RR$ for $N$ a compact oriented 3-manifold.  We foresee even more interesting applications related to the geometric Langlands program for the factorization homology on $M = S \times \RR^2$ where $S$ is a compact oriented 2-manifold, which we will describe further below.

\begin{prop}[See Proposition \ref{holo_4d_fact_hom_prop} and \ref{B_4d_fact_hom_prop}]
If $M$ is a compact complex surface that can be realized as a discrete quotient of an open subset $U$ of $\CC^2$, then we can compute the factorization homology of the holomorphic twist on $M$.  This factorization homology is given by $\det(\mr H^{\bullet, \bullet}_{\ol \dd}(M) \otimes \gg)[d_M]$, where $d_M$ coincides with the Euler characteristic of $M$ modulo 2.

For \emph{any} compact oriented 4-manifold $M$ we can compute the factorization homology of the B-twist on $M$.  This factorization homology is given by $\det(\mr H^{\bullet}_{\mr{dR}}(M) \otimes \gg)[d_M]$.
\end{prop}

Things become even more interesting when we consider the factorization homology on the product $N \times \RR$, where $N$ is a compact oriented 3-manifold.  When we compute this factorization homology we obtain an $\bb E_1$ algebra (essentially a dg associative algebra).  We can study these $\bb E_1$ algebras using the data of a filtration: we can compute the first page of the spectral sequence precisely, and then investigate the modifications obtained by passing to later spectral sequence pages.  In particular, for the B-twist we obtain the following description.

\begin{theorem}[See Corollary \ref{nontriv def}]
The factorization homology of the B-twist on $N \times \RR$ is an algebra of differential operator type (see Definition \ref{def DOT}).  It can be computed using a spectral sequence whose first page is the algebra of differential operators on the formal moduli space associated to the abelian graded Lie algebra $\mr H^\bullet(N) \otimes |\gg|$, 
where $|\gg|$ is the underlying vector space of the gauge Lie algebra.
\end{theorem}

As we mentioned above, such constructions have the potential for many further interesting applications when we consider the factorization homology associated to a Riemann surface.
Especially interesting to us is the possibility of giving a direct, mathematical passage from the Kapustin--Witten theories to the quantum geometric Langlands correspondence.
Let us sketch how such an application might be developed, although we do not study this idea in detail in this paper.

Choose a closed, oriented 2-manifold $S$ and consider the compactification of a Kapustin-Witten theory $\mathcal T$ along $S$,
which produces a 2-dimensional topological field theory $\mc{T}_S$.
Our work immediately implies that the factorization algebra of observables $\obsq_{\mc{T}_S}$ of this theory is given by the framed $\bb E_2$-algebra $\int_S \obsq_{\mc{T}}$, namely the factorization homology of our framed $\bb E_4$-algebra over~$S$.

This construction is rather abstract but there are two concrete things to try:
\begin{itemize}
\item It is possible, using the machinery of \cite{Book1} and \cite{LiVertex}, to extract a dg vertex algebra from this construction,
and so the results of this paper lead to a plethora of potentially novel vertex algebras to study, each labeled by a twisting parameter $Q$, a complex reductive group $G$, a coadjoint orbit $x \in [\gg^*/G]$, and a surface~$S$.

\item Such $\bb E_2$ algebras appear throughout 2-dimensional field theory, often under more familiar names like Gerstenhaber algebras or Batalin--Vilkovisky algebras.
In the setting of homological mirror symmetry and the extended topological field theories of Baez--Dolan--Lurie, such algebras often appear as the Hochschild cohomology of categories of branes.
Optimistically, and following the ideas of \cite{EY2}, we might hope that $\obsq_{\mc{T}_S}$ has a natural relationship with the $\bb E_2$ algebras already known in geometric Langlands theory, such as those studied in \cite{ArinkinGaitsgory}.  
\end{itemize}

\begin{remark}
Let us elaborate on the last point a little more.  Given an $\bb E_{n-1}$-monoidal category $\mc C$, one can compute its Hochschild cochains, or the algebra of derived endomorphisms of the unit object.  
This process will yield an $\bb E_n$ algebra, which one views as parameterizing $\bb E_{n-1}$-monoidal deformations of $\mc C$.  
Most relevant for us here is the monoidal category $\mr{Sph}_G$, the spherical Hecke category of the geometric Langlands program (expected to arise as the category of line operators in an A- or B-twist of $\mc N=4$ gauge theory). 
Its algebra of Hochschild cochains can be identified with $\sym^\bullet(\gg[-2])^G$ as a filtered $\bb E_2$ algebra (see \cite[Section 12.4]{ArinkinGaitsgory}).  

Beraldo \cite{BeraldoTC} pushes this further and views $\mr{Sph}_G$ as an $\bb E_3$-monoidal category, using the 3-dimensional pair of pants, so the Hochschild cochains will have the structure of an $\bb E_4$-algebra.  
Beraldo then computes the factorization homology of $\mr{Sph}_G$ on a $d$-manifold $M$,
obtaining an $\bb E_{3-d}$-monoidal category that he identifies with the category $\bb H(\mr{LS}^{\mr{Betti}}_G(M))$ defined in \cite{BeraldoH}.
It would be very interesting to compare our results with Beraldo's, 
as his $\bb E_4$-algebra appears to coincide with the $\bb E_4$-algebra that we construct in this paper as quantum local observables of the B-twisted theory.~\hfill$\Diamond$\end{remark}

We plan to pursue these questions in the future, and we welcome conversations on these topics.

In other directions, our techniques apply to other gauge theories, 
notably the construction of families of factorization algebras (e.g., $\bb E_n$ algebras) over moduli of vacua.
It would be interesting to explore how these analogs of deformation quantization relate to prior work by physicists in related contexts, such as Seiberg--Witten theory.

\subsection*{Acknowledgements}
The authors would like to thank David Ben-Zvi, Justin Hilburn, Pavel Safronov, and Philsang Yoo for their helpful comments during the preparation of this paper. 
CE and OG would like to thank the participants of the informal Friday lunch seminar at UMass for providing much impetus towards this project, so our appreciation goes to Filip Dul, Andreas Hayash, Tina Kanstrup, Ivan Mirkovi\'c, and Alexei Oblomkov.
The National Science Foundation supported O.G. through DMS Grant No. 1812049.

\section{Review of the BV Formalism}

\subsection{Classical BV Formalism} \label{BV_review_section}

In this article we will study classical and quantum field theory using the \emph{Batalin--Vilkovisky (BV) formalism} \cite{BatalinVilkovisky}.  This is a model for perturbative field theory using homological algebra or, more accurately, formal derived geometry.  For a more detailed account of these techniques we refer the reader to~\cite{Book2}.

We think of a classical field theory as being encoded by the derived critical locus of its equations of motion.  Formally locally near a given solution, this derived critical locus can be described by a formal moduli space together with a $(-1)$-shifted symplectic structure. 
There is a convenient way to model a formal moduli space using purely algebraic objects:
by the fundamental theorem of deformation theory \cite{Hinich, LurieSAG}, a formal moduli space can be modeled by a dg Lie algebra or $L_\infty$ algebra.
Concretely, this claim means that for a formal space $\mc{X}$, there is an $L_\infty$ algebra $\gg_{\mc{X}}$ such that each map $f: \spec(R) \to \mc{X}$ corresponds to a solution of the Maurer-Cartan equation in $\mf{m}_R \otimes \gg_{\mc{X}}$, where $R$ is a local artinian connective dg algebra with maximal ideal $\mf{m}_R$.
Hence, we offer the following definition of a classical field theory.

\begin{definition}
A perturbative \emph{classical field theory} on a manifold $M$ is a sheaf $\mc E$ of $\ZZ$-graded $C^\infty_M$-modules on $M$ together with
\vspace{-10pt}
\begin{enumerate}
\item A $L_\infty$ algebra structure $\{[\cdot]_k\}$ on the shift $\mc E[-1]$ where every higher bracket $[\cdot]_k$ is a polydifferential operator.  We will often denote the one-fold bracket $[\cdot]_1$ by~$\d_{\mc E}$.
\item A $(-1)$-shifted symplectic structure $\omega$.  That is, a skew-symmetric non-degenerate pairing
\[\omega \colon \mc E \otimes \mc E \to \mr{Dens}_M[-1]\]
of graded $C^\infty_M$-modules.
\end{enumerate}
The pairing must be invariant (i.e., cyclic) for the $L_\infty$ brackets.
\end{definition}

We call such a sheaf of a $L_\infty$ algebras, where each bracket is built from differential operators, a {\em local} $L_\infty$ algebra. 
We will sometimes refer to $(\mc E, \d_{\mc E})$ as the \emph{classical BV complex} of the classical field theory.
Note that we will later weaken the $\ZZ$-grading in the definition to a mere $\ZZ/2$-grading. 

\begin{remark}
A feature (at times, irritating) of our definition is that the fields $\mc E$ do not form a local $L_\infty$ algebra, but instead form a {\em shifted} local $L_\infty$ algebra.
It is often convenient simply to describe the local $L_\infty$ algebra directly and leave it to the reader to recover the fields by applying a cohomological shift down by~1. 
~\hfill$\Diamond$\end{remark}

To connect with more standard ways of describing field theories,
it is helpful to bear in mind the following construction.
The local $L_\infty$ structure on $\mc E$ defines an action functional $S$ for the theory by
\[
S(\varphi) = \omega(\varphi, \d_{\mc E} \varphi) + \sum_{k \geq 2} \frac{1}{k!} \int_M \omega(\varphi, [\varphi, \ldots, \varphi]_k).
\]
The first term is quadratic, while the rest
\[
I(\varphi) = \sum_{k \geq 2} \frac{1}{k!} \int_M \omega(\varphi, [\varphi, \ldots, \varphi]_k)
\] 
we call the classical interaction.
The action functional satisfies the \emph{classical master equation}
\[\d_{\mc E}I + \frac 12 \{I,I\} = 0.\]
Here the shifted Poisson bracket $\{-,-\}$ on local functionals is defined using the shifted symplectic structure on fields.
In fact, every such solution to the classical master equation is equivalent to a perturbative classical field theory as in the above definition. 
Hence it is straightforward to translate between traditional presentations of classical BV theories (by action functionals) and the presentation by local $L_\infty$ algebras with shifted pairings.
(See \cite[Chapter 5.3]{CostelloBook} and \cite[Section 1.2]{ESW} for details.)

\begin{remark}
Together with the differential $\d_{\mc E} + \{I,-\}$ the bracket $\{-,-\}$ equips the (shifted) space of local functionals $\oloc(\mc E)[-1]$ with the structure of a dg Lie algebra. 
The classical master equation is, in fact, equivalent to the Maurer--Cartan equation in this dg Lie algebra.
Hence this dg Lie algebra $\oloc(\mc E)[-1]$ encodes the formal deformations of the classical field theory. 
~\hfill$\Diamond$\end{remark}

The theories we will study in this paper, arising as twists of supersymmetric gauge theories, are of a particularly nice class: they are holomorphic in the following sense.

\begin{definition}
A classical field theory on a complex manifold $X$ is called \emph{holomorphic} if $\mc E$ arises as the sheaf of smooth sections of a graded holomorphic vector bundle equipped with a holomorphic differential operator, and the local $L_\infty$ brackets on $\mc E[-1]$ are given by holomorphic polydifferential operators.
\end{definition}

As we will discuss in the next section, perturbative quantization becomes cleaner and simpler when we restrict to holomorphic theories~\cite{BWhol}.

We'll build classical field theories using a few key constructions, which we now summarize.
\begin{example} \label{CS_example}
 Let $\gg$ be an $L_\infty$ algebra equipped with a non-degenerate invariant pairing. 
 \vspace{10pt}
 \begin{enumerate}
  \item Let $X$ be a Calabi--Yau manifold of dimension $m$.
  Consider the space of fields
  \[ \mc E = (\Omega^{0,\bullet}_X \otimes \gg[1], \ol \dd),\]
  with $L_\infty$ structure on $\mc E[-1]$ inherited from the $L_\infty$-structure on $\gg$ and the wedge pairing of Dolbealt forms.  The Calabi--Yau structure on $X$ and the invariant pairing on $\gg$ equips $\mc E$ with a $(2-m)$-shifted symplectic pairing.  
   \item Let $Y$ be a complex manifold of dimension $n$.  Consider the space of fields
  \[\mc E = (\Omega^{\bullet,\bullet}_Y \otimes \gg[1], \ol \dd),\]
  with $L_\infty$ structure on $\mc E[-1]$ inherited from the $L_\infty$-structure on $\gg$ and the wedge pairing of $(p,q)$ forms.  The invariant pairing on $\gg$ equips $\mc E$ with a $(2-2n)$-shifted symplectic pairing.
   \item Let $M$ be a smooth manifold of dimension $d$.  Consider the space of fields
  \[\mc E = (\Omega^{\bullet}_M \otimes \gg[1], \d_{\mr{dR}}),\]
  with $L_\infty$ structure on $\mc E[-1]$ inherited from the $L_\infty$-structure on $\gg$ and the wedge pairing of differential forms.  The invariant pairing on $\gg$ equips $\mc E$ with a $(2-d)$-shifted symplectic pairing.
  \item We can combine these constructions, and consider the theory on $X \times Y \times M$ whose fields are
  \[\mc E = (\Omega^{0,\bullet}_X \otimes \Omega^{\bullet, \bullet}_Y \otimes \Omega^\bullet_M \otimes \gg[1], \ol \dd_X + \ol \dd_Y + \d_M).\]
  It is equipped with a $(2-d-m-2n)$-shifted symplectic structure.
 \end{enumerate}
 \end{example}
 
So far, these constructions only define classical field theories under an assumption on the dimension, ensuring that the shifted symplectic structure has degree $-1$ (such theories are called \emph{generalized Chern--Simons theories}).  We can define, however, a classical field theory in any dimension using the following construction.
 
\begin{definition}
Let $L$ be a local $L_\infty$ algebra on $M$ given by smooth sections of a finite-dimensional graded vector bundle on $M$.  
The \emph{cotangent theory} associated to $L$ is the theory whose fields are
\[\mc E = L[1] \oplus L^![-2],\]
where $L^! = L^* \otimes \dens_M$.  The $L_\infty$ structure on $\mc E[-1]$ is defined using the $L_\infty$ structure on $L$ and the coadjoint action of $L$ on $L^*$, and the shifted symplectic pairing is defined by pairing together the two summands. 
\end{definition}

\begin{definition}
Let $X$ and $Y$ be complex manifolds, and let $M$ be a smooth manifold.  We define a \emph{generalized BF theory} on $X \times Y \times M$ to be the cotangent theory to the local $L_\infty$ algebra
\[(\Omega^{0,\bullet}_X \otimes \Omega^{\bullet, \bullet}_Y \otimes \Omega^\bullet_M \otimes \gg , \ol \dd_X + \ol \dd_Y + \d_M).\]
See \cite[Section 1.6.1]{ESW} for more details.
\end{definition}

 We will conclude this section with one more piece of notation: a formal version of Simpson's Hodge stack $X_{\mr{Hod}}$ of a smooth scheme $X$ \cite{Simpson}.  
Simpson defined a derived stack $X_{\mr{Hod}}$ over $\bb A^1$ that interpolates between the de Rham stack of $X$ at the point $1 \in \bb A^1$, and the Dolbeault stack of $X$ (the 1-shifted tangent space $T[1]X$) at $t=0$.

\begin{definition}
The \emph{Hodge algebra} $L_{t\text{-Hod}}$ of an $L_\infty$ algebra $L$ is defined to be the $\CC[t]$-linear $L_\infty$ algebra $(L[1] \oplus L) \otimes \CC[t]$ with $L_\infty$ structure inherited from the $L_\infty$ structure on $L$ along with the adjoint action of $L$ on $L[1]$, where we add the operator $t \cdot \mr{id} \colon L \otimes \CC[t][1] \to L \otimes \CC[t]$ to the differential.  

We write $L_\mr{Dol}$ for the $L_\infty$ algebra obtained by evaluating the parameter at $t=0$.  We write $L_{\mr{dR}}$ for the contractible $L_\infty$ algebra obtained by evaluating the parameter at $t=1$.  
\end{definition}

To motivate this terminology, we unpack some consequences of this definition.
Let $L$ be an $L_\infty$ algebra that we view as encoding a formal moduli space $\mc{X}$.
Under the dictionary of derived deformation theory, the algebra $\mc{O}(X)$ of functions on $\mc{X}$ is encoded by $\clie^\bu_{\mr{Lie}}(L)$, the Chevalley-Eilenberg cochains of $L$.
Observe that 
\[
\clie^\bu_{\mr{Lie}}(L_{\mr{Dol}}) = \clie^\bu_{\mr{Lie}}(L, \sym(L^*[-2])) = \mc{O}(T[1]\mc{X}),
\]
the de Rham forms on $\mc{X}$ (i.e., the complex without the differential!).
Similarly,
\[
\clie^\bu_{\mr{Lie}}(L_{\mr{dR}}) = \Omega(\mc{X}),
\]
the de Rham complex. 
Indeed, Simpson's deformation turns on the de Rham differential.

\subsection{Quantum BV formalism} \label{quantum_section}

Let us now discuss quantization in the BV formalism, as articulated in \cite{CostelloBook}.  The theories we will study in this paper will all be particularly amenable to quantization.  The general procedure has two steps.
\begin{enumerate}
 \item First, we construct an effective collection of perturbative field theories, i.e. a collection $\{I[\Lambda]\}_{\Lambda \in \RR_{>0}}$ of effective interaction functionals in $\OO^+(\mc E)[\![\hbar]\!]$, which are compatible with renormalization group flow.  That is, for $\Lambda_1 < \Lambda_2$ we can determine $I[\Lambda_2]$ as the sum of weights associated to Feynman diagrams using the scale $\Lambda_1$ interaction $I[\Lambda_1]$.  
We refer to such a set as an {\em effective collection} of interactions.
 
 There is a standard method for the construction of such a collection of effective interactions by constructing counterterms associated to the singular parts of certain Feynman integrals.  
 
 \item Next, we must verify that the effective collection we have constructed satisfies the scale $\Lambda$ \emph{quantum master equation} for every value of $\Lambda$.  At a fixed scale $\Lambda$, one can define a shifted Poisson bracket $\{-,\}_\Lambda$ on local functionals and BV operator $\Delta_\Lambda$ inverse to the symplectic form.  The quantum master equation encodes the following compatibility between the effective action and the BV operator:
\begin{equation}\label{eqn:qme}
\d_{\mc E}I[\Lambda] + \frac 12\{I[\Lambda],I[\Lambda]\}_\Lambda + \hbar \Delta_\Lambda I[\Lambda] = 0.
\end{equation}
In practice, there may be an obstruction to the existence of a quantization of a given classical field theory solving the quantum master equation up to a given order in $\hbar$.  This obstruction, sometimes called an {\em anomaly}, is given by a class in $\mr H^1(\oloc(\mc E))$ that can be computed explicitly in terms of Feynman diagrams for the given theory.
\end{enumerate}

We gain a lot, however, by restricting attention to holomorphic theories of cotangent type.  First, let us use the cotangent type assumption.

The following simple combinatorial fact will be extremely useful for understanding the Feynman diagrams that can occur in cotangent type theories.

\begin{lemma}\label{graph_lemma}
Every connected directed multigraph where each vertex includes at most one incoming edge takes the form of a directed loop, with a number of outgoing directed trees attached.
\end{lemma}

\begin{corollary} \label{cotangent_cor}
The only non-zero Feynman diagrams that can occur in a theory of cotangent type have at most one loop.
\end{corollary}

\begin{proof}
Consider a cotangent theory with BV complex $\mc E = T^*[-3]\mc L$, and denote fields in the base factor $\mc L$ by $\alpha$, and fields in the fiber factor $\mc L^![-3]$ by $\beta$.  The shifted Poisson structure pairs $\alpha$ and $\beta$ fields, so we can direct the edges of Feynman diagrams for such a theory by labelling one input of the propagator by $\alpha$, and the other by $\beta$.  The $L_\infty$ structure on $\mc E$ is inherited from the $L_\infty$ structure on $\mc L$ and the action of $\mc L$ on its dual, which means that all vertices in the Feynman diagram include either zero or one $\beta$ factors.  The claim then follows from Lemma~\ref{graph_lemma}.
\end{proof}

\begin{example} \label{tadpole_example}
For an example of a wheel, consider the ``tadpole'' diagram, corresponding to a wheel with one external leg.  
Suppose that our classical BV complex takes the form
\[\mc E = \mc E' \otimes \gg,\]
where $\mc E'$ is a $(-3)$-shifted symplectic complex, and where $\gg$ is a finite-dimensional Lie algebra equipped with a non-degenerate invariant pairing (such theories need not necessarily be of cotangent type).  
The weight of the diagram when evaluated on a field $\phi \otimes X \in \mc E' \otimes \gg$ is proportional to the trace $\mr{tr}_\gg(X)$ of $X$, taken in the adjoint representation.  
If $\gg$ is unimodular --- in particular, if $\gg$ is reductive --- this trace automatically vanishes and therefore the weight of the diagram does too.  
\end{example}

Let us now take advantage of the holomorphic assumption on a classical field theory.  We can use this assumption in two ways.  Firstly, the construction of a prequantization is straightforward: there are no counter-terms.

\begin{theorem}[{\cite[Theorem 3.1]{BWhol}}] \label{prequant_thm}
 Any holomorphic classical field theory on $\CC^2$ admits a one-loop effective collection involving no counterterms.  
 In particular, any holomorphic theory of cotangent type admits such an effective collection that is exact at one-loop.
\end{theorem}

\subsubsection{The holomorphic gauge}
\label{sec:holgauge}

It will be useful to explicitly characterize the form of the effective action $I[\Lambda]$ for these types of theories. 
As mentioned above it is given as a sum of weights of Feynman graphs 
\[
I[\Lambda] = \sum_{\Gamma} W_\Gamma(P_{0<\Lambda} , I) 
\]
where the sum is over connected directed multigraphs and $P_{0<\Lambda}$ is the scale $\Lambda > 0$ propagator.

The propagator is constructed by the following method.
A characterizing feature of a holomorphic theory on $\CC^n$ is that the kinetic term in the action involves the $\dbar$ operator for some holomorphic bundle. 
The propagator $P_{0<\Lambda}$ is a regularization of the integral kernel for the distributional operator $\dbar^{-1}$. 
Our choice of regularization starts by fixing the flat metric on $\CC^n$ and uses the standard heat kernel regularization of the inverse Laplacian $\triangle^{-1}$. 
We then use the formal distributional equation $\dbar^{-1} = \dbar^* \triangle^{-1}$ to define our propagator.  
Here $\dbar^*$ is the gauge-fixing operator. 
Details of this construction can be found in~\cite[Section 3]{BWhol}.

The propagator is a distribution valued in the symmetric square of the space of fields ${\rm Sym}^2(\Bar{\mc{E}})$. 
The space of fields of any holomorphic theory on $\CC^2$ is of the form $\Omega^{0,\bu}(\CC^2) \otimes E_0$ where $E_0$ is some vector space,
so we can view the propagator as a distributional section of $E \boxtimes E$ over $\CC_z^2 \times \CC_w^2$.
As such, it takes the following form
\begin{equation} \label{propagator_eq}
P_{0<\Lambda} (z,w) = \int_{t=0}^\Lambda \d t \frac{1}{(2 \pi i t)^d} \sum_{j=1}^d (-1)^{j-1}  \left(\frac{\zbar_j - \Bar{w}_j}{4 t} \right)  e^{-|z-w|^2 / 4t}  \prod_{i \ne j}^d (\d \zbar_i - \d \Bar{w}_i) \otimes c_{E_0}
\end{equation}
where $c_E$ is an element of ${\rm Sym}^2(E_0)$ that is specific to the given theory.

The weight $W_\Gamma(P_{0<\Lambda} , I)$ associated to a graph is evaluated as follows. 
The graph $\Gamma$ consists of two types of vertices, univalent ones and non-univalent ones. 
The univalent vertices label the input fields.
One associated the interaction $I$ to each internal (i.e. non-univalent) vertex (which necessarily has valency $\geq 3$). 
One associates the propagator $P_{0<\Lambda}$ to each internal edge.  
The weight of the diagram is obtained from the expression 
\[P_{0<\Lambda}^{\otimes \#\text{internal edges}} \otimes I^{\otimes \#\text{internal vertices}}\]
by tensor contraction according to the shape of the diagram.  
If the graph has $v$ vertices and $e$ edges, this weight is presented as an integral over $\CC^{2v}$ with $e$ insertions of the propagator $P_{0<\Lambda}$. 
See \cite[Section 2.3.6]{CostelloBook} for a more detailed account.
See also \cite[Section 3]{BWhol} specifically for the construction of weights in holomorphic theories.

If you unwind these definitions in the case of a simple wheel-shaped Feynman diagram (see Figure \ref{fig:anomaly} below), with $n$ trivalent vertices, one finds the following expression for the weight:
\[
I_n[\Lambda](\alpha) = 
\int_{(z_1, \ldots, z_n) \in (\CC^2)^n}
\tr\big( \alpha(z_1) \wedge P_{0<\Lambda}(z_1,z_2) \wedge \alpha(z_2) \wedge \cdots \wedge \alpha(z_n) \wedge P_{0<\Lambda}(z_n,z_1) \big),
\]
where the propagator $P_{0<\Lambda}$ is computed using the expression from equation \ref{propagator_eq} above.  In the limit as $\Lambda$ goes to infinity, the propagator is a kind of inverse to $\d_{\mc E}$, so this expression can be through of a version of $\tr(\d_{\mc E}^{-n})$. Altogether, the one-loop quantum correction can be seen as encoding a determinant.

\subsubsection{One-loop Anomalies}

Now it remains to compute the obstruction to the one-loop anomaly to the quantum master equation: if this vanishes then a holomorphic cotangent theory admits a one-loop exact quantization.  We can compute this anomaly in terms of a single Feynman diagram using the following result.

\begin{theorem}[{\cite[Lemmas 4.6 and 4.7]{BWhol}}] \label{wheel_thm}
The one-loop anomaly for a holomorphic field theory on $\CC^d$ is given by the weight associated to a wheel Feynman diagram with $(d+1)$ internal edges, see Figure~\ref{fig:anomaly}. 
\begin{figure}[h]
\begin{center}
\begin{tikzpicture}[line width=.2mm, scale=1.5]
		\draw[fill=black] (0,0) circle (1cm);
		\draw[fill=white] (0,0) circle (0.99cm);
		\draw (140:2) -- (145:1);
		\draw (145:2) -- (145:1);
		\draw (150:2) -- (145:1);
		\draw (210:2) -- (215:1cm);
		\draw (220:2) -- (215:1cm);
		\draw (27.5:2) -- (35:1);
		\draw (42.5:2) -- (35:1);
		\draw (-35:2) -- (-35:1cm);
        \clip (0,0) circle (1cm);
\end{tikzpicture}
\caption{An example of a wheel Feynman diagram, i.e. a one-loop Feynman diagram whose internal edges are arranged in a closed loop.}
\label{fig:anomaly}
\end{center}
\end{figure}
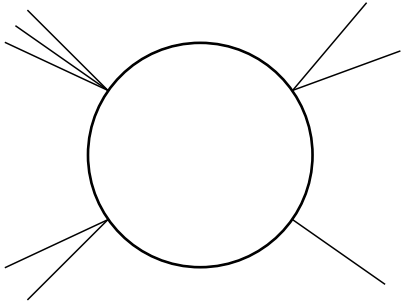
\end{theorem}

As a consequence, given a holomorphic theory of cotangent type on $\CC^d$, to check that it admits a one-loop exact quantization we only need to verify the vanishing of a wheel Feynman diagrams with $(d+1)$ vertices.

\subsection{Equivariant BV formalism} \label{equivariant_section}

Fix a classical BV theory $(\mc E, I, \omega)$. 
Let $(\mf h, [-,-]_{\mf h}, \d_{\mf h})$ be a dg Lie algebra. 
Recall that $\oloc (\mc E)[-1]$ is equipped with the structure of a dg Lie algebra with differential $\d_{\mc E} + \{I,-\}$ and bracket $\{-,-\}$. 

\begin{definition}
A {\em classical (Hamiltonian) action} of $\mf h$ on the classical BV theory $(\mc E, I, \omega)$ is a map of dg Lie algebras
\[
\rho \colon \mf h \to \oloc(\mc E) [-1] .
\]
\end{definition}

When discussing actions we often drop the word Hamiltonian as we will only consider this sort of Lie algebra action on a classical field theory in this paper.
Notice that through the BV bracket, any Hamiltonian action defines an action on the observables 
\[
\rho \colon \mf h \to \oloc(\mc E) [-1] \xto{\{-,-\}} {\rm End}(\mr{Obs}) .
\]

\begin{remark}
As in the ordinary (non-equivariant) BV formalism, there is a sort of master equation that governs Hamiltonian actions. 
Indeed, $\rho$ defines an element 
\begin{equation}\label{eqn:irho}
I_{\rho} \in \clie^\bu(\mf h) \otimes \oloc(\mc E) [-1] 
\end{equation}
of cohomological degree $1$. 
The condition that $\rho$ be a map of dg Lie algebras is equivalent to the following {\em equivariant} classical master equation
\[
(\d_{\mf h} \otimes 1) I_{\rho} + (1 \otimes \d_{\mc E}) I_{\rho} + \{I, I_{\rho}\} + \frac12 \{I_{\rho}, I_{\rho}\} = 0 .
\]
The graded vector space appearing in \eqref{eqn:irho} has the natural structure of dg Lie algebra.
The master equation is the Maurer--Cartan equation for this dg Lie algebra. 
~\hfill$\Diamond$\end{remark}

Quantization of equivariant BV theories proceeds along similar lines as the non-equivariant case. 
One starts with a (non-equivariant) quantum BV theory prescribed by some effective collection $\{I[\Lambda]\}$.
Next, one constructs an effective collection of functionals that depend both on the fields $\mc E$ and $\mf h$ and asks that an appropriate equivariant quantum master equation holds.

The notion of a renormalization group flow is defined for a subspace
\[
\mc O ^+(\mf h [1] \oplus \mc E) [[\hbar]] \subset \mc O(\mf h [1] \oplus \mc E) [[\hbar]],
\]
which consists of functionals which are at least cubic modulo $\hbar$ when restricted to $\mc E$ and satisfy mild technical support and smoothness conditions.
We refer to \cite[Definition 13.2.1.1]{Book2} for a precise definition. 

With the notion of renormalization group flow in hand it makes sense to ask for an equivariant collection $\{I_{\rho}[\Lambda]\}_{\Lambda > 0}$ which satisfy the renormalization group equations. 
We refer to such a set of functionals as an equivariant effective collection. 

\begin{definition}
Let $\{I[\Lambda]\}$ be a quantum BV theory. 
A {\em quantum action} of $\mf h$ on this theory is an equivariant effective collection
\[
\{I_{\rho}[\Lambda]\} \in \mc O^+(\mf h[1] \oplus \mc E) / \mc O(\mf h[1]) 
\]
that satisfies:
\begin{itemize}
\item[(1)] the equivariant quantum master equation, which states that 
\begin{equation}\label{eqn:eqqme}
\d_{\mf h} I_{\rho}[\Lambda] + \d_{\mc E} I_{\rho} [\Lambda] + \frac12 \{I_{\rho}[\Lambda] , I_{\rho}[\Lambda] \}_\Lambda + \hbar \Delta_\Lambda I_{\rho} [\Lambda] = 0 .
\end{equation} 
\item[(2)] 
Under the natural map 
\[
\mc O ^+(\mf h [1] \oplus \mc E) [[\hbar]] \to \mc O ^+(\mc E) [[\hbar]]
\]
given by restricting to functionals just of $\mc E$, the image of $I_\rho [\Lambda]$ is the original action $I[\Lambda]$ defining the non-equivariant theory. 
\item[(3)] Locality and support conditions. 
\end{itemize}
\end{definition}

We refer to \cite[Definition 13.2.2.1]{Book2} for a discussion of the locality and support conditions which we do not recall here.
They are designed in such a way that modulo $\hbar$ the $\Lambda \to 0$ limit of the functional $I_{\rho}[\Lambda]$ returns a classical Hamiltonian action.
 
\begin{remark}\label{rmk:eqcollection}
The holomorphic gauge will also apply to this equivariant situation, so long as the $\mf{h}$-action only involves holomorphic differential operators \cite{GRWthf}. 
The naive effective collection $\{I_{\rho}[\Lambda]\}$ is constructed in a similar way as in the non equivariant theory. 
One starts with a classical equivariant action $I_{\rho}$ and forms the sum over graphs 
\[
I_{\rho}[\Lambda] = \sum_{\Lambda} W_{\Gamma} (P_{0<\Lambda}, I_{\rho}) .
\]
Again, the univalent vertices are labeled by input fields, which can now be elements of $\mf{h}$ or $\mc{E}$. 
We assume that the set of univalent vertices labeled by $\mc{E}$ is non empty. 
The non-univalent vertices are labeled by $I_{\rho}$. 
The edges are labeled by $P_{0<\Lambda}$, the same as in the non equivariant case. 
This is where the term ``background field'' method comes from, since we are not treating the Lie algebra $\mf{h}$ as a propagating field. 
~\hfill$\Diamond$\end{remark}

\section{Twists of 4d \texorpdfstring{$\mc N=4$}{N=4} Super Yang--Mills Theory}

\subsection{Summary of Twists}\label{summary_twist_section}

Given a classical field theory with an action of the abelian super Lie algebra $\Pi \CC$, there is a procedure called \emph{twisting} that modifies the differential, by adding a generator of the $\Pi \CC$ action.  One can generate many examples of such actions by starting with a supersymmetric field theory, i.e. a theory equipped with an action of a $\ZZ/2\ZZ$-graded extension of a complexified Poincar\'e algebra, and choosing an odd element $Q$ of the algebra such that $[Q,Q]=0$ (often referred to as a ``nilpotent supercharge'').  We refer to \cite{ElliottSafronov,ESW} for a complete description of the twisting procedure in the BV formalism, as well as a classification of all twists of supersymmetric Yang--Mills theories.

In this section we'll first describe the possible twists of 4d $\mc N=4$ super Yang--Mills theory on $\RR^4$ in terms of the supersymmetry algebra, and then we'll describe those twists in the language of the classical BV formalism.

\begin{definition}
The \emph{$\mc N=k$ supertranslation algebra} in four dimensions is the complex super Lie algebra with underlying super vector space 
\[\mc T_k = \CC^4 \oplus \Pi(W \otimes S_+ \oplus W^* \otimes S_-)\]
where $W$ is a $k$-dimensional vector space, with a non-trivial Lie bracket 
\[(W \otimes S_+) \otimes (W^* \otimes S_-) \to \CC^4\]
inherited from the $\spin(4)$-equivariant isomorphism $S_+ \otimes S_- \to \CC^4$ and the evaluation pairing $W \otimes W^* \to \CC$.
\end{definition}

An odd element $Q$ of $\mc T_k$ can be identified with a pair of maps
\[Q_+ \colon S_+^* \to W, \quad Q_- \colon W \to S_-.\]
Such an element satisfies $Q^2 = 0$ if and only if the composite map $Q_- \circ Q_+ \colon S_+^* \to S_-$ is zero.  We can stratify the moduli space of nilpotent supercharges by the ranks of the two maps $Q_+$ and $Q_-$.  In other words, by specifying a pair of integers $(n_+,n_-)$ where $n_\pm$ are at most 2.

\begin{remark} \label{22_supercharge_rmk}
In most cases, there is no difference between two twists of super Yang--Mills theories associated to supercharges of the same rank, because the action of the group $\spin(4) \times \SL(W)$ is transitive on supercharges of a given rank.  There is an exception, however, for rank $(2,2)$ supercharges in the $\mc N=4$ supertranslation algebra.

A square-zero rank $(2,2)$ supercharge $Q$ induces a linear isomorphism $S_+^* \oplus S_- \to W$.  Indeed, any supercharge $(Q_+, Q_-) \in S_+ \otimes W \oplus S_- \otimes W^*$ induces linear maps $S_+^* \to W$ and $W \to S_-$; being a rank $(2,2)$ square-zero supercharge is equivalent to the statement that the sequence
\[0 \to S_+^* \to W \to S_- \to 0\]
is a split short exact sequence.

The $\SL(4;\CC)$-representation $W$ carries a canonical volume form, as does the sum $S_+^* \oplus S_-$, and the $\spin(4) \times \SL(W)$-orbits in the space of square-zero correspond to the ratio $q$ between these two volume forms (see~\cite[Section 4.4]{ElliottSafronov}).  

One can show that twists of super Yang--Mills theory by supercharges with $q \ne 1$ are all equivalent (``generic'' rank $(2,2)$ supercharges), but they differ from the twist by such a supercharge at the point $q=1$ (''special'' rank $(2,2)$ supercharges).
~\hfill$\Diamond$\end{remark}

Now, let us describe the theories produced by twisting $\mc N=4$ super Yang--Mills theory on $\RR^4$ by these classes of supercharge.  Let $\gg$ be a reductive Lie algebra equipped with a fixed nondegenerate equivariant pairing (that, in particular, induces a canonical isomorphism between $\gg$ and $\gg^*$).

\begin{remark} \label{family_def}
When we refer to a \emph{family} of quantum field theories over an affine space $\CC^n$, we mean a quantum field theory where the effective interactions are valued in the ring $\CC[t_1, \ldots, t_n]$, and where all the relevant structures are linear over this ring.  For a detailed definition we refer to~\cite[Section 7.3]{Book2}.
~\hfill$\Diamond$\end{remark}

Let $\eps$ be a formal variable of degree~$-1$. 

\begin{theorem}[{\cite[Section 10.3]{ESW}}] \label{twist_theorem}
There is a three-parameter family of twists of $\mc N=4$ super Yang--Mills theory on $\CC^2$ that can be described perturbatively by the $\CC[t_1,t_2,u]$-linear space of fields 
\[\mc E = \bigg(\Omega^{\bullet, \bullet}\big(\CC_{z_1} \times \CC_{z_2}, \gg[1] \oplus \eps \gg[2]\big)[t_1,t_2,u] \; , \; \ol \dd + t_1 \dd_{z_2} + t_2 \dd_{z_2} + u \frac{\d}{\d \eps} \bigg) .\]
We will write a general field as $\alpha + \eps \beta$, where 
\[
\alpha \in \Omega^{\bu,\bu}(\CC^2, \gg[1]) \quad\text{ and }\quad  \beta \in \Omega^{\bu,\bu}(\CC^2, \eps \gg[2]).
\]
The interaction is
\[
I (\alpha, \beta) = \frac12 \int \beta \wedge [\alpha, \alpha] .
\]
where the bracket is encoded by the wedge product of differential forms and the Lie bracket on $\gg$. 
\end{theorem}

Notice that the interaction $I$ is independent of the parameters, the only dependence on the parameters $u,t_1,t_2$ is in~$\d_{\mc E}$. 
In the terminology of Section~\ref{BV_review_section}, we have a family of generalized Chern--Simons theories valued in the Hodge Lie algebra $\gg_{u \text{-Hod}}$ associated to~$\gg$.

We can match up subspaces of this three-dimensional family of twists with the classes of nilpotent supercharges discussed above (see \cite[Section 10.3]{ESW}).
\begin{itemize}
 \item The point $(0,0,0)$ is a \emph{holomorphic} twist, obtained from a supercharge of rank $(1,0)$. It is a holomorphic BF theory.
 \item Points of the form $(t_1,0,0)$ or $(0,t_2,0)$ with $t_i \ne 0$ are obtained from supercharges of rank $(1,1)$.  These twists are holomorphic in one complex direction, and topological in the other.  They are often referred to as \emph{Kapustin twists}~\cite{KapustinHolo}. 
 \item Points of the form $(t_1,t_2,0)$ with both $t_1$ and $t_2 \ne 0$ are obtained from special rank $(2,2)$ supercharges, as in Remark~\ref{22_supercharge_rmk}.
 \item Points of the form $(0,0,u)$ with $u \ne 0$ are topological twists obtained from supercharges of rank $(2,0)$.
 \item Points of the form $(t_1,0,u)$ or $(0,t_2,u)$ with both $t_i$ and $u \ne 0$ are topological twists obtained from supercharges of rank~$(2,1)$.
 \item Points of the form $(t_1,t_2,u)$ with all three coordinates non-zero are topological twists obtained from generic supercharges of rank $(2,2)$.  These twists occur in the quantum geometric Langlands conjecture at generic values of the level.
\end{itemize}

There is the two-dimensional family $\mc E_{u=0}$ that we refer to as the family of ``B-type'' twists of $\mc N=4$ super Yang--Mills theory.  
We refer to subfamilies with $u \ne 0$ as ``A-type'' twists. 

\begin{remark} \label{A_B_name_remark}
We use the terminology of A- and B-type twists following \cite{KapustinWitten}.  Kapustin and Witten study the dimensional reduction of twists of 4d $\mc N=4$ super Yang--Mills theories along a curve $C$, where one obtains the twist of an $\mc N=(2,2)$ supersymmetric sigma model whose target is the Hitchin system $\higgs_G(C)$.  These twists are either A-models (so that the category of boundary conditions is modelled by a category of twisted D-modules), or B-models (so that the category of boundary conditions is modelled by a category of coherent sheaves).  
~\hfill$\Diamond$\end{remark}

\begin{figure}[h]
\begin{center}
\tdplotsetmaincoords{75}{130}
\begin{tikzpicture}[tdplot_main_coords,fill opacity=0.2]

\draw (-3,-3,-3) --  (-3,-3,3) -- (-3,3,3) -- (-3,3,-3) -- cycle; 
\draw (3,-3,-3) --  (3,-3,3) -- (3,3,3) -- (3,3,-3) -- cycle; 
\draw (-3,-3,-3) --  (-3,-3,3) -- (3,-3,3) -- (3,-3,-3) -- cycle; 
\draw (-3,3,-3) --  (-3,3,3) -- (3,3,3) -- (3,3,-3) -- cycle; 
\draw[fill=gray]  (-3,0,3) -- (3,0,3) -- (3,0,-3) -- cycle; 
\draw[fill=green] (3,-3,-3) -- (3,3,-3) -- (-3,3,3) -- (-3,-3,3) -- cycle;
\draw[fill=gray]  (3,0,-3) --(-3,0,-3) -- (-3,0,3) -- cycle; 

\draw[thick, red] (0,-5,0) -- (0,5,0);
\draw[line width = 2pt] (0,0,-3) -- (0,0,3);
\draw[line width = 1.5pt] (-3,0,0) -- (3,0,0);

\node[opacity=1] (line) at (3.6,0,0) {$t_1$};
\node[opacity=1] (line) at (0,0,3.6) {$t_2$};
\node[opacity=1, red] (line) at (0,4,0.5) {$u$};
\fill[black,opacity=1] (0,0,0) circle (3pt);
 
\end{tikzpicture}
\end{center}
\caption{The space $\CC^3$ of square zero supercharges that we are considering.  The vertical plane in grey is the two-dimensional space of B-type twists.  The diagonal plane in green indicates the two-dimensional subspace of supercharge compatible with the $\SO(4)$ Kapustin--Witten twisting homomorphism. }
\end{figure}
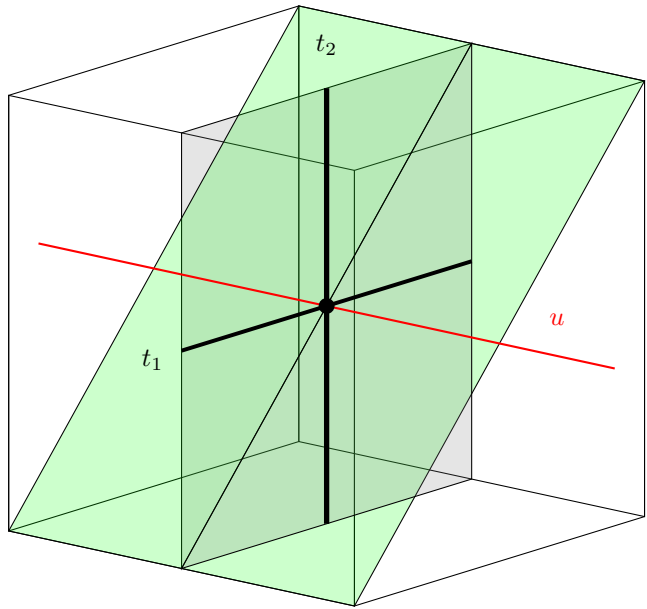

\begin{remark}
All twists of $\mc N=4$ super Yang--Mills theory occur in families of this type.  If one fixes a holomorphic supercharge $Q_{\mr{hol}}$, the space of nilpotent supercharges commuting with $Q_{\mr{hol}}$ is three-dimensional.  One does need to take care, however: the identification between such a space of supercharges and the space of deformations of the holomorphic twist is non-linear.  See \cite{EY3} for further details. 
~\hfill$\Diamond$\end{remark}

\subsection{Anomaly Vanishing} \label{anomaly_section}

In this section we will establish the following theorem.
\begin{theorem} \label{anomaly_theorem}
The three-parameter family of classical field theories of Theorem \ref{twist_theorem} admits a one-loop exact quantization to a three-parameter family of quantum field theories on $\RR^4$.
\end{theorem}

In other words, we will construct a $\CC[t_1,t_2,u]$-linear solution to the quantum master equation whose $\hbar \to 0$ limit agrees with the classical family of Theorem \ref{twist_theorem}.
We will do this using the method discussed in Section~\ref{quantum_section} because the family $\mc E$ consists of holomorphic theories.  

We will proceed in two steps.
\begin{enumerate}
\item First, we argue that the theory $\mc E_{0,0,0}$ over the special point $u=t_1=t_2=0$ admits a one-loop exact quantization. 
\item We show that this quantization extends $\CC[t_1,t_2,u]$-linearly to a quantization over the three-dimensional family of theories.  
\end{enumerate}
The quantum corrections that we will construct will, rather remarkably, be \emph{constant} over the three dimensional space $\CC^3_{t_1,t_2,u}$ of twisting parameters.  In other words, the effective collection $\{I[\Lambda]\}$ of interactions will be independent of the parameters $t_1,t_2$ and $u$.

\begin{remark} \label{uniqueness_remark}
In fact, if $\gg$ is semisimple then the family of quantizations that we construct is uniquely characterized by the fact that, for points of the form $(t_1,t_2,0) \in \CC^3$, the quantization preserves the cotangent structure, in the sense that the effective interactions $I[\Lambda]$ only involves the $\alpha$ fields.  One can show this by a computation involving the obstruction-deformation complexes of the twisted theories.

Let us begin by discussing the obstruction-deformation complex of the B-twisted theory $\mc E_{1,1,0}$ (a similar computation will apply for any theory $\mc E_{t_1,t_2,0}$ with $t_1$ and $t_2$ both non-zero).  Using \cite[Lemma 3.5.4.1]{Book2} we can see that the obstruction-deformation complex takes the form
\begin{align*}
\mr C^\bullet_{\text{red,loc}}(\mc E_{1,1,0}) &\iso \Omega^\bullet(\RR^4; \mr C^\bullet_{\mr{red}}(\gg[\eps]))[4] \\
&\iso \mr C^\bullet_{\mr{red}}(\gg, \sym^\bullet(\gg^*[-2]))[4].
\end{align*}
In particular, $\mr H^0_{\text{red,loc}}(\mc E_{1,1,0})$ can be identified with $\mr H^4(\gg) \oplus (\mr H^0(\gg) \otimes \sym^2(\gg^*)^\gg)$, which is one-dimensional if $\gg$ is semisimple.  Restricting to quantizations compatible with the cotangent structure is equivalent to taking $\CC^\times$-invariants of this complex, where $\CC^\times$ acts on $\eps$ with weight one.  When we take these invariants, the resulting cohomology group is trivial. So the quantization of the theory $\mc E_{1,1,0}$ (and a fortiori the quantization of the full family $\mc E$ of theories) is unique \footnote{The deformation that we've realized here that breaks the cotangent structure appears in our story too, it is the deformation of the B-twisted theory $\mc E_{1,1,0}$ to an A-twisted theory $\mc E_{1,1,u}$.}.

One can also calculate the obstruction-deformation complex of the holomorphically twisted theory $\mc E_{0,0,0}$.  This was discussed by Costello in \cite[Theorem 21.0.3]{CostelloSUSY}, where the part of the obstruction-deformation complex of the holomorphic twist consisting of functionals invariant for the action of a certain group of symmetries (namely $(\CC^\times \ltimes \CC^2) \times \GL(2;\CC)$, where the first factor acts on spacetime, and the latter factor acts by R-symmetries) was calculated.  
~\hfill$\Diamond$\end{remark}

\subsubsection{The Holomorphic Twist}

We will begin by establishing the existence of the quantization for a large class of holomorphic theories. Consider \emph{holomorphic BF theory} on $\CC^2$, whose space of fields is
\[(A,B) \in \Omega^{0,\bullet}(\CC^2, L[1]) \oplus \Omega^{2,\bullet}(\CC^2,L^\vee)\]
where $L$ is now a general graded Lie algebra.  
The action functional of the theory is $\int B \wedge F_A$ where $F_A = \dbar A + \frac12 [A \wedge A]$.
If $L$ is equipped with a degree 0 invariant pairing (for instance if $L=\gg$ is an ordinary reductive Lie algebra), 
then the standard Calabi--Yau form on $\CC^2$ identifies the space of fields with $\Omega^{0,\bullet}(\CC^2, T[1](L[1]))$.

\begin{remark}
We can recover $\mc E_{0,0,0}$ by letting $L = \gg[\eps_1,\eps_2]$, where $\eps_i$ are parameters of degree 1. 
These parameters are distinct from the odd parameter we called $\eps$ in Theorem \ref{twist_theorem} and correspond to the holomorphic one-forms $\d z_i$.
That is, in the notation of Theorem \ref{twist_theorem}, we have the following relationship among the fields.
Expand the field $A \in \Omega^{0,\bullet}(\CC^2, \gg[\eps_1,\eps_2][1])$ in BF theory as
\[
A = A_0 + \eps_i A_{i} + \eps_1 \eps_2 A_{12} .
\]
Then, the corresponding field in $\mc E_{0,0,0}$ is
\[
\alpha = A_0 + \d z_i A_i + \d z_1 \d z_2 A_{12}  .
\]
A similar relationship holds between the fields $\beta$ and $B$.

In fact, one can check that holomorphic twists of $\mc N=1$ and $\mc N=2$ super Yang--Mills theories in dimension 4 are also theories of this type (for $L = \gg$ and $L = \gg[\eps]$ respectively), and so the result below will also reply to twisted theories of this more general class.
~\hfill$\Diamond$\end{remark}

For any graded Lie algebra $L$ the holomorphic gauge of Section \ref{sec:holgauge} determines an effective collection of interactions $\{I_L[\Lambda]\}$ for the holomorphic BF theory. 
Since BF theory is of cotangent type, for each $\Lambda$ the coefficient of $\hbar^n$, for $n \geq 2$, in the functional $I_L [\Lambda]$ is zero.
So it admits a decomposition as
\[
I_L [\Lambda] = I_L^{(0)} [\Lambda] +  \hbar I_L^{(1)} [\Lambda]  .
\]
Here $I_{L}^{(0)}$ is a sum over trees $\sum_{\text{trees} \, \downY} I_{\downY}[\Lambda]$ and $I_L^{(1)} [\Lambda]$ is a sum over graphs with a single loop.

\begin{prop} \label{holo_anomaly_prop}
Let $\delta$ be an odd parameter of degree $\pm 1$
and suppose that $L = L'[\delta]$ for some graded Lie algebra $L'$. 
The effective collection $\{I_L[\Lambda]\}$ satisfies the quantum master equation. 
\end{prop}

The assumption implies that the classifying space $BL$, the formal moduli problem associated to the $L_\infty$ algebra $L$, can be identified with the formal moduli problem~$T[\pm 1] B L'$.

\begin{proof}
The order $\hbar^0$ component of the quantum master equation is equivalent to the classical master equation, so we focus our attention to the order $\hbar^1$ component.
We will show that the obstruction to solving the one-loop quantum master equation 
\begin{equation}\label{eqn:qme1}
\Theta_L[\Lambda] = \hbar \Delta_\Lambda I^{(0)}_L [\Lambda] + \left(Q I^{(1)}_L [\Lambda] + \frac12 \{I^{(0)}_L [\Lambda], I^{(1)}_L [\Lambda]\} \right)
\end{equation}
satisfies $\lim_{\Lambda \to 0} \Theta[L] = 0$. 

We apply \cite[Theorem C.5]{LiLi} to show that the obstruction vanishes. 
Since $\Delta_\eps I = 0$ for any $\eps > 0$, 
it follows that the $\Lambda \to 0$ limit of the obstruction is
\[
\lim_{\Lambda \to 0} \Theta_L [\Lambda] = \lim_{\Lambda \to 0}  \lim_{\eps \to 0} \sum_{\text{wheel} \, \wheel} \sum_{\text{edge} \, e} W_{\wheel, e}(P_{\eps<\Lambda}, K_{\eps} , I) .
\]

We analyze the weight $W_{\wheel, e}(P_{\eps<\Lambda}, K_{\eps} , I)$ where $e$ is a fixed edge in a wheel $\wheel$ with $v$ vertices.
The weight is product of an analytic and an algebraic piece
\begin{equation}\label{eqn:algweight}
W_{\wheel, e}(P_{\eps<\Lambda}, K_{\eps} , I) = w^{\mr{an}}_{v, e}(p_{\eps<\Lambda}, k_{\eps} , I) \; w^{\mr{alg}}_{v} (k^{\mr{alg}}, I^{\mr{alg}}) .
\end{equation}
Notice that the algebraic weight does not depend on the distinguished edge $e$. 
We will show that $w^{\mr{alg}}_{v} (k^{\mr{alg}}, I^{\mr{alg}})$ is identically zero.

The algebraic factor $k^{\mr{alg}}$ of both the propagator $P_{\eps < \Lambda}$ and the heat kernel $K_\eps$ is given by the element in $ \sym^2(L^*)$ corresponding to the choice of invariant pairing on $L$.
When one contracts the factors of $(k^{alg})^{\otimes v}$ thus computing the algebraic weight of the graph, what results is an invariant polynomial for the Lie algebra~$L$. 

Recall we are in the situation that $L=L'[\delta]$ where $\delta$ is a parameter of degree $\pm 1$. 
If we assign weight $1$ to $\delta$, then it is clear that the classical action is weight $-1$ and the propagator is weight $1$.
Thus, the weight of any wheel representing the anomaly is weight zero. 
In particular, it suffices to show that the invariant polynomial $w^{\mr{alg}}_{v} (k^{\mr{alg}}, I^{\mr{alg}})$ vanishes when restricted to the subalgebra $L' \subset L'[\delta]$. 

Up to scale, the invariant polynomial $w^{\mr{alg}}_{v} (k^{\mr{alg}}, I^{\mr{alg}})$ sends $X \in L'$ to
$\mr{Tr}_{L} (X^v)$,
where we mean the trace in the adjoint representation. 
But, as $L'$-modules, we have $L = L'[\delta] = L' \oplus L'[\pm1]$. 
Thus, for every element $Y \in L'$, the trace is zero $\mr{Tr}_L (Y) = \mr{Tr}_{L'} (Y) - \mr{Tr}_{L'}(Y) = 0$. 
\end{proof}

This computation has the following consequence for our holomorphic theory in the case $L' = \mf g [\eps_1]$ and $\delta = \eps_2$, where $\mf g$ is an ordinary Lie algebra. 

\begin{corollary}
\label{cor:holtheory}
The effective collection $\{I[\Lambda]\} = \{I_{\mf g [\eps_1,\eps_2]} [\Lambda]\}$ of Proposition~\ref{holo_anomaly_prop} is a solution to the quantum master equation for the theory~$\mc E_{0,0,0}$.
\end{corollary}

We note that this fact did not require us to know anything about the analytic form of the factor in the obstruction.  
This result is discussed in \cite[Proof of Proposition 7.7]{SWSuperconformal} and in~\cite[Proposition 6.5]{GRW}. 

\subsubsection{Extending to the Family of Twists} \label{Btwist_section}

Let us now consider the quantization of our three-dimensional family of theories $\mc E$. 
We do it by leveraging the quantization of the purely holomorphic theory of the previous section.  

The effective collection $\{I[\Lambda]\}$ for the holomorphic theory $\mc E_{0,0,0}$ was constructed using the $\CC[\eps]$-linear extension of the $\dbar^*$ operator acting on de Rham forms. 
We observe that this $\dbar^*$ operator commutes with the $\CC[t_1,t_2,u]$-linear operator 
\[
t_1 \dd_{z_2} + t_2 \dd_{z_2} + u \frac{\d}{\d \eps} .
\]
In particular, the effective collection $\{I[L]\}$ constructed in the previous section extends to an effective collection for the $\CC[t_1,t_2,u]$-linear family $\mc E$. 
This implies that the propagator $P_{\eps<\Lambda}$ constructed from $\dbar^*$ is constant over the entire three-parameter family. 
Thus the effective action $I [\Lambda]$, constructed in the same way as in the purely holomorphic theory, is also constant over the entire three-parameter family. 
It remains to see that this effective action satisfies the $\CC[t_1,t_2,u]$-linear quantum master equation.

\begin{prop} \label{B_anomaly_prop}
The effective collection $\{I[\Lambda]\}$ of Corollary \ref{cor:holtheory} is a solution to the quantum master equation for the $\CC[t_1,t_2,u]$-family of theories $\mc E$ on~$\CC^2$. 
\end{prop}

\begin{proof} 
The obstruction to solving the $\CC[t_1,t_2,u]$-linear quantum master equation is very similar to the obstruction to solving the quantum master equation for the purely holomorphic theory.
The obstruction is given by the formula
\begin{equation}\label{eqn:qme2}
\Theta_{u=0} [\Lambda] = \hbar \Delta_\Lambda I^{(0)} [\Lambda] + \left(Q I^{(1)} [\Lambda] + \frac12 \{I^{(0)} [\Lambda], I^{(1)} [\Lambda]\} \right) + \left(t_1 \dd_{z_2} + t_2 \dd_{z_2} + u \dd_{\eps}\right) I^{(1)}[\Lambda]  .
\end{equation}
We will show that $\lim_{\Lambda \to 0} \Theta[L] = 0$. 

The $\Lambda \to 0$ limit of the first two terms in the above equation vanish by the same argument as in the proof of Proposition \ref{holo_anomaly_prop}. 
It remains to show that 
\begin{equation}\label{eqn:familyqme}
\lim_{\Lambda \to 0} \left(t_1 \dd_{z_2} + t_2 \dd_{z_2} + u \dd_\eps \right) I^{(1)}[\Lambda] = 0 .
\end{equation}
This result is deduced from a fact about wheel diagrams, just as in the proof of Proposition~\ref{holo_anomaly_prop}. 

We apply \cite[Theorem C.5]{LiLi}, which gives the following formula for the one-loop effective action:
\[
I^{(1)}[\Lambda] = \sum_{\text{trees} \, \downY} \sum_{v \in \downY} \sum_{\text{wheel} \, \wheel} W_{\downY, v} \left(P_{\eps<\Lambda}, I, W_{\wheel}(P_{\eps<\Lambda}, I) \right) .
\]
We focus our attention just on the piece $W_{\wheel}(P_{\eps<\Lambda}, I)$ where $\wheel$ is a wheel diagram.  It is enough to show that the $\Lambda \to 0$ limit of this term vanishes for all wheel diagrams.

Suppose $\wheel$ is a wheel with $v$ vertices. 
The weight $W_{\wheel}(P_{\eps<\Lambda}, I)$ can be decomposed as a product of an analytic and algebraic factor as
\[
W_{\wheel, e}(P_{\eps<\Lambda}, I) = w^{\mr{an}}_{v}(p_{\eps<\Lambda} , I) \; w^{\mr{alg}}_{v} (k^{\mr{alg}}, I^{\mr{alg}}) .
\]
We crucially observe that the dependence of this weight on the values of $t_1$ and $t_2$ occur entirely in the analytic factor $p_{\eps<\Lambda}$ of the propagator.  Both $k^{\mr{alg}}$ and $I^{\mr{alg}}$ are independent of $t_1$ and $t_2$.  Therefore the factor $w^{\mr{alg}}_{v} (k^{\mr{alg}}, I^{\mr{alg}})$ is the same algebraic weight as in Equation \eqref{eqn:algweight}.  We showed above this algebraic weight vanishes. It follows that Equation~\eqref{eqn:familyqme} holds.
\end{proof}

\subsubsection{A-Type Twists} \label{A_subsection}

In this section we provide a more extended discussion that is specific to the subfamily of twists with $u \ne 0$---the family of $A$-type twists as defined in Remark~\ref{A_B_name_remark}. 

\begin{remark}
Note that the classical BV complexes of A-type twist are all contractible.  In other words these theories are uninteresting when viewed as perturbative field theories in isolation (as opposed to as part of the full three-dimensional family of theories).  These A-type theories, however, do have interesting structure when one considers aspects other than the perturbative classical field theory, such as their global moduli spaces of solutions (as discussed in~\cite{EY1, EY3}).  
~\hfill$\Diamond$\end{remark}

Over this locus the family of classical theories is no longer of cotangent type, so cannot be presented in the form $T^*[-1] \mc M$. 
In particular, there is no {\em a priori} guarantee of the existence of a one-loop exact collection of effective actions.  There is a potential contribution coming from diagrams with more than one loop.

However, we have shown in Proposition~\ref{B_anomaly_prop} that the effective collection $\{I[\Lambda]\}$ produced via the holomorphic gauge fixing operator does produce a one-loop exact quantization over the entire family, in particular we obtain a quantization over the locus of $A$-type twists.  In other words, Feynman diagrams with more than one loop do not contribute to the quantization even when $u \ne 0$.  

While this result may appear surprising at first, to see why it is true we need only note that both the classical interaction $I$ \emph{and} the propagator $P_{\eps < \Lambda}$ are independent of the parameters $t_1, t_2$ and $u$.  Indeed, the heat kernel, and hence the propagator, are defined using the commutator $[\ol \dd^*, \ol \dd + t_1 \dd_{z_1} + t_2 \dd_{z_2} + u \frac \dd{\dd \eps}]$, which is independent of the three parameters.  Therefore the weight $W_\Gamma(P_{\eps < \Lambda}, I)$ is constant across the family.  In particular it vanishes if $\Gamma$ has more than one loop.

An alternative argument, which applies to deformations of BF theories more generally, may offer insight as well.  Let us view an A-twist as deforming the \emph{interaction} by a quadratic term $I_u^A(\alpha, \beta)$.  
We can argue directly that such deformations are trivial.

\begin{lemma} \label{A_diagram_lemma}
Fix a holomorphic BF theory on $\CC^n$, and let $I^A$ denote a quadratic deformation of the classical interaction involving no spacetime derivatives.  The weight of all Feynman diagrams with at least one loop involving the bivalent vertex associated to $I^A$ vanishes identically.
\end{lemma}

Our result then follows by applying this lemma to the theory $\mc E_{u=0}$ on $\CC^2$ with $I^A = I^A_u$ for some non-zero choice of~$u$.

\begin{proof}
Consider any Feynman diagram $\Gamma$ including the bivalent vertex associated to $I^A$, say with $k+1$ vertices.  Let us say that this bivalent vertex is located at $w \in \CC^n$, and the remaining vertices are located at $z_1, \ldots, z_k$.  The weight of this Feynman diagram is given by a limit of expressions of the form
\[W_\Gamma^{\eps < \Lambda} = \int_{[\eps, \Lambda]^{k+1}}\int_{\CC^{kn}_{z_1, \ldots, z_k}}\left(\int_{\CC^n_w} \ol \dd_w^* K_\eps(z_1 - w) \wedge \ol \dd_w^* K_\eps(w - z_2) \d w_1 \cdots \d w_n\right)F(z_1, \ldots, z_k) \d t_1, \ldots, d t_{k+1}, \]
where $F(z_1, \ldots, z_k)$ is some $t_i$ dependent differential form that is, crucially, independent of the location $w$ of the A-type bivalent vertex.  It is enough to note that the factor $\int_{\CC^n_w} \ol \dd_w^* K_\eps(z_1 - w) \ol \dd_w^* K_\eps(w - z_2) \d w_1 \cdots \d w_n$ vanishes for all values of $z_1, z_2$.  Indeed, if $A$ and $B$ are Dolbeault forms we can always calculate
\begin{align*}
\int_{\CC^n} \ol \dd^* A \wedge \ol \dd^* B \d w_1 \cdots \d w_n &= \int_{\CC^n} \ol \dd (\ast A) \wedge \ol \dd (\ast B) \d w_1 \cdots \d w_n \\
&=  \int_{\CC^n} \d (\ast A) \wedge \d (\ast B \wedge \d w_1 \cdots \d w_n) \\
&= 0
\end{align*}
using integration by parts.  Therefore the weight of the Feynman diagram $\Gamma$ vanishes identically.

The only remaining possibility is that all the A-type bivalent vertices are only connected to one propagator, i.e. that they directly connect to external edges.  Any such diagram must have at most one loop, and therefore has vanishing weight by~\cite[Proposition 4.4]{BWhol}.
\end{proof}

\section{Perturbation Theory Around a Vacuum}
\label{modvac}

So far we have discussed these field theories when expanded around a special solution to the equations of motion: the trivial solution {\em aka} the zero section.
At the classical level, the local $L_\infty$ algebras we have introduced describe the formal neighborhood of this solution,
and we have produced perturbative quantizations around this solution.
There are, however, other important solutions to consider -- namely, translation-invariant solutions --
and here we will show that our results extend easily to this class of solutions.  We will also discuss phenomena like symmetry breaking for these twisted theories.

The main outcome of this section is the following construction.

\begin{theorem} \label{vacua_family_of_theories_thm}
The family of quantum field theories over $\CC^3$ constructed in the previous section extend to a family of quantum field theories over the stack $\CC^3 \times \left[\gg^*/G\right]$, where $\left[\gg^*/G\right]$ is the coadjoint quotient stack. 
\end{theorem}

Here the incorporation of the $\gg^*$ factor corresponds to computing the perturbation theory around a non-zero vacuum solution for any given twisted theory.  The subfamily over $\CC^3 \times \{0\}$ recovers the family of quantum field theories constructed in the previous section.

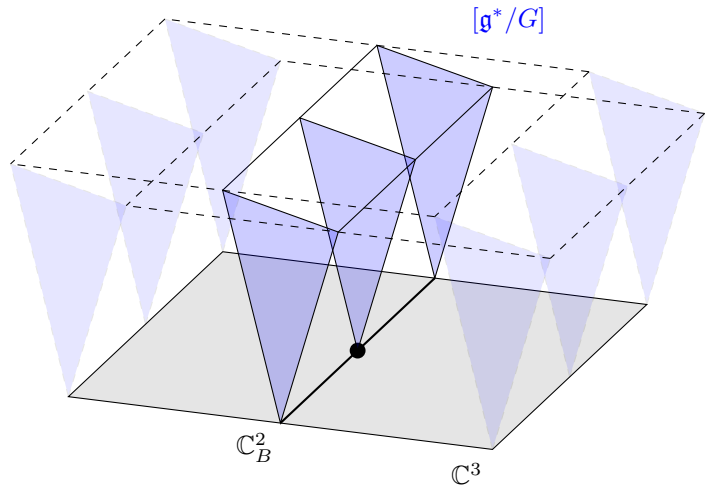
\begin{figure}[h]
 \begin{center}
 \tdplotsetmaincoords{70}{110}
 \begin{tikzpicture}[tdplot_main_coords,fill opacity=0.2]

 \draw[fill=gray] (3,-3,0) -- (3,3,0) -- (-3,3,0) -- (-3,-3,0) -- cycle;
 \draw[thick] (-3,0,0) -- (3,0,0);
 \fill[black,opacity=1] (0,0,0) circle (3pt);
 
 \draw[fill=blue] (-3,0,0) -- (-3.5,-1,3) -- (-2.5,1,3) -- cycle;
 \draw[fill=blue] (0,0,0) -- (-0.5,-1,3) -- (0.5,1,3) -- cycle;
 \draw[fill=blue] (3,0,0) -- (2.5,-1,3) -- (3.5,1,3) -- cycle;
 \draw (-3.5,-1,3) -- (2.5,-1,3);
 \draw (-2.5,1,3) -- (3.5,1,3);
 
 \draw[dashed, fill=blue, opacity=0.1] (-3,-3,0) -- (-3.5,-4,3) -- (-2.5,-2,3) -- cycle;
 \draw[dashed, fill=blue, opacity=0.1] (0,-3,0) -- (-0.5,-4,3) -- (0.5,-2,3) -- cycle;
 \draw[dashed, fill=blue, opacity=0.1] (3,-3,0) -- (2.5,-4,3) -- (3.5,-2,3) -- cycle;
 \draw[dashed] (-3.5,-4,3) -- (2.5,-4,3);
 \draw[dashed] (-2.5,-2,3) -- (3.5,-2,3);
 
 \draw[dashed, fill=blue, opacity=0.1] (-3,3,0) -- (-3.5,2,3) -- (-2.5,4,3) -- cycle;
 \draw[dashed, fill=blue, opacity=0.1] (0,3,0) -- (-0.5,2,3) -- (0.5,4,3) -- cycle;
 \draw[dashed, fill=blue, opacity=0.1] (3,3,0) -- (2.5,2,3) -- (3.5,4,3) -- cycle;
 \draw[dashed] (-3.5,2,3) -- (2.5,2,3);
 \draw[dashed] (-2.5,4,3) -- (3.5,4,3);
 
 \draw[dashed] (-3.5,2,3) -- (-3.5,-4,3);
 \draw[dashed] (2.5,2,3) -- (2.5,-4,3);
 \draw[dashed] (-2.5,4,3) -- (-2.5,-2,3);
\draw[dashed] (3.5,4,3) -- (3.5,-2,3);

\node[opacity=1] (line) at (4,0,0) {$\mathbb C^2_B$};
\node[opacity=1] (line) at (4,3,0) {$\mathbb C^3$};
\node[color=blue, opacity=1] (line) at (-4.5,0.5,3.2) {$[\gg^*/G]$};
\end{tikzpicture} 
\end{center}
\caption{An illustration of the space $\CC^3 \times [\gg^*/G]$ parameterising the choice of a pair $(Q,[x])$, where $Q$ is a square-zero supercharge and $[x]$ is a choice of vacuum.}
\end{figure}

\begin{remark}
Let us clarify what is meant by a family over $[\gg^*/G]$.
We simply mean a family of theories over the affine space $\gg^*$ equipped with an action of $G$ that is equivariant for the map down to the base. This equivariant family of theories over $\gg^*$ thus determines a family of theories over the quotient stack.
~\hfill$\Diamond$\end{remark}

Below we will first discuss what we mean by vacua, describe the associated family of classical theories, and then describe the quantization.

\subsection{Moduli of Vacua}

\subsubsection{First Steps Towards Vacua} \label{B_vacua_section}

In this section we will review the idea of the moduli space of vacua for a classical field theory, and what it means to consider not just a single perturbative field theory, but a family of perturbative field theories parameterized by a choice of vacuum.  This section will provide motivation for the constructions we describe in the subsequent sections, where we demonstrate that the twisted field theories studied above can be generalized to families of perturbative field theories depending on an auxiliary space, the \emph{coadjoint quotient stack} $[\gg^*/G]$.

Let us begin by recalling the idea of a vacuum solution.
\begin{definition}
The (classical) \emph{moduli of vacua} of a classical field theory on $\RR^n$ is the derived stack of translation invariant solutions to the equations of motion.
\end{definition}

\begin{remark}
Note that in order to study such moduli spaces rigorously, we must work with the nonperturbative classical field theory, rather than studying the perturbative expansion around a fixed solution.
A completely rigorous discussion would require us to formulate a derived moduli space of solutions, amongst which we would pick out the translation-invariant sublocus.  
As yet, no adequate treatment of field theory using global derived differential super-geometry is yet available to do this.
~\hfill$\Diamond$\end{remark}

\begin{remark}
In Remark~\ref{vacua of untwisted theory} below, 
we turn to the vacua of the untwisted theory,
in light of our discussion of the vacua for the twisted theories.
~\hfill$\Diamond$\end{remark}

Let us begin by studying the moduli of vacua for the B-twisted theory, that is, the twist occuring at the point $(1,1,0) \in \CC^3$.  There is a natural choice of derived stack that enhances the perturbative classical field theory $\mc T_{1,1,0}$ of Section~\ref{summary_twist_section}.

\begin{claim} \label{nonperturbative_B_claim}
The moduli stack of solutions to the equations of motion for the B-twisted theory on a 4-manifold $X$ is $\mc M_B(X) = \mr{Map}(X_{\mr{dR}}, T^*[3]BG)$.  For instance if $X$ is compact, $\mc M_B(X) =T^*[-1]\mr{Flat}_G(X)$, where $\mr{Flat}_G(X)$ is the moduli stack of flat $G$-bundles on~$X$.

If $X = \RR^4$, then we can similarly identify
\begin{align*}
 \mc M_B(\RR^4) &= T^*[3]\mr{Flat}_G(\RR^4) \\
 &\iso T^*[3]BG \\ 
 &\iso \left[\gg^*[2]/G\right],
\end{align*}
the ($-2$-shifted) \emph{coadjoint quotient stack}.  The group $\RR^4$ of translations acts trivially on this stack, so the moduli of vacua is also equivalent to~$\left[\gg^*[2]/G\right]$.
\end{claim}

\begin{remark}
Although we are not rigorously justifying these nonperturbative descriptions of our twisted theories, we do refer to a related discussion in \cite{EY1}, where descriptions of this type are derived from the starting point of a non-perturbative description for holomorphic Chern--Simons theory on $\mc N=4$ super twistor space.
~\hfill$\Diamond$\end{remark}

Let us now begin to work in a slightly different context: for the rest of this section we will work with a $\ZZ/2\ZZ$-grading, not a $\ZZ$-grading. 
That is, we will forget the $\ZZ$-grading on our BV complexes down to a $\ZZ/2\ZZ$-grading.  

\begin{remark}
After reducing the grading, we gain the freedom to study perturbation theory around non-zero points in the moduli of vacua, which would otherwise have corresponded to inhomogeneous (but purely even) ideals in its ring of functions.  At least in the example of the B-twisted theory it is possible to recover a $\ZZ$-graded story at the expense of introducing some more technical modifications, which we will not be concerned with in this paper.  See \cite[Section 3.4]{EY2} for a discussion of this issue.
~\hfill$\Diamond$\end{remark}

The discussion of Claim~\ref{nonperturbative_B_claim} motivates the following definition.

\begin{definition}
For the B-twisted theory with gauge group given by a reductive algebraic group $G$,
the (stacky) {\em moduli of vacua} is the coadjoint quotient stack $[\gg^*/G]$.
Its {\em coarse moduli of vacua} is the quotient space $\hh^*/W$,
arising from the action of the Weyl group~$W$ on the dual space $\hh^*$ to the Cartan subalgebra $\hh$ of~$\gg$.
\end{definition}

It is worthwhile to note how these two moduli differ.
The stack is a rather sophisticated approach that keeps track of a lot of information,
such as the isotropy group (or stabilizer) of a point $x \in \gg^*$.
The coarse space shows up when one constructs a quotient via the ring of functions,
because 
the $G$-invariant functions on $\gg^*$ are isomorphic to the $W$-invariant functions on $\hh^*$.
In the setting of algebraic geometry, this result is known as the Chevalley restriction theorem~\cite{ChevalleyICM}; one says that the affine variety $\hh^*/W$ is the affinization of the algebraic stack~$[\gg^*/G]$. 
There is a canonical map $[\gg^*/G] \to \hh^*/W$,
sending a conjugacy class to its ``generalized eigenvalues'' (or better, its characteristic polynomial).

This simpler quotient space forcibly identifies distinct solutions and hence is merely a coarse approximation.
To see this, consider the case $G = \GL_n(\CC)$, and let us use a nondegenerate invariant pairing on $\gg$ to identify $\gg^*$ and $\gg$.
The quotient $\hh/W$ arises by looking at diagonal matrices up to permutation of eigenvalues (e.g., by taking the characteristic polynomial or by taking traces of powers of the matrix).
But not every matrix in $\gg$ is conjugate to a diagonal matrix,
as we know from the theory of Jordan normal form.
Hence the map $[\gg^*/G] \to \hh^*/W$ remembers the generalized eigenvalues of a matrix (up to permutation) but does not remember the Jordan blocks.

\begin{remark}
Physicists often describe constructions involving coarse moduli of vacua such as $\hh^*/W$, but mention where various points have ``enhanced symmetry'' (see e.g. \cite{SeibergVacua,LutyTaylor}). 
The stacky quotient captures this enhanced symmetry data in a mathematically precise way.
~\hfill$\Diamond$\end{remark}

\subsubsection{Vacua for Other Twists} \label{holo_vacua_section}

Now, what happens when we try to generalize this discussion from the B-twist to the other twisted theories?  Let us next consider the holomorphic twist.  We can make a natural (and again, not rigorously justified) guess for a non-perturbative description for the holomorphic twist, much like we did for the B-twist.

\begin{claim} \label{nonperturbative_holo_claim}
The moduli stack of solutions to the equations of motion for the holomorphically twisted theory on a complex surface $X$ is $\mc M_{\mr{hol}}(X) = \mr{Map}(X_{\mr{Dol}}, T^*[3]BG)$.  In particular, if $X$ is compact, $\mc M_{\mr{hol}}(X) = T^*[-1]\higgs_G(X)$, where $\mr{Higgs}_G(X)$ is the moduli stack of $G$-Higgs bundles on~$X$.

If $X = \CC^2$, then we can similarly identify
\[\mc M_{\mr{hol}}(\CC^2) = T^*[3]\higgs_G(\CC^2).\]
\end{claim}

The group $\RR^4$ of translations now acts non-trivially on $\mc M_{\mr{hol}}(\CC^2)$, 
but we can easily pick out a substack of translation invariant solutions, namely the stack $\left[\gg^*[2]/G\right] = T^*[3]BG$, 
where we consider only the trivial Higgs bundle with zero Higgs field, modulo translation invariant gauge transformations.  
This space is \emph{not} the full moduli of vacua; 
for instance, we could consider the shifted cotangent space to all constant Higgs fields, not necessarily zero-valued.  
However, only this substack of zero Higgs fields extends to vacua for the B-type twists.
We will thus focus on those vacua that extend to the whole space of twists, in a certain sense to be explained.

In fact, we will exhibit a family of theories where we vary \emph{two} parameters:
\begin{enumerate}
\item first, the choice of twisting parameter $Q \in \CC^3$, and
\item second, a point $x \in \left[\gg^*/G\right]$.
\end{enumerate}
We have seen why $\left[\gg^*/G\right]$ parametrizes vacua for the B-type twists,
but the situation for the A-type twists is more subtle.
We will not construct the moduli of vacua for the A-type twists, 
for reasons discussed in the next Remark,
but we will produce an A-type theory for each $x \in \gg^*$ below.

\begin{remark}\label{A_vacua_remark}
The moduli of vacua for these theories have the potential to include new and interesting additional features.  The moduli stack we have been considering for the holomorphic twist can alternatively be described using Simpson's notion of the Dolbeault stack: we have a shifted adjoint quotient stack
\[
BG_\mr{Dol} = T[1]BG = [\gg[2]/G],
\]
where $BG = [\pt/G]$.
Hence, in the $\ZZ/2\ZZ$-graded setting, 
the adjoint quotient stack is isomorphic to $\Pi T(BG)$ ---
the parity-shifted tangent bundle to $BG$ --- 
and hence it is a Dolbeault stack as
\[
[\gg/G] = BG_\mr{Dol}.
\]
This identification is convenient for us, but 
moving to the $\ZZ/2\ZZ$-graded world drastically changes the geometry:
as a $\ZZ/2\ZZ$-graded space, there are many points in $BG_\mr{Dol}$ (i.e. maps from $\pt = \spec(\CC)$),
while in $\ZZ$-graded derived geometry, there is only one point in~$BG_\mr{Dol}$.

Passing to the A-twist should correspond {\em in the $\ZZ/2\ZZ$-graded setting} to Simpson's Hodge deformation, moving from the Dolbeault to the de Rham stack of $BG$.  It would be interesting to perform a careful analysis of this construction, but we will not do this in the present paper.
~\hfill$\Diamond$\end{remark}

\begin{remark}\label{vacua of untwisted theory}
Finally, we turn to vacua for the original, untwisted theory. 
{\it A priori} one might expect the moduli of vacua here to be even larger than that of the holomorphic twist.  
Such vacua are discussed by Vafa and Witten in \cite[Section 5.1]{VafaWitten}, where the moduli of vacua for the untwisted $\mc N=4$ theory (without turning on a mass deformation) is described as
\[\left[\{X,Y,Z \in \gg^3 \colon [X,Y]=[Y,Z]=[Z,X]=0\}/G\right],\]
the adjoint quotient of the triply commuting variety of $\gg$.  There is a natural embedding of $\left[\gg^*/G\right]$ into this quotient stack using the diagonal map $\gg \to \gg^3$, and identifying $\gg^*$ and $\gg$ using a non-degenerate invariant pairing.
~\hfill$\Diamond$\end{remark}

\subsection{Families of Classical Field Theories over \texorpdfstring{$[\gg^*/G]$}{[g*/G]}}

We turn to an important feature of the moduli of vacua:
it parametrizes a family of perturbative classical field theories, by remembering the formal neighborhood of the vacua inside the space of all solutions.
We will describe a sheaf of local $L_\infty$ algebras over $[\gg^*/G]$ that encodes this information.

\subsubsection{The holomorphic twist and the Higgs mechanism}

For simplicity of notation, we continue to use a nondegenerate invariant pairing on $\gg$ to identify $\gg^*$ with $\gg$,
and we will freely pass between them below.
Concretely, if we pick an element $x \in \gg$, it determines a translation-invariant solution $(\dbar,x)$ to the equations of motion for holomorphic BF theory.
There is a local $L_\infty$ algebra on $\CC^2$ given by 
\[
(\Omega^{\bu,\bu}(\CC^2) \otimes \gg[\eps][1], \dbar + \eps [x, -])
\]
where $\gg[\eps]$ is the tensor product of $\gg$ with $\CC[\eps]$ with $|\eps| = -1$.
The Maurer-Cartan equation encodes the perturbation expansion of holomorphic BF theory around this solution, as can be quickly checked.
This construction clearly varies nicely in the choice of  $x \in \gg$, 
so that we get a bundle of local $L_\infty$ algebras (equivalently, perturbative classical field theories), as only the $[-]_1$ bracket varies.
It is straightforward to verify that this bundle is equivariant under conjugation,
so that this bundle on $\gg^*$ descends to a bundle on the quotient~$[\gg^*/G]$.

\begin{lemma}
For each complex reductive group $G$,
this construction produces a family $\mc{T}_{0,0,0}^{\mr{vac}}[G]$ over the moduli of vacua $[\gg^*/G]$ of ($\ZZ/2\ZZ$-graded) perturbative classical field theories on~$\CC^2$.
\end{lemma} 

We assert that this family should encode the formal neighborhood of the translation invariant subspace $[\gg^*/G]$ sitting inside the moduli space of solutions.

This family has some lovely and intriguing features, 
called the {\em symmetry breaking phenomenon} or Higgs mechanism by physicists.
The idea is quite simple: there is an equivalence of the perturbative theory around a translation-invariant solution $x \in \gg$ with a different perturbative theory.
In our case, for each point $x$, there is an isotropy (or stabilizer) group $G_x \subset G$ fixing $x$ under conjugation,
and the perturbative expansion of the $G$-BF theory around $x$ agrees with the perturbative expansion for the $G_x$-BF theory around $0$.  Let us say this more precisely.

\begin{lemma}\label{symmetry_breaking_lemma}
There is an equivalence of perturbative theories
\[\mc T_{0,0,0}^{\mr{vac}}[G]\big|_x \iso \mc T_{0,0,0}^{\mr{vac}}[G_x]\big|_0\]
between the perturbative theory over a vacuum $x$ for gauge group $G$, and the perturbative theory over the vacuum 0 for gauge group given by the centralizer~$G_x$.
\end{lemma}

\begin{proof}
Consider the inclusion of complexes
\[
(\Omega^{\bu,\bu}(\CC^2) \otimes \gg_x[\eps][1], \dbar) \hookrightarrow (\Omega^{\bu,\bu}(\CC^2) \otimes \gg[\eps][1], \dbar + \eps [x, -]),
\]
which provides, in fact, a quasi-isomorphism of local $L_\infty$ algebras.  Indeed, it is enough to show that the complex $\Pi \gg \xto{\ad_x} \gg$ is canonically isomorphic to the direct sum of the acyclic complex 
$\Pi \gg_x^\perp \xto{\ad_x} \im(\ad_x)$
and the complex $\Pi \gg_x \xto{0} \gg_x$.

The inclusion of $\Pi \gg_x^\perp$ into $\Pi \gg_x$ can be postcomposed with $\ad_x$,
and the image is clearly the whole image of $\ad_x$. 
Hence the inclusion of the acyclic complex into $\TT_{[x]} [\gg/G]$ (or rather, our representing complex) is manifest.

We also know that $\Pi \gg = \Pi \gg_x \oplus \Pi \gg_x^\perp$, 
so it remains only to give a canonical isomorphism between the even component $\gg$ and the direct sum $\im(\ad_x) \oplus \gg_x$.
In other words, we need to give a canonical isomorphism between $\im(\ad_x)^\perp$ and $\gg_x$.
The key fact is that for any elements $a,b \in \gg$,
\[
\kappa(a, [x,b]) = \kappa([a,x], b) = - \kappa([x,a], b);
\]
in other words, $\ad_x$ is skew self-adjoint for the pairing $\kappa$.
Thanks to this fact, we see that 
\begin{align*}
y \in \im(\ad_x)^\perp &\iff \kappa(y, \ad_x z) = 0 \, \forall z \in \gg\\
&\iff  \kappa(\ad_x y, z) = 0 \, \forall z \in \gg\\
&\iff  \ad_x y =0.
\end{align*}
That is, $\gg_x = \im(\ad_x)^\perp$.
\end{proof}

In the sense of Lemma \ref{symmetry_breaking_lemma}, the symmetry group (i.e., the gauge group) has ``broken'' from $G$ to $G_x$.  
For instance, for $G = \GL_n(\CC)$ and $x$ a diagonal matrix with distinct eigenvalues, the stabilizer is the torus $(\CC^\times)^n$, 
so that we have broken from a nonabelian gauge theory down to an abelian gauge theory.  
(For further discussion of symmetry breaking in this style, see~\cite{ElliottGwilliam}.)

\begin{remark}
The symmetry breaking lemma implies that in fact, the theory over \emph{any} vacuum $[x] \in [\gg^*/G]$ can be lifted from a $\ZZ/2\ZZ$-graded theory to a $\ZZ$-graded theory.  These lifted theories will not vary smoothly in the $[\gg^*/G]$-family, they only exist when considering each point (or each stratum, where we stratify $[\gg^*/G]$ by conjugacy classes of the centralizers.
~\hfill$\Diamond$\end{remark}

\subsubsection{The B twists}

We can immediately generalize this discussion --- of the existence of a family of theories, as well as of symmetry breaking --- from the holomorphic twist to the B-twist by replacing $\dbar$ everywhere with the twisted differential $\dbar + t_1 \partial_{z_1} + t_2 \partial_{z_2}$.
In particular, we have the following.

\begin{lemma}
This construction produces a family $\mc{T}_{t_1,t_2,0}^{\mr{vac}}$ over $[\gg^*/G] \times \CC^2_{t_1,t_2}$ of ($\ZZ/2\ZZ$-graded) perturbative classical field theories on~$\RR^4$.
Moreover, at a point $x \in \gg$, 
there is an equivalence 
\[
\mc{T}_{t_1,t_2,0}^{\mr{vac}}[G]\big|_x \iso \mc{T}_{t_1,t_2,0}^{\mr{vac}}[G_x]\big|_0
\]
of perturbative theories.
\end{lemma}

\subsubsection{The A twists}

We can extend our family of theories depending on a choice of point in the moduli of vacua to the full 3-dimensional family of twisted theories.  We can motivate the following construction by considering the orthogonal decomposition $\gg = \gg_x \oplus \gg_x^\perp$ associated to a chosen invariant nondegenerate pairing.  Recall from our argument for Lemma \ref{symmetry_breaking_lemma} that the operator $\eps \mr{ad}_x$ restricts to an isomorphism from $\gg_x^\perp$ to $\eps \gg_x^\perp$.  The operators $\eps \ad_x$ and $u \frac{\dd}{\dd \eps} \id_\gg$ clearly do not commute if $u$ is non-zero, but if we restrict the latter term to the identity on the centralizer subalgebra $\gg_x$ then the operators $\eps \ad_x$ and $u \frac{\dd}{\dd \eps} \id_{\gg_x}$ do commute.  Therefore we can extend our family of theories to all of $\CC^3 \times [\gg^*/G]$ in the following manner.

\begin{definition}
Let $\mc{T}_{t_1,t_2,u}^{\mr{vac}}$ denote the family over $[\gg^*/G]$ of perturbative classical field theories on $\RR^4$ where the fiber at a point $x \in \gg^*$ is encoded by
\[
(\Omega^{\bu,\bu}(\CC^2) \otimes \gg[\eps][1], \dbar + t_1 \partial_{z_1} + t_2 \partial_{z_2} + \eps \ad_x + u \frac \dd{\dd \eps} \id_{\gg_x}).
\]
\end{definition}

It is straightforward to see that at a point $x \in \gg$, 
there is an equivalence 
\[
\mc{T}_{t_1,t_2,u}^{\mr{vac}}[G]\big|_x \iso \mc{T}_{t_1,t_2,u}^{\mr{vac}}[G_x]\big|_0
\]
of perturbative theories.

As we discussed in Remark~\ref{A_vacua_remark}, 
it would be interesting to realize a family of theories of this type in terms of perturbation theory over the image of a map $[\gg^*/G] \to \mc M_A$, 
where $\mc M_A$ is the moduli stack of solutions to the equations of motion in our A-twisted theory, viewed in $\ZZ/2\ZZ$-graded derived geometry.  
We do not have such a derivation of this family of theories at present.

\subsection{Families of Quantum Field Theories over \texorpdfstring{$[\gg^*/G]$}{[g*/G]}}
The classical moduli of vacua parametrizes a family of perturbative classical field theories,
so it is natural to ask whether we can quantize this family.
We will see that we can, using precisely the techniques that let us quantize the theory at the origin of the coadjoint quotient $[\gg^*/G]$.
(The origin is given by the ``trivial'' solution~$(\dbar, 0)$.)

\begin{remark}
As we will see in Section \ref{En_section}, this quantization gives rise to a family of $\bb E_4$-stacks quantizing the odd symplectic structure on the coadjoint quotient stack $[\gg^*/G] \iso \Pi T^*BG$.
~\hfill$\Diamond$\end{remark}

The key observation is that at every point $((t_1,t_2,u), [x]) \in \CC^3 \times [\gg^*/G]$,
the perturbative classical theory is a deformation of the theory at $(0, [0])$ where only the first order bracket is changed: we add to $\dbar$ the operator
\[
t_1 \partial_{z_1} + t_2 \partial_{z_2} + \eps [x, -] + u \frac{\dd}{\dd \eps} \id_{\gg_x}.
\]
(Equivalently, one can say that we have only changed the quadratic term of the action functional.)
Hence the perturbative quantization exhibits very similar behavior to what we saw when $x = 0$; the only change is that the bivalent vertices are slightly altered.

\begin{theorem} \label{quantum_vacua_family_thm}
The family $\mc T_{t_1,t_2,u}^\mr{vac}$ of $\ZZ/2\ZZ$-graded classical field theories over $\CC^3 \times \left[\gg^*/G\right]$ admits a one-loop exact quantization.
\end{theorem}

\begin{proof}
Let us begin by verifying that there is no anomaly for the quantization of our family of theories over $\CC^3 \times \gg^*$.  
This follows from a similar argument as in Proposition \ref{B_anomaly_prop}.
The gauge fixing operator $\dbar^*$ commutes with the deformation $\eps [x,-]$ for any $x \in \gg^*$. 
Therefore, the effective action $I[\Lambda]$ is independent of the value of $x$.
To show that $I[\Lambda]$ solves the quantum master equation as a family of effective actions over $\CC^3 \times \gg^*$ it suffices to show that $[x, I[\Lambda]] = 0$ for all $x \in \gg^*$ and $\Lambda > 0$. 
This follows from an identical argument to the one we used to prove Equation \eqref{eqn:familyqme}.

Now, we must argue that the classical coadjoint action of $G$ extends to the quantum level.
We have already seen that the linear differential 
\[
\dbar + t_1 \partial_{z_1} + t_2 \partial_{z_2} + \eps \ad_x + u \frac \dd{\dd \eps} \id_{\gg_x}
\]
descends to the coadjoint quotient. 
It remains to see that the effective action $I[\Lambda]$ also descends. 
But, since the propagator and the classical action functional are both strictly $G$-invariant, so is the effective action $I[\Lambda]$ for each $\Lambda > 0$, since the effective action is constructed as a function of the classical action and the propagator.
\end{proof}

\section{Framing Anomaly}
\label{rotation_anomaly_section}

\subsection{Homotopically Trivial \texorpdfstring{$G$}{G}-Actions} \label{background_field_section}
In this subsection we will discuss what it means for a Lie group $G$ to act \emph{homotopically trivially} on a quantum field theory.  Our main example will involve theories on $\RR^n$ with a homotopically trivial action of the group $\mr{ISO}(n) \iso \SO(n) \ltimes \RR^n$ of isometries of $\RR^n$.  As we will see, the presence of such an action allows us to realize the factorization algebra of local observables in the theory as a \emph{framed $\bb E_n$ algebra}.

Let us start by discussing the concept of a homotopy trivialization of a $\gg$-action.  
\begin{definition}
Let $\gg$ be a Lie algebra.  
Write $\gg_{\mr{dR}}$ for the dg Lie algebra whose underlying graded vector space is $\gg[1] \oplus \gg$, with differential given by the identity map $\gg[1] \to \gg$ (so that the resulting cochain complex is contractible), and with bracket given by the Lie bracket on $\gg$ and the adjoint action of $\gg$ on $\gg[1]$.
We write elements of $\gg_{\mr{dR}}$ as pairs $(\Tilde{X}, X)$ with $X\in \gg$ in degree zero and $\Tilde{X} \in \gg[1]$ in degree $-1$. 
\end{definition}

Because $\gg_{\mr{dR}}$ is contractible, actions of such dg Lie algebras on classical and quantum field theories are not, in and of themselves, very interesting. 
There \emph{is} however a non-trivial notion that we can consider.

\begin{definition}
Suppose the Lie algebra $\gg$ is a symmetry of a field theory on a manifold $M$, meaning $\gg$ acts by derivations on the observables $\mr{Obs}(M)$.
A \emph{homotopy trivialization} of the $\gg$ action is a $\gg_{\mr{dR}}$-action on $\mr{Obs}(M)$ that restricts to the original $\gg$-action along the inclusion $\gg \inj \gg_{\mr{dR}}$.
\end{definition}

\subsection{Framing Anomaly Vanishing} 

We will now apply these general ideas to the specific theories we have been studying in this paper: twists of $\mc N=4$ supersymmetric Yang--Mills theory on $\RR^4$.  Specifically, let us consider the B-twisted theory $\mc E_{1,1,0}$ from Section \ref{Btwist_section}.  In this section we will prove that this theory carries a de Rham action of $\mr{ISO}(4)$ when the gauge group $G$ is reductive (in fact it will be sufficient for $G$ to be unimodular).

\begin{theorem} \label{B_twist_de_Rham_action_thm}
The quantization of the classical BV theory $\mc E_{1,1,0}$ carries a homotopy trivialization of the action of the group $\mr{ISO}(4)$ of isometries of $\RR^4$.  
\end{theorem}

\begin{remark}
If we consider a more general B-type theory in our family, say $\mc E_{t_1,t_2,0}$, our proof will apply to a subgroup of the full group of isometries.  If $t_1 \ne t_2$ are both non-zero, then the quantum field theory will admit a homotopy trivialized action of the group $(\SO(2) \times \SO(2)) \ltimes \RR^4$.  If $t_1 \ne 0$ but $t_2 = 0$ (the situation of the ``Kapustin twist''), then the quantum field theory will admit a homotopy trivialized action of the group $\SO(2) \ltimes \RR^2$.  In both cases, there is no obstruction to the existence of this action: the relevant anomaly vanishes.
~\hfill$\Diamond$\end{remark}

\begin{remark}
It is worth noting that the action of the group $\SO(4)$ of rotations on our family of theories does not come from the action of rotations on the untwisted $\mc N=4$ gauge theory.  At the level of complexified Lie algebras we can obtain a ``twisted'' $\so(4;\CC)$ action instead by embedding $\SO(4)$ diagonally in the Lie algebra $\so(4;\CC) \times \sl(4;\CC)$, where the second factor is the algebra of R-symmetries, and where $\so(4;\CC) \iso \sl(2;\CC) \oplus \sl(2;\CC)$ embeds into $\sl(4;\CC)$ block diagonally.  This is the Kapustin--Witten twisting homomorphism considered in \cite{KapustinWitten}, and the 2-parameter family of Kapustin--Witten supercharges consist of exactly those supercharges that are stabilized by the image of this Lie algebra embedding.  Since we do not treat the untwisted $\mc N=4$ theory in this paper, we don't need to keep track of this non-trivial origin for the rotation action.
~\hfill$\Diamond$\end{remark}

The proof of Theorem \ref{B_twist_de_Rham_action_thm} will involve the following steps.
\begin{enumerate}
 \item First we will describe the classical action of $\mr{ISO}(4)$, and its homotopy trivialization.
 \item We will argue that there is no obstruction to the quantization of the $\mr{ISO}(4)$ action itself, using the fact that the ingredients of the quantization of the classical field theory are manifestly isometry invariant. 
\item Next we will argue that there is no obstruction to the homotopy trivialization of the $\RR^4$ action by translations.  
We will use the fact that this trivialization commutes with a natural choice of gauge-fixing condition to show that the relevant Feynman diagrams vanish.
\item It remains to address the homotopy trivialization of the action of rotations.  
We will use the decomposition $\so(4) \iso \su(2) \oplus \su(2)$: if we choose a complex structure the two factors act by holomorphic and antiholomorphic rotations respectively.  We can treat each factor in turn, using two different choices of gauge-fixing operator associated to two conjugate complex structures.  
For each choice, the theory with compatible $\su(2)$ background fields will be a holomorphic field theory as in \cite{BWhol}, so we can use the results therein.  
In particular, the only Feynman diagrams that contribute to the anomaly are wheels with three external edges.  We will show in Theorem \ref{Hodge_anomaly_vanishes} that this anomaly vanishes, completing the proof.
\end{enumerate}

With this plan in mind, we first address the homotopy trivialization of the $\mr{ISO}(4)$ action on the classical field theory $\mc E_{1,1,0}$.

\subsubsection{Classical framing}

Recall that the linear space of fields is two copies of the de Rham complex valued in shifted copies of $\gg$:
\[
\mc E_{1,1,0} = \Omega^{\bu} (\RR^4 , \gg[1]) \oplus \Omega^{\bu} (\RR^4 , \eps \gg[2]) .
\]
The group $\mr{ISO}(4)$ acts on the differential forms on $\RR^4$ by natural automorphisms. 
This extends to an action at the level of classical observables.

\begin{prop}
The action of the group $\mr{ISO}(4)$ on the fields $\mc E_{1,1,0}$ extends to an action on the classical observables 
\[
\mr{Obs}^{\mr{cl}}_{1,1,0} = \left(\mc O(\mc E_{1,1,0}) \, , \, \{S_{1,1,0}, -\} \right) .
\]
\end{prop}
\begin{proof} 
It suffices to observe that the action preserves the Lie bracket on differential forms valued in $\gg$ as well as the invariant pairing on $\gg$. 
\end{proof}

The infinitesimal action of the complexified Lie algebra $\mf{iso}(4;\CC) \iso (\sl(2;\CC) \oplus \sl(2;\CC)) \ltimes \CC^4$ acts on $\mc E_{1,1,0}$ by the Lie derivative of complex valued vector fields on $\RR^4$.
In fact, this action is Hamiltonian in the sense that it is given by the BV bracket with a fixed local functional. 
\begin{lemma}
Let $J_X$ be the local functional $J_X = \int \langle \beta \wedge L_X \alpha \rangle$ on $\mc{E}_{1,1,0}$. 
Then $L_X = \{J_X,-\}$. 
\end{lemma}

For the B-twisted theory $\mc E_{1,1,0}$ we can extend this action to include a homotopical trivialization for the infinitesimal action of isometries using Cartan's formula $L_X = [\d, \iota_X]$.   

For any vector field $X$ on $\RR^4$, let $\eta_X$ be the following endomorphism of cohomological degree $-1$ of the linear space of fields
\begin{align*}
\eta_X & \colon \mc E_{1,1,0} \to \mc E_{1,1,0}[-1] \\
\eta_X (\alpha + \eps \beta) & = \iota_X(\alpha) + \eps \iota_X(\beta) 
\end{align*}
where $\iota_X(-)$ is the contraction of a differential form.
\begin{prop}
\label{prop:cdr}
The collection of endomorphisms $\{\eta_X | X \in \mf{iso}(4;\CC)\}$ defines a homotopy trivialization of the classical field theory $\mc E_{1,1,0}$. 
\end{prop}
\begin{proof}
The operator $\eta_X$ extends to a derivation on observables in the natural way. 
The fact that the operator $\eta_X$ defines a homotopy trivialization of the derivation determined by $X$ is Cartan's formula. 
\end{proof}

There is another, soon to be convenient, description of the endomorphism $\eta_X$. 
For any vector field $X$ consider the local functional $\widetilde{J}_X \in \oloc(\mc E_{1,1,0})[-2]$
defined by
\[\widetilde{J}_X = \int \langle \beta \wedge \iota_{\wt X} \alpha\rangle.\]

\begin{lemma} \label{l:J} 
For any vector field $X$ one has $\eta_X = \{\widetilde{J}_X, -\}$ acting on $\mr{Obs}_{1,1,0}(\RR^4)$.
Moreover, 
\[
L_X = \big\{ \{S_{1,1,0}, \widetilde{J}_X\} , - \big\} 
\]
as derivations on $\mr{Obs}^{\mr{cl}}_{1,1,0}(\RR^4)$. 
\end{lemma}
\begin{proof}
The first part is a straightforward calculation. 
The second part follows from the fact that the only term in $S_{1,1,0}$ that does not commute with $\eta_X$ is the kinetic term $\int \beta \d \alpha$.
The equation is then equivalent to the equality of local functionals $J_X = \{S_{1,1,0}, \widetilde{J}_X\}$ which follows from Cartan's formula $[\d, \iota_X] = L_X$.
\end{proof}

\subsubsection{Quantum framing}

We have, therefore, constructed a homotopically trivialized action of the group of isometries on the \emph{classical} observables of the B-twisted 4d $\mc N=4$ theory.  
In the rest of this section we will quantize this classical symmetry thus proving Theorem \ref{B_twist_de_Rham_action_thm}. 

\begin{prop} \label{quantization_of_isometry_prop}
The quantization of the theory $\mc E_{1,1,0}$ constructed in Section \ref{Btwist_section} admits an $\mr{ISO}(4)$ action.
\end{prop}

\begin{proof}
This follows from \cite[Proposition 9.1.1.2]{Book2} and the remark that immediately follows it.  The proposition in question says that if we are given a quantization of a smoothly $G$-equivariant classical field theory with the property that the effective interaction $I[\Phi]$ associated to every parametrix $\Phi$ is $G$-invariant, then the quantum theory is also smoothly $G$-equivariant.

To see that the effective interactions are $\mr{ISO}(4)$-invariant at every scale, it is enough to observe that the classical interaction and the propagator are $\mr{ISO}(4)$ invariant, by definition of the effective interaction $I[\Phi]$.  This is clear for the classical interaction (defined using the Lie bracket for the gauge Lie algebra and the wedge product of differential forms).  We can define a propagator using the gauge fixing operator $\d^*$ on $\mc E_{1,1,0}$.  Since this operator and the BV bracket on $\mc E_{1,1,0}$ are both $\mr{ISO}(4)$, so is the propagator.
\end{proof}

\begin{remark} \label{smooth_action_rmk}
In fact, this argument shows that the factorization algebras of quantum observables in the B-twisted theory admits a smooth action of the group $\mr{ISO}(4)$, as in \cite[Chapter 4.8]{Book1}.  This will be relevant in Section \ref{factorization_section} below, where we will use the existence of a smooth homotopically trivial action of the isometry group to identify this factorization algebra as a framed $\bb E_4$-algebra.
~\hfill$\Diamond$\end{remark}

Infinitesimally, we see that the classical action of the Lie algebra $\mf{iso}(4;\CC)$ also persists to the quantum level.
It remains to quantize the homotopy trivialization of this action. 
Recall that the derivation $\eta_X$ trivialized the classical action of a vector field $X$. 
Unlike the action of $\mf{iso}(4;\CC)$, it is not immediate that the action by the derivation $\eta_X$ is compatible with the quantization. 
In fact, we will see that for generic Lie algebra $\gg$, there is a potential anomaly to quantizing these trivializations. 

We will address the quantization of the homotopy trivialization by making use of the observation in Lemma \ref{l:J} that, classically, the action of a vector field $X$ and its homotopy trivialization are through the BV bracket of local functionals $J_X$, $\widetilde{J}_X$. 
Instead of directly quantizing the endomorphisms $\{\eta_X\}$, we attempt to quantize the collection of local functionals $\{J_X\}$ by using the background field method of equivariant BV quantization recalled in Section \ref{equivariant_section}. 

We study the Hamiltonian action by the dg Lie algebra $\mf{iso}(4;\CC)_{\mr{dR}}$. 
This process begins by applying the RG flow to the local functionals $I,J, \widetilde{J}$, in a similar way that we constructed $I[\Lambda]$ as the RG flow of the classical interaction $I$.

Let $I_{\rm dR} [\Lambda]$ be the naive effective collection for the $\mf{iso}(4;\CC)_{\mr{dR}}$-equivariant theory, as defined in Remark \ref{rmk:eqcollection}. 
We need to show that $I_{\rm dR}[\Lambda]$ solves the equivariant QME \eqref{eqn:eqqme}. 
The non equivariant theory $\mc{E}_{1,1,0}$ is manifestly holomorphic, but notice that $\mf{iso}(4;\CC)$ involves the action of non holomorphic vector fields.
In particular, the techniques to solving the QME using the holomorphic gauge do not apply. 
Nevertheless, we can break the problem up in the following way. 

Consider the subalgebra 
\[
\CC^2 \oplus \mf{sl}(2;\CC) \subset \mf{iso}(4;\CC)
\]
consisting of holomorphic translations and holomorphic volume preserving linear transformations. 
For any vector field $X$ in this subalgebra, the operators $L_X$ and $\eta_X$ are by holomorphic differential operators. 
Thus, we have the following statement.

\begin{lemma}
Consider the theory $\mc E_{1,1,0}$ on $\CC^2$ with its homotopically trivialized $\SU(2)$-action.  
The corresponding classical field theory with background fields valued in 
\[
\big[\CC^2 \oplus \sl(2;\CC)\big]_{\mr{dR}}
\]
is a holomorphic field theory with background fields in the sense of \cite{GRWthf}, and therefore the only possible obstructions to solving the equivariant QME are given by the weights of wheel Feynman diagrams with three external edges. 
\end{lemma}

Let us now prove that the homotopy trivializations we have defined classically for the rotation action on our B-twisted theory extend to the quantum level, i.e. that the anomaly vanishes.  

\begin{theorem} \label{Hodge_anomaly_vanishes}
Suppose the Lie algebra $\gg$ defining $\mc{E}_{1,1,0}$ is reductive. 
There exists a quantization of the homotopy trivialization of the $\sl(2;\CC)$ action on $\mc E_{1,1,0}$. 
\end{theorem}

\begin{proof}

By Theorem \ref{wheel_thm} we must show that the weight of a wheel-shaped Feynman diagram with three external legs vanishes.  
By Proposition \ref{quantization_of_isometry_prop} it suffices to characterize the anomaly only in the case that it depends on background fields valued in $\CC^2 [1] \oplus \sl(2;\CC) [1]$. 
In this case, the vertices of the wheel diagram labeling the anomaly are only labeled by the classical BF interaction $I$ or $\widetilde{J}_X$ for some vector field $X \in \CC^2 \oplus \sl(2;\CC)$. 
 
We can decompose the weight of such a wheel into a sum, depending on the number of legs labelled by BF theory fields $\alpha \in \Omega^{\bullet, \bullet}(\CC^2;\gg)$.

\begin{enumerate} 

 \item Two $\alpha$-legs.  If the remaining vertex is labelled by $\widetilde{J}_X$,  then we will show that the weight of the diagram is zero for form-degree reasons.  Indeed, identify $\wt X$ as a holomorphic vector field $f_1(w) \frac{\dd}{\dd(\d w_1)} + f_2(w) \frac{\dd}{\dd(\d w_2)}$ on $\CC^2$.   The Feynman weight in question is a limit of integrals of the form
  \[\int_{(u,v,w) \in (\CC^2)^3} \alpha_1(u)\alpha_2(v) P(u,v)P(v,w) \left( f_1(w)\frac{\dd}{\dd(\d w_1)} K(w,u) + f_2(w)\frac{\dd}{\dd(\d w_2)} K(w,u)\right).\]
 The factor $P(u,v)P(v,w) \left( f_1(w)\frac{\dd}{\dd(\d w_1)} K(w,u) + f_2(w)\frac{\dd}{\dd(\d w_2)} K(w,u) \right)$ of the integrand is a $(5,4)$ form on $\CC^6$.  However, by definition of the heat kernel, this form lies in the sub-bundle generated by the 1-forms $\d u - \d v, \d v - \d w, \d \ol u - \d \ol v, \d \ol v - \d \ol w$, in other words, 1-forms pulled back from a four-dimensional subspace of $\CC^6$.  Any $(5,4)$ form of this type must necessarily vanish.
 
 \item One $\alpha$-leg.  Let us write the input field $\alpha \in \Omega^{\bullet, \bullet}(\CC^2;\gg)$ as a pure tensor $\alpha_{\mr{alg}} \otimes \alpha_\gg$, where $\alpha_{\mr{alg}}$ is a differential form on $\CC^2$ and $\alpha_\gg$ is an element of $\gg$ (without loss of generality).  Any diagram of this type has weight proportional to $\mr{Tr}_\gg(\alpha_\gg)$, where the trace takes place in the adjoint representation.  Because $\gg$ is reductive, this trace, and therefore the weight of the diagram, vanishes.
\end{enumerate}

Summing all these contributions, we find that the total weight of a wheel-shaped Feynman diagram is zero, and therefore the anomaly vanishes, as desired. 
\end{proof}

\begin{corollary}
The anomaly for the $\mf{iso}(4;\CC)_\mr{dR}$-equivariant quantization of the B-twisted theory $\mc E_{1,1,0}$ vanishes.
\end{corollary}

\begin{proof}
We split $\mf{iso}(4;\CC)_{\mr{dR}}$ up as a sum 
\[
\big[\CC^2 \oplus \sl(2;\CC)\big]_{\mr{dR}} \oplus \big[\CC^2 \oplus \sl(2;\CC)\big]_{\mr{dR}}
\]
The first factor consists of holomorphic vector fields in a fixed complex structure.
The second consists of holomorphic vector fields in the conjugate complex structure. 
This implies that the full action of $\mf{iso}(4;\CC)_{\mr{dR}}$ is anomaly-free as desired.
\end{proof}

Let us conclude this section by extending this result for the B-twist to the full Kapustin--Witten family of topological twists.  That is, to the $\CC^2 \bs \{0\}$ family of supercharges spanned by the B-twisting supercharge $(1,1,0)$ along with the A-twisting supercharge $(0,0,1)$.

\begin{theorem} \label{KW_family_framing_anomaly_thm}
There is an anomaly-free action of $\mf{iso}(4;\CC)_{\mr{dR}}$ on the twisted theory $\mc E_{t,t,u}$ for any $(t,u) \in \CC^2 \bs \{0\}$, extending the action described above for the point $(1,0)$.
\end{theorem}

\begin{proof}
The classical action of $\mc{E}_{t,t,u}$ is 
\[
S_{t,t,u} = \int \beta \wedge \dbar \alpha + t \int \beta \wedge \partial \alpha + u \int \beta \wedge \beta + I .
\]
Even when $u \ne 0$, the equivariant action $I + J + \wt{J}$ still satisfies the CME since $\{\int \beta \wedge \beta , J_X \} = \{\int \beta \wedge \beta , \wt{J}_X\} = 0$ for any vector field $X$. 

Now, we must verify that this action extends to the quantum level.  
Proposition \ref{quantization_of_isometry_prop} continues to apply for our family of theories, therefore it is enough to generalize Theorem \ref{Hodge_anomaly_vanishes} from $\mc E_{0,0,0}$ to the $\CC[u]$-family of theories $\mc E_{0,0,u}$.  In order to do this, we only need to obseve that the proof of Lemma \ref{A_diagram_lemma} continues to apply for the equivariant version of the interaction, so the Feynman weights that contribute to the anomaly of the theory $\mc E_{0,0,u}$ are independent of the value of $u$.
\end{proof}

Finally, we can combine this equivariant structure with the construction of the family of theories over the moduli of vacua, to obtain the following equivariant family.

\begin{corollary} \label{family_of_isometry_actions_cor}
The family of quantum field theories of Theorem \ref{quantum_vacua_family_thm} can be enhanced to define a family of $\mr{iso}(4;\CC)_{\mr{dR}}$-equivariant one-loop exact $\ZZ/2\ZZ$-graded quantum field theories over the stack $(\CC^2 \bs \{0\}) \times [\gg^*/G]$.
\end{corollary}

\section{The Factorization Algebras of Observables} \label{factorization_section}

The collection of observables of a field theory carry an exceedingly rich mathematical structure.
For a perturbative theory, the notion of a factorization algebra captures important aspects of this structure,
notably the local-to-global features.
In brief, for a field theory on a manifold $M$, its factorization algebra $\obs$ of observables encodes the following information.
\begin{itemize}
\item For each open set $U$, such as an open ball, the collection of observables $\obs(U)$ that only depend on the fields' behavior inside $U$.
\item Given observables $O_1, \ldots, O_k$ that depend on disjoint regions $U_1, \ldots, U_k$, the assignment of a product observable $O_1 \cdots O_k$.
\end{itemize}
In particular, the factorization algebra knows the operator product expansion.
For a precise definition, see \cite{Book1} and \cite{Book2},
which also show how to construct and analyze the observables for any BV theory.
Indeed, the central result of those books is that a factorization algebra arises automatically from the BV formalism, simply by keeping track of the support of observables.

The real benefit of this result is that factorization algebras are objects of intrinsic mathematical interest and have already received extensive development, 
for purposes independent of field theory.
As the preeminent example, consider a topological field theory on $\RR^n$, such as Chern-Simons theory on $\RR^3$.
In this case, the observables are locally constant, in the sense that $\obs(U) \simeq \obs(U')$ if $U \subset U'$ are two balls. 
It is then a theorem that the factorization algebra is equivalent to the data of an algebra over the little $n$-disks operad,
also known as an $\bb E_n$ algebra.
Such algebras were introduced by topologists in the 1960s and have received intensive study since,
so that we can seek to deploy their insights in the setting of topological field theory.
Recently much activity has centered around the notion of factorization homology, 
which (when defined) pairs an $n$-manifold $M$ with a $\bb E_n$ algebra $A$ to produce a cochain complex,
much as traditional homology pairs a space with an abelian group \cite{LurieHA, AyalaFrancis}.
This process produces rich invariants of manifolds and $\bb E_n$ algebras (see, e.g., \cite{BBJ1, GunninghamJordanSafronov, Knudsen, HHKWZ} for some recent work).

Our goal in this section is to begin the analysis of the observables of the twists of 4-dimensional $\mc{N} = 4$ super Yang--Mills theories.
In particular, we begin by explaining when the observables of a topological twist are {\em framed} $\bb E_4$ algebras,
which allows one to compute factorization homology on arbitrary {\em oriented} 4-manifolds.
Next, we explain how to capture the process of canonically quantizing and then imposing constraints is captured by a filtration on the factorization algebra of observables.
This filtration answers a puzzle raised in \cite{EY2} about how quantization of the B-twist affects its observables.
Finally, we unwind what factorization homology should mean and compute its values on various classes of manifolds.

\subsection{Factorization Algebras and Framed \texorpdfstring{$\bb E_n$}{En} Algebras} \label{En_section}

Let us begin by discussing the notion of a smooth $G$-action on a factorization algebra.  We will keep this discussion brief, and refer to \cite[Section 2.2]{ElliottSafronov} and \cite[Chapter 4.8]{Book1} for a more detailed discussion.  

Let $G$ be a Lie group with Lie algebra $\gg$, equipped with an action on $\RR^n$ by diffeomorphisms.  Let $\obs$ denote a factorization algebra on $\RR^n$.

\begin{definition}
A \emph{smooth action} of $G$ on the factorization algebra $\obs$ consists of 
\begin{enumerate}
 \item an isomorphism $\alpha_g(U) \colon \obs(U) \to \obs(gU)$ for each $g \in G$ and each open subset $U \sub \RR^n$, and
 \item a dg Lie algebra homomorphism $\gg \to \mr{Der}(\obs)$, where $\mr{Der}(\obs)$ is the dg Lie algebra of derivations of $\obs$ (see \cite[Chapter 4, Definition 8.1.2]{Book1}).
\end{enumerate}
This data must together satisfying the following conditions.  
\begin{enumerate}
 \item For all $U$ and all $g_1, g_2 \in G$, $\alpha_{g_1}(g_2U) \circ \alpha_{g_2}(U) = \alpha_{g_1g_2}(U)$.
 \item For all $g \in G$, the maps $\alpha_g$ commute with the factorization algebra structure on $\obs$.
 \item For all pairwise disjoint open subsets $U_1, \ldots, U_k \sub V \sub \RR^n$, the map 
  \[m \colon \{(g_1, \ldots, g_k) \in G^k \colon g_k(U_k) \text{ are disjoint subsets of } V\} \to \hom\left(\bigotimes_{i=1}^k \obs(U_i), \obs(V)\right),\]
 defined by first acting by $(g_1, \ldots, g_k) \in G^k$ then using the factorization structure, is smooth.
 \item The maps $\alpha_g$ and the infinitesimal $\gg$ action $\rho$ are compatible in the following sense.  For all $X \in \gg$ and $i = 1, \ldots, k$, we have
 \[\dd_{X,i}m_{g_1, \ldots, g_k}(\mc{O}_1, \ldots, \mc{O}_k) \iso m_{g_1, \ldots, g_k}(\mc{O}_1, \ldots, \ldots, \rho(X)\mc{O}_i, \ldots, \mc{O}_k)\]
 where the map $\dd_{X,i}$ is the derivative on $G^k$ with respect to the tangent vector
 \[(0, \ldots, L_{g_i}(X), \ldots, 0) \in T_{g_1, \ldots, g_k}G^k,\]
 where the non-zero element is placed in the $i^{\mr{th}}$ slot.
\end{enumerate}
\end{definition}

\begin{definition}
A \emph{homotopy trivialization} of a smooth $G$-action on $\obs$ is a $G$-equivariant cochain map $\eta \colon \gg[1] \to \mr{Der}(\obs)$ so that $[\eta(X), \eta(Y)] = 0$ for all $X,Y \in \gg$, and $\d \eta(X) = \rho(X)$ for all $X \in \gg$.

Equivalently, let $\gg_{\mr{dR}}$ denote the dg Lie algebra $\gg[1] \overset{\mr{id}}{\to} \gg$, where $\gg$ acts on $\gg[1]$ by the adjoint action.  A homotopy trivialization of a smooth $G$-action is an extension of the infinitesimal action $\rho$ of $\gg$ to an action of $\gg_{\mr{dR}}$.
\end{definition}

We can construct factorization algebras with $G$-actions as the algebras of classical or quantum observables of a field theory.  Let us consider examples where $G$ is a subgroup of the group $\mr{ISO}(n) = \SO(n) \ltimes \RR^n$, acting on $\RR^n$ by isometries.  For instance, $G$ could be the group $\RR^n$ of translations, or the full group $\mr{ISO}(n)$ of all isometries.  Factorization algebras with an action of $G_{\mr{dR}}$ of this type admit another interpretation, in terms of homotopical algebra.

We will use the following two types of object
\begin{itemize}
 \item \emph{$\bb E_n$ algebras} are cochain complexes with the structure of an algebra for the operad of \emph{little $n$-disks}.  We will work with a specific model for this operad.  The space of $k$-ary operations of this operad can be described by taking the singular cochains of the manifold $\mr{Disk}_n(k)$ of all collections of $k$-disjoint $n$-disks inside the unit $n$-disk.  We refer to for instance \cite[Chapter 4.1]{MarklShniderStasheff} for a definition (see also \cite[Definition 2.12]{ElliottSafronov} in a similar context.)
 \item \emph{Framed $\bb E_n$ algebras}, or $\bb E_n^{\mr{fr}}$-algebras are likewise algebras for the operad of \emph{framed little $n$-disks}, which can alternatively be realized as the semidirect product $\SO(n) \ltimes \bb E_n$.
\end{itemize}

\begin{remark}
The word ``framing'' occurs in factorization homology in two places; this would benefit from some unpacking.  Given an (ordinary, unframed) $\bb E_n$ algebra $
\mc A$, one can attempt to compute the factorization homology of a \emph{framed} manifold with coefficients in $\mc A$.  Using this input, Scheimbauer's theorem \cite{ScheimbauerThesis} provides a \emph{framed} TQFT.

On the other hand, given a framed $\bb E_n$ algebra $\mc A$, one can attempt to compute the factorization homology of any \emph{oriented} manifold with coefficients in $\mc A$.  According to the cobordism hypothesis, as in \cite[Theorem 2.4.18]{LurieCobordism}, Scheimbauer's construction will extend to provide an \emph{oriented} TQFT.
~\hfill$\Diamond$\end{remark}

Now, in the presence of a homotopically trivial action of the group of translations, one can understand the factorization algebra of observables of a quantum field theory as an $\bb E_n$ algebra.  We emphasise that this construction is explicit; it does not simply abstractly imply the existence of an $\bb E_n$ action, but instead allows for the construction of an action of an explicit model for the $\bb E_n$ operad.

\begin{theorem}[{\cite[Corollary 2.30 and 2.39]{ElliottSafronov}}] \label{En_theorem}
Let $\obs$ be a factorization algebra on $\RR^n$, and suppose that the factorization map $\obs(B_r(0)) \to \obs(B_R(0))$ associated to the inclusion of concentric balls $B_r(0) \inj B_R(0)$ for $r < R$ is a quasi-isomorphism.

If $\obs$ is equipped with a smooth $\RR^n_{\mr{dR}}$ action, then the cochain complex $\obs(B_1(0))$ can be equipped with a canonical $\bb E_n$ algebra structure.

If $\obs$ is equipped with a smooth $\mr{ISO}(n)_{\mr{dR}}$ action, then the cochain complex $\obs(B_1(0))$ can be equipped with a canonical $\bb E^{\mr{fr}}_n$ algebra structure.
\end{theorem}

By applying Theorem \ref{En_theorem} to the twisted $\mc N=4$ theories we have been considering in this paper, particularly by Theorem \ref{KW_family_framing_anomaly_thm}, we can deduce the following.

\begin{corollary}
\label{framed E4}
The family of factorization algebras of quantum observables $\obsq_{\lambda,\lambda,\mu}$ over $\CC^2 \bs \{0\}$ canonically carries the structure of a family of $\bb E_4^{\mr{fr}}$-algebras.
\end{corollary}

\begin{proof}
For the factorization algebras of observables of the full $\CC^2 \bs \{0\}$-family of theories, it is manifestly the case that the factorization map associated to the inclusion of concentric balls is a quasi-isomorphism (theories in this family are either topological BF theories, or theories for which the algebra of local observables on any ball is contractible).  So we only need to know that the factorization algebra carries a smooth $\mr{ISO}(4)_{\mr{dR}}$-action in order to apply Theorem \ref{En_theorem}.  

The factorization algebra of quantum observables carries a smooth action of the group $\mr{ISO}(4)$ by Proposition \ref{quantization_of_isometry_prop} as discussed in Remark \ref{smooth_action_rmk}, we need to realize a homotopy trivialization of the infinitesimal $\mf{iso}(4;\CC)$ action.  We have constructed a potential for the action on the global observables $\obsq(\RR^4)$ in Theorem \ref{KW_family_framing_anomaly_thm}.  This action preserves the local observables $\obsq(U)$ supported on an open set $U \sub \RR^4$: this is immediate from the expression for the classical Hamiltonian $J^{\mr{eq}}$ defining the homotopy trivialization.
\end{proof}

Given an $\bb E_4$-algebra, according to Scheimbauer's theorem \cite[Corollary 4.6.4]{ScheimbauerThesis} there is a canonically associated framed extended 4-dimensional TQFT, valued in the Morita category of $\bb E_4$-algebras.  Since in this case our $\bb E_4$-algebras have been equipped with the richer structure of an $\bb E_4^{\mr{fr}}$-algebra, these extended 4-dimensional TQFTs can be promoted to \emph{oriented} extended 4-dimensional TQFTs, i.e. functors from the $(\infty, 4)$-category of oriented bordisms, as constructed in \cite[Definition 2.6.10]{ScheimbauerThesis}.

\begin{remark}
By combining this construction with our results from Section \ref{modvac} we will obtain, say for the point $(1,1,0) \in \CC^3$, a sheaf of $\ZZ/2\ZZ$-graded framed $\bb E_4$-algebras over $[\gg^*/G]$.  Alternatively, we could recover the $\ZZ$-grading by applying an even grading shift to obtain a sheaf of $\ZZ$-graded framed $\bb E_4$-algebras over $[\gg^*[2]/G]$.  Because we can canonically identify the shifted coadjoint quotient with $T^*[3]BG$, this stack is 3-shifted symplectic, or $\bb P_4$.  It is natural to expect that the $\bb E_4$-algebra of quantum local observables that we have constructed provides a quantization of this stack in the sense of \cite{CPTVV}.~\hfill$\Diamond$\end{remark}

\subsection{Filtrations and the Free-to-Interacting Spectral Sequence}
\label{sec filtrations}

The observables of any BV theory sometimes admit natural filtrations,
and these typically offer both conceptual insight and calculational leverage.
Two obvious filtrations are
\begin{itemize}
\item the {\it classical-to-quantum} filtration on quantum observables $\obsq$ arising from powers of~$\hbar$, where
\[
\obsq \supset \hbar \obsq \supset \cdots \supset \hbar^k \obsq \supset \cdots,
\]
whose associated graded is $\obscl[[\hbar]]$, and
\item the {\it free-to-interacting} filtration on {\em classical} observables $\obscl$ arising from powers of the augmentation ideal, where 
\[
\obscl = \symc({\mc E}^*) \supset \symc^{> 0}({\mc E}^*) \supset \cdots \supset \symc^{> k}({\mc E}^*) \supset \cdots,
\]
whose associated graded is the classical observables of the underlying free theory.
\end{itemize}
The first filtration amounts to imposing the classical equations of motion and then tracking quantum corrections in powers of $\hbar$.
In other words, it is a version of the loop expansion for QFT.
The second filtration amounts to imposing the linearized equations of motion and then tracking the corrections due to terms in the action functional of higher polynomial degree.
In other words, it is a version of perturbatively solving PDE.
Here, however, we are interested in more subtle filtrations.

\subsubsection{The Free-to-Interacting Filtration} \label{F2I_section}

The following filtration captures the notion of canonically quantizing first and then imposing the equations of motion.

\begin{definition}
The {\em free-to-interacting} filtration on {\em quantum} observables $\obsq$ is the decreasing filtration where
\[
F^k_{\FtI} \obsq = \prod_{2m+n \geq k} \hbar^m \sym^n({\mc E}^*).
\]
The associated graded is the quantum observables of the underlying free theory.
\end{definition}

Note that the free-to-interacting filtration is exhaustive (that is, $\bigcup_{k \in \ZZ} F^k_{\FtI} \obsq = \obsq$) and separated (that is, $\bigcap_{k \in \ZZ} F^k_{\FtI} \obsq = \emptyset$, because $F^k_{\FtI} \obsq$ is empty if $k < 0$).

The free-to-interacting spectral sequence -- the spectral sequence arising from this filtration -- begins by computing the cohomology of the free quantum theory,
which is typically easy to calculate, 
so that the spectral sequence gives us traction on the observables of the interacting theory.
For instance, on a closed manifold,
this spectral sequence always collapses on the first page (see Theorem 5.6.1 of~\cite{OGthesis}).

Below, we will compute the factorization homology of the {\em abelian} gauge theories arising as twists of 4-dimensional ${\mc N}=4$ super Yang--Mills theory,
and then leverage that information to understand homology for the non-abelian gauge theories using the free-to-interacting spectral sequence.  The underlying free theory of a non-abelian super Yang--Mills theory can be understood simply by turning off the non-trivial Lie bracket on the gauge Lie algebra, producing an abelian gauge theory.

This filtration has an interesting consequence for the theories living in the family of B-type twists, including the holomorphic and Kapustin twists.
Each of these theories has an exact one-loop quantization, 
so the differential on the quantum observables only has terms of order $\hbar^0$ and order $\hbar^1$.
We express it heuristically as
\[
Q + \hbar \Delta + \{I, -\} + \hbar \{I^q, -\},
\]
where $I$ is the classical interaction term and $I^q$ is the quantum interaction term.
(We are suppressing here any dependence on a parametrix or length scale.)
The first page of the spectral sequence is the cohomology with respect to $Q + \hbar \Delta$,
but the differential on the next page depends on the cubic term of $I$ and the linear term of $I^q$.
Later pages depend on higher order terms of $I$ and $I^q$.
Hence, this spectral sequence breaks up the problem of analyzing the observables into a series of more tractable problems.

The free-to-interacting filtration also resolves the following puzzle about the quantum observables of these twisted theories.  {\it A priori}, it does not seem like the quantum observables for the B-twisted theory are very interesting, for the following reason.

\begin{lemma}
The Lie algebra of classical observables in the B-twisted theory $\mc E_{1,1,0}$ admits no non-trivial quantum deformations as a cotangent theory.
\end{lemma}

This is proven in \cite[Proposition 3.19]{EY2} and the discussion that immediately follows it, but we will give a brief argument here.

\begin{proof}
The B-twisted classical field theory is described by the local Lie algebra $\Omega^\bullet(\RR^4) \otimes \gg[\eps]$,
where $|\eps| = -1$.
Note that by the Poincar\'e lemma, this local Lie algebra is quasi-isomorphic to $\gg[\eps]$, viewed as a locally constant sheaf of dg Lie algebras on $\RR^4$.

As the classical observables $\obscl_B$ of the B-twist are the Chevalley-Eilenberg cochains of the local Lie algebra, 
we thus see that the classical observables are quasi-isomorphic to $\clie_{\lie}^\bu(\gg[\eps])$ as a locally constant prefactorization algebra on $\RR^4$.
Hence, the classical observables encode $\clie_{\lie}^\bu(\gg[\eps])$ as an $\bb E_4$-algebra.

This algebra is, of course, a {\em commutative} algebra but now viewed merely as an $\bb E_4$-algebra.
When we quantize the B-twist, we produce a deformation of this $\bb E_4$-algebra,
but by the HKR theorem for $\bb E_n$ algebras of Calaque and Willwacher \cite{CalaqueWillwacher},
one finds that there are no nontrivial $\bb E_4$-deformations that are equivariant for the action of $\CC^\times$, acting on $\eps$ with weight one. 
\end{proof}

In other words, the quantization seems to have done nothing interesting, in the sense that the quantum observables are equivalent to the classical observables.  This result might be concerning, but our filtration offers some relief.

\begin{lemma}
As a \emph{filtered} $\bb E_4$-algebra, the quantum observables of a B-type twist are a nontrivial deformation of the classical observables.
\end{lemma}

Before giving the short proof, we recall a motivating but simpler example.
Consider the Chevalley-Eilenberg cochains $\clie_{\lie}^\bu(\gg)$ of a semisimple Lie algebra $\gg$.
As a dg commutative algebra, it is formal, i.e., quasi-isomorphic to its cohomology,
which is a free graded-commutative algebra.
Little about the original Lie algebra can be read off from the cohomology.
On the other hand, there is a canonical filtration on $\clie_{\lie}^\bu(\gg)$ by powers of the augmentation ideal,
whose associated graded is the free graded-commutative algebra $\sym(\gg^*[-1])$.
Thus from $\clie_{\lie}^\bu(\gg)$ as a {\em filtered} algebra, 
one can fully recover the Lie algebra $\gg$
(in brief, via the tangent complex of the associated graded algebra and via the differential).

\begin{proof}
Consider the quantum observables of the B-type twist equipped with the free-to-interacting filtration.
The associated graded is the complex of quantum observables of an {\em abelian} B-type twist,
which is a deformation of the complex of classical observables of that abelian B-type twist.
Hence, to prove the lemma, it suffices to show that this associated graded, the quantum observables for an abelian theory, nontrivially the classical observables of the abelian theory.

To see this, we study a simpler, closely related deformation.
Let $M$ be a closed 3-manifold and consider the manifold $M \times \RR$, equipped with the projection $\pi \colon M \times \RR \to \RR$. 
The quantum observables of the B-type twist on $M \times \RR$ push forward to a factorization algebra $\pi_* \obsq$ on $\RR$.
As this factorization algebra is locally constant, it corresponds to an $\EE_1$-algebra.
This correspondence applies to the associated graded as well,
so that $\pi_* \gr \obsq$ corresponds to some $\EE_1$-algebra that is a deformation of the algebra for the $\pi_* \gr \obscl$.
In other words, to prove the lemma, it suffices to show that we obtain a nontrivial deformation of $\EE_1$-algebras when we quantize the {\em abelian} B-type twist on the manifold~$M \times \RR$.

We give a precise version of this claim in Lemma~\ref{abelian B def} and Corollary~\ref{nontriv def}.
(Strictly speaking, the results stated there use the fiberwise polynomial observables discussed below,
but the argument applies to the observables here.)
But such a claim is a quick consequence of the techniques and results of Chapter 4,~\cite{Book1}.
It is shown there that for a free BV theory on $\RR$, 
the classical observables correspond to algebraic functions on a dg vector space of the form $T^* V$
and that the quantum observables correspond to the corresponding dg Weyl algebra.
A Weyl algebra is a nontrivial deformation, as needed for our lemma here.
\end{proof}

\subsubsection{The Fiberwise Polynomial Observables and their Filtration}

For a cotangent theory, there is a variant definition of the observables, both classical and quantum.
The essential idea is to ask for observables that are formal power series along the ``base'' (i.e., as functions on the $L$-direction) but that are polynomial along the ``fiber'' (i.e., as functions on the $L^!$-direction).
It might help the reader to bear in mind that this kind of construction appears in the definition of differential operators on a manifold~$X$:
the symbol of a differential operator is a smooth function on $T^* X$ that is a polynomial when restricted to a cotangent fiber~$T^*_x X$.
(We elaborate on this analogy in Section~\ref{DOTs} below.)

\begin{definition}
For the cotangent theory associated to a local $L_\infty$ algebra $L$, 
the {\em fiberwise polynomial classical observables} assign to each open set $U$,
the cochain complex
\[
\obscl_{\mr{fp}}(U) = \left( \bigoplus_{n = 0}^\infty \prod_{m = 0}^\infty  \sym^m((L(U)[1])^*) \widehat{\otimes} \sym^n((L^!(U)[-2])^*), \{I, -\}\right).
\]
The differential preserves the observable that are homogeneous of degree $n$ with respect to the $L^!$-inputs.
\end{definition}

The reason we can implement this construction is that the local $L_\infty$ algebra of a cotangent theory is the extension of a local $L_\infty$ algebra $L$ by its module $L^![-3]$,
viewed as an abelian $L_\infty$ algebra.
For such an extension $\gg \ltimes V$ of a Lie algebra $\gg$ by a module $V$, 
the algebra of Chevalley--Eilenberg cochains has the form
\[
\clie_{\lie}^\bu(\gg \ltimes V) \cong \clie_{\lie}^\bu(\gg, \symc(V^*[-1])) \cong \prod_{n=0}^\infty \clie_{\lie}^\bu(\gg, \sym^n(V^*[-1])).
\]
There is thus a natural grading by symmetric powers of $V^*[-1]$,
which we will call the {\it fiberwise weight}.
By taking the sum rather than the product over the fiberwise weights,
we obtain a natural subalgebra 
\[
\clie_{\lie}^\bu(\gg, \sym(V^*[-1])) \cong \bigoplus_{n=0}^\infty \clie_{\lie}^\bu(\gg, \sym^n(V^*[-1]))
\]
that is polynomial in the $V$-direction.
In the definition above of fiberwise polynomial observables, 
we likewise replaced the product over those symmetric powers by the direct sum.

When we quantize a cotangent theory,
only one-loop diagrams appear, as we have seen already,
and these only depend on the $L$-inputs (the base of the cotangent space).
Hence, the $\hbar$-weighted term $I^\mr{q}$ of the quantized action is only a function of the $L$-inputs,
and the operator $\{I^\mr{q}, -\}$ thus has the property that it lowers the fiberwise weight by one.
Thus, if this quantized action satisfies the quantum master equation,
the differential on the quantum observables preserves the subalgebra of fiberwise polynomial observables.  This argument shows that the following sub-factorization algebra is well-defined.

\begin{definition}
Given an exact one-loop quantization of the cotangent theory,
the algebra of {\em fiberwise polynomial quantum observables} $\obsq_{\mr{fp},\hbar}$ is the sub-factorization algebra of $\obsq$ given by taking the direct sum over symmetric powers in~$L^!$.
\end{definition}

As the differential is a polynomial in $\hbar$, we can even set $\hbar$ equal to 1 (or any nonzero value),
a very special feature of exact one-loop quantizations.
(It implies that, in a sense, the theory is less perturbative and perhaps closer to a complete answer.)
In particular, the quantum corrections to the classical action functional are treated on an equal footing.

\begin{definition}
Given an exact one-loop quantization of the cotangent theory,
let $\obsq_{\mr{fp}}$ denote the quotient $\obsq_{\mr{fp}, \hbar}/(\hbar -1)$.
\end{definition}

The reader might recognize this quotient as analogous to a situation in quantum mechanics:
sometimes one imposes the commutation relation $[p,x] = i \hbar$ (viewing $\hbar$ as a parameter) or one imposes the commutation relation $[p,x] = i$.
The second appears when viewing the operators as differential operators on $\RR$, with $p = i \frac{\d}{\d x}$;
the first appears as a deformation quantization of functions on~$T^*\RR$.

There is a useful filtration one can put on these observables.
Consider the underlying graded vector space of the observables
\[
\bigoplus_{n = 0}^\infty \prod_{m = 0}^\infty  \sym^m((L(U)[1])^*) \widehat{\otimes} \sym^n((L^!(U)[-2])^*).
\]
It has a natural bigrading with the $(m,n)$-component consisting of $\sym^m((L(U)[1])^*) \widehat{\otimes} \sym^n((L^!(U)[-2])^*)$.
Given a homogeneous $(m,n)$-graded component,
the classical differential $\{I,-\}$ preserves the second term $n$ (polynomial degree in $L^!$-inputs) but can increase the first term $m$ (polynomial degree with respect to $L$-inputs).
The BV Laplacian sends $(m,n)$ to $(m-1,n-1)$, lowering both terms equally.
The quantum correction $\{I^\mr{q}, - \}$ lowers the second term $n$ by $1$ but can increase the first term $m$ arbitrarily, so if $f$ has bidegree $(m,n)$, the $\{I^\mr{q},f\}$ is a sum of terms with bidegree $(m+k, n-1)$ for any $k \geq 0$.

An arbitrary element in these observables has the form
\[
f = \sum_{\text{finite } n} \sum_{m = 0}^\infty f_{m,n},
\]
where there are only finitely many values of $n$
such that  the formal power series
\[
f_{(n)} = \sum_{m = 0}^\infty f_{m,n}
\]
is nonzero.
Running over the nonzero terms $f_{m,n}$ in $f$, 
we see that there is a maximum value of the difference $n-m$
because $m$ is bounded below by 0 and there are only finitely many values of $n$.
Set
\[
[f]_{\overline{\Delta}} = \max_{m,n} \{ n-m \, : \, f_{m,n} \neq 0\}.
\]
Note that the quantum differential $d^\mr{q}$ satisfies
\[
[d^\mr{q} f]_{\overline{\Delta}} \leq [f]_{\overline{\Delta}}
\]
so long as $I^\mr{q}$ has no linear term.

\begin{remark}
Let us comment on why this last assumption holds for the twisted $\mc N=4$ gauge theories.  The linear term in $I^{\mr{q}}$ is given by the weight of a ``tadpole'' Feynman diagram: a wheel with one external edge.  As we discussed in Example \ref{tadpole_example}, since $\gg$ is reductive, the weight of such a diagram will necessarily be zero.
~\hfill$\Diamond$\end{remark}

With this condition being satisfied, we can define the following filtration.

\begin{definition}\label{antidiag filt}
The {\em anti-diagonal filtration} on $\obsq_{\mr{fp}}$ is
\[
F^k_{\overline{\Delta}} \obsq_{\mr{fp}} =  \{ f \, \colon \, [f]_{\overline{\Delta}} = k\},
\]
and the differential preserves this filtration when the linear quantum interaction term $I^{\mr{q}}_1$ vanishes.
\end{definition}

This filtration is again exhaustive and separated; 
exhaustiveness very much depends on our restriction to fiberwise polynomial observables.

The spectral sequence arising from this filtration begins by computing the cohomology of the fiberwise polynomial quantum observables of a {\em free} quantum theory.
Here the differential is $\{I_2+ I^\mr{q}_2, - \} + \Delta$, where the subscript 2 indicates that we take the quadratic term of the quantized action functional.
Hence it is a close cousin of the free-to-interacting spectral sequence we have discussed already,
and many observations we made there apply in this situation.
In particular, the following result holds by an identical argument.

\begin{lemma}
As a filtered $\bb E_4$-algebra, the fiberwise polynomial quantum observables $\obsq_{\mr{fp}}$ of a B-type twist are a nontrivial deformation of the fiberwise polynomial classical observables.
\end{lemma}

We now explain a nice feature of this filtration when the manifold is closed.

\begin{lemma}\label{obs as det}
Let $M$ be a closed manifold, and let $L$ be a local $L_\infty$ algebra on $M$ 
such that its cotangent theory has an exact one-loop quantization.
Then the cohomology of the fiberwise polynomial quantum observables $\obsq_{\mr{fp}}(M)$ is isomorphic to 
\[
\det \left(\mr H^\bu(L(M), \ell_1)\otimes \gg\right)[d_M],
\]
where $d_M$ modulo 2 agrees with the Euler characteristic.
\end{lemma}

Here $(L(M), \ell_1)$ denotes the underlying elliptic complex of the local $L_\infty$ algebra,
and this complex has finite-dimensional cohomology because $M$ is closed.

\begin{remark}
It is possible to write down the shift $d_M$ explicitly.  It takes the form
\[d_M = \sum_{k \in \ZZ} c_k \dim(\mr H^k(L(M), \ell_1)),\]
where 
\[c_k = \begin{cases}
         k & \text{ if } k \text{ is odd} \\
         1-k & \text{ if } k \text{ is even.} 
        \end{cases}
\]
~\hfill$\Diamond$\end{remark}

\begin{proof}
Consider the spectral sequence for the anti-diagonal filtration.
The first page is the cohomology of the quantum observables for a {\em free} cotangent theory.
By Lemma 13.7.1.1 of \cite{Book2}, this cohomology is one-dimensional,
and so the spectral sequence collapses.
We outline this argument.

For the fiberwise quantum observables of a free theory, 
there is another filtration by fiberwise polynomial degree.
Indeed for any dg vector space $V$, 
we can consider fiberwise polynomial functions on its shifted cotangent space $T^*[-1] V = V \oplus V^*[-1]$,
which is the dg commutative algebra
\[
\bigoplus_{n \geq 0} \symc(V^*) \otimes \sym^n(V[1])).
\]
We call $n$ the fiberwise weight,
and we note that the differential preserves fiberwise weight.
If we have a BV Laplacian arising from the -1-shifted symplectic pairing on the shifted cotangent space,
then it lowers this fiberwise weight by one.
Thus, on the quantum fiberwise observables,
there is a filtration by fiberwise weight and hence a spectral sequence.
The first page is simply the cohomology of the classical fiberwise observables.
Under our hypotheses, since $V$ is an elliptic complex $(L(M), \ell_1)$ on a closed manifold,
we know that $\mr H^\bu(V)$ is finite-dimensional,
so that the first page of the spectral sequence is isomorphic to
\[
\bigoplus_{n \geq 0} \symc(\mr H^\bu(V^*)) \otimes \sym^n(\mr H^\bu(V)[1])).
\]
The residual differential is the BV Laplacian for the shifted cotangent space $T^*[-1]\mr H^\bu(V)$.
We can then invoke Lemma 13.7.1.1 of \cite{Book2} to see that the second page is one-dimensional.
We have thus computed that the cohomology of the fiberwise polynomial quantum observables of a free cotangent theory on a closed manifold is one-dimensional.
\end{proof}

\subsection{Algebras of Differential Operator Type}
\label{DOTs}

For cotangent theories, the factorization algebras of quantum observables have the flavor of twisted differential operators.
We wish to capture what we mean with some precision.
(Perhaps more accurately, for an $n$-dimensional topological field theory of cotangent type, the observables are twisted $\bb E_n$ algebras of differential operators. 
We return to this issue after discussing the $\bb E_1$ case, 
which is more familiar.)

Recall the notion of twisted differential operators for a smooth manifold $X$ (or complex manifold or smooth variety, depending on the geometric context in which one wishes to work).
Let $\mc{O}_{\mr{fp}}(T^*X)$ denote the ring of fiberwise polynomial functions on the cotangent bundle $T^* X$ as a sheaf on $X$.
In other words, it is the sheaf $\sym_{\mc{O}_X}(\mc{T}_X)$,
where $\mc{T}_X$ denotes the tangent sheaf of $X$.
Hence $\mc{O}_{\mr{fp}}(T^*X)$ is naturally graded by symmetric powers in vector fields.

\begin{definition}
An {\em algebra of twisted differential operators} is a sheaf $A$ on $X$ of positively filtered associative algebras such that $\gr A = \mc{O}_{\mr{fp}}(T^*X)$.
\end{definition}

For each line bundle $L \to X$, the differential operators on $L$ are an example of an algebra of twisted differential operators.
In general, the cohomology group $\mr H^1(X, \Omega^1_{\mr{cl}})$ parametrizes twisted differential operators up to isomorphism of such filtered algebras, where $\Omega^1_{\mr{cl}}$ denotes the sheaf of closed 1-forms.
(More precisely, we are referring here to {\em locally trivial} TDOs \cite{BBJantzen}.)

We wish to study the analogous construction but where $X$ is replaced by a formal moduli space.
In this situation, we can offer a refinement of the usual notion because for a formal moduli space,
its algebra of functions itself already has a natural filtration.

Let $\mc{X}$ be a formal moduli space with basepoint $x$, 
modeled by an $L_\infty$ algebra $L$.  
We can take $\clie^\bu_{\rm Lie}(L)$ as a model for $\mc{O}(\mc{X})$, 
the dg commutative algebra of functions on $\mc{X}$.
Let $|\mc{X}|$ denote the linear formal moduli space associated to the tangent complex ${\bb T}_x \mc{X}$;
it can be modeled by the abelian $L_\infty$ algebra $|L|$, namely the underlying cochain complex of~$L$.
Note that $\mc{O}(\mc{X})$ has a natural filtration,
in essence determined by the maximal ideal of functions vanishing at the basepoint.
In our concrete model $\clie^\bu_{\rm Lie}(L)$, 
the filtration is by powers of the augmentation ideal.
By construction, 
\[
\gr_{\overline{\Delta}} \mc{O}_{\mr{fp}}(\mc{X}) \simeq \mc{O}_{\mr{fp}}(|\mc{X}|)
\]
so that this filtration identifies $\mc{O}_{\mr{fp}}(\mc{X})$ as a commutative deformation of functions on a linear formal moduli space. 

We want to study algebras that closely resemble the algebra of differential operators on $\mc{X}$.
As in the classic situation, we are interested in deformation quantizations of fiberwise polynomial functions on $T^* \mc{X}$,
by which we mean the dg algebra
\[
\mc{O}_{\mr{fp}}(T^*\mc{X}) = \sym_{\mc{O}(\mc{X})}({\bb T}_x \mc{X}).
\]
It can be modeled by~$\clie_{\mr{Lie}}^\bu(L, \sym(L[1]))$.
There is an obvious grading by symmetric powers of $L[1]$,
which is the analog of the grading on fiberwise polynomial functions for a smooth manifold~$X$.

On the other hand, we also have an anti-diagonal filtration along the lines of Definition~\ref{antidiag filt}.
It is a natural extension of the filtration on $\mc{O}(\mc{X})$,
the functions on the base~$\mc{X}$.
By construction, 
\[
\gr_{\overline{\Delta}} \mc{O}_{\mr{fp}}(T^*\mc{X}) \simeq \mc{O}_{\mr{fp}}(T^*|\mc{X}|)
\]
so that the anti-diagonal filtration identifies $\mc{O}_{\mr{fp}}(T^*\mc{X})$ as a commutative deformation of functions on a linear {\em symplectic} formal moduli space. 

We will now discuss deformation quantizations that satisfy conditions with respect to both filtrations.

\begin{definition}\label{def DOT}
For a formal moduli space $\mc{X}$, 
a dg associative algebra $A$ is {\em of differential operator type} if it admits two filtrations $F^\bu_{\mr{fib}}$ (positively graded) and $F^\bu_{\overline{\Delta}}$ (unbounded) such that
\begin{itemize}
\item the associated graded for the fiberwise filtration satisfies
\[
\gr_{\mr{fib}} A \simeq \mc{O}_{\mr{fp}}(T^*\mc{X})
\]
where the right hand side is equipped with its grading by symmetric powers of vector fields, and
\item the associated graded for the anti-diagonal filtration satisfies
\[
\gr_{\overline{\Delta}} A \simeq D(|\mc{X}|),
\]
where the right hand side denotes the algebra of differential operators on $|\mc{X}|$.
These differential operators can be viewed as a fiberwise polynomial version $\mr{Weyl}_{\mr{fp}}(T^*|\mc{X}|)$ of the Weyl algebra for the symplectic dg vector space $T^*|\mc X|$.
\end{itemize}
\end{definition}

\begin{example}
The algebra $D(\mc X) = U_{\mr{fp}}(T_{\mc X})$ of differential operators on $\mc X$, defined by taking a fiberwise polynomial version of the universal enveloping algebra, is of differential operator type.  The underlying graded vector space of $D(\mc X)$ can be written out explicitly as
\[\symc^\bullet(L^*) \otimes \sym^\bullet(L[1]),\]
and equipped with a pair of filtrations as in Section \ref{sec filtrations}.  One can readibly check that these do indeed define filtered dg algebra structures on $D(\mc X)$.

The very simplest example is provided by the case where $\mc X$ is purely classical (here meaning non-derived), so that $L$ is concentrated in degree one.  Then $\mc X = |\mc X|$, and we are able to use the trivial anti-diagonal filtration.
\end{example}

We will see that for the framed $\bb E_4$ algebra of the B-twist, 
when we take factorization homology over a closed 3-manifold $N$,
we get such algebras of differential operator type.

\begin{remark}
One can generalize Definition \ref{def DOT} in the following way.  Replace ``dg associative'' in the definition with $\bb E_n$,  replace $T^*\mc{X}$ with $T^*[1-n]\mc{X}$, and replace the Weyl algebra on $T^*|\mc{X}|$ with the $\bb E_n$-Weyl algebra on $T^*[1-n]|\mc{X}|$.  Such an $\bb E_n$ algebra deformation of $\mc{O}_{\mr{fp}}(T^*[n-1]\mc{X})$ is defined to be an \emph{$\bb E_n$ algebra of differential operator type}.  One can build examples of such $\bb E_n$ algebras as the $\bb E_n$-enveloping algebras of the tangent sheaf of $\mc X$, as in \cite{Knudsen}.

In this sense, the B-twist produces an $\bb E_4$-algebra of differential operator type on $\mc{X} = \Flat_G(\RR^4)^\wedge_0$, the formal neighborhood of the trivial flat $G$-bundle.
~\hfill$\Diamond$\end{remark}

\section{Factorization Homology Computations}

We will now explore some consequences of our results on manifolds other than $\RR^4$.
The key mechanism for global results is factorization homology:
a factorization algebra satisfies a local-to-global condition that means its value on an interesting manifold $M$ is determined by its behavior on ``small'' opens (e.g., disks sitting inside $M$).
In particular, the observables of a perturbative field theory (whether classical or quantum) form a factorization algebra, by the results of \cite{Book2}, 
and it is interesting to ask about the global observables on various manifolds,
which provide interesting algebraic structures depending on the manifold and which can thus encode invariants of the manifolds.

Let us focus on 4-dimensional theories such as those discussed in this paper.
If the 4-manifold $M$ is closed, we will see that the complex of global quantum observables is one-dimensional, 
if a quantization exists on that manifold.
(We will assume that hypothesis for the moment.)
If the manifold is a product $M = \RR \times N$ with $N$ a closed 3-manifold,
then the global observables determine a factorization algebra on $\RR$ by pushing forward along the projection to $\RR$.
This 1-dimensional factorization algebra encodes an interesting dg associative (or $A_\infty$) algebra
that a physicist might call the algebra of observables of the compactification of the 4-dimensional theory along $N$.
If the manifold is a product $M = \RR^2 \times N$ with $N$ a closed 2-manifold,
then the pushforward of observables to $\RR^2$ encodes a dg vertex algebra 
(possibly even just an $\EE_2$ algebra, if the 4-dimensional theory is topological along the $\RR^2$ directions).

A key assumption here was that the theory makes sense on a given 4-manifold $M$ and that a quantization exists.
This assumption must be checked and need not always be satisfied.
For instance, the holomorphic twist as a {\em classical} theory is defined on a complex surface,
so that $M$ must be equipped with a complex structure.
Our explicit quantization in this paper provides a quantization on manifolds closely related to $\CC^2$, notably
\begin{itemize}
\item complex tori such as $\CC^2/\Lambda$, where $\Lambda$ is a lattice $\ZZ^4 \hookrightarrow \CC^2$;
\item open subsets of $\CC^2$, notably $\CC^2 \setminus \{0\}$, which is diffeomorphic to $\RR \times S^3$ and hence offers an example of compactification along a 3-manifold; and
\item products $\CC \times E$, where $E$ is an elliptic curve.
\end{itemize}
In this paper we will discuss only such manifolds, deferring a treatment of other complex surfaces to later work.

On the other hand, we have provided an $\SO(4)$-equivariant quantization of twists living in the subspace of $\CC^3$ of elements of the form $(t, t, u)$,
so that our quantizations make sense on any oriented $4$-manifold.
For these twists,
the global quantum observables are the factorization homology of the framed $\EE_4$ algebra $\obsq_{(t,t,u)}$ associated with our quantization.
For such twists, we find that
\begin{itemize}
\item for any closed oriented $4$-manifold $M$, the cochain complex $\int_M \obsq_{(t,t,u)}$ has one-dimensional cohomology;
\item for any closed oriented $3$-manifold $N$, the factorization homology $\int_N \obsq_{(t,t,u)}$ is equivalent to an $\EE_1$-algebra;
\item for any closed oriented $2$-manifold $N$, the factorization homology $\int_N \obsq_{(t,t,u)}$ is equivalent to an $\EE_2$-algebra; and
\item for any closed oriented $1$-manifold $N$, the factorization homology $\int_N \obsq_{(t,t,u)}$ is equivalent to an $\EE_3$-algebra.
\end{itemize}
Indeed, this framed $\EE_4$-algebra $\obsq_{(t,t,u)}$ determines a fully extended and oriented 4-dimensional topological field theory, in the sense of Baez--Dolan--Lurie, by \cite{ScheimbauerThesis}.
It would be worthwhile to unravel how this functorial TFT arising from the Kapustin--Witten theories relates to other constructions occurring in the literature.
We defer this very interesting question to future explorations,
although we make some suggestions in this paper.

In this subsection, we will treat each type of twist separately.
For each twist we will first describe the factorization homology of the classical observables, 
before describing the quantum observables for abelian gauge theories and then deducing consequences for quantum observables of the nonabelian gauge theories,
using the filtrations we have already introduced.
The wonderful thing about the abelian gauge theories is that computations reduce to understanding the cohomology of elliptic complexes (notably the Dolbeault, de Rham, or mixed complexes).
Such questions are often rather tractable, as we will see,
following the treatment of abelian Chern--Simons theory in Section 4.5 of \cite{Book1}.
(We will freely refer to results from that book with minimal exposition.)

Note that in the abelian case, there is nothing interesting to say about the family of theories over the moduli of vacua,
as the family is manifestly constant.

\subsection{The Holomorphic Twist}

Let $G$ be a reductive algebraic group.
For any complex surface $M$, consider a holomorphic principal $G$-bundle $P \to M$ and let $\dbar_P$ denote the $\dbar$-connection on the adjoint bundle $\ad(P) = P \times^G \gg$.
The space of fields is then
\[
\mc{E}_P = \Omega^{\bu,\bu}(M, \ad(P))[1] \oplus \Omega^{\bu,\bu}(M, \ad(P)^*)[2],
\]
which is the space of sections of a graded vector bundle $E_P \to M$.
We equip this graded vector space with the differential $\dbar_P$.
The cohomology of the space of fields is
\[
\mr H^\bu(\mc{E}_P) = \mr H^{\bu,\bu}_{\dbar}(M,\ad(P) )[1] \oplus \mr H^{\bu,\bu}_{\dbar}(M, \ad(P)^*)[2],
\]
where we view $\mr H^{i,j}_{\dbar}(M, \ad(P))$ as concentrated in cohomological degree $i+j$ and use the subscript $\dbar$ to indicate the differential.
When $P$ is the trivial bundle $G \times M$,
this reduces to
\[
\mr H^\bu(\mc{E}_{\mr{triv}}) = \mr H^{\bu,\bu}_{\dbar}(M )\otimes \gg[1] \oplus \mr H^{\bu,\bu}_{\dbar}(M) \otimes \gg^*[2].
\]
For a closed complex surface  (i.e., compact and without boundary), 
the cohomology $\mr H^\bu(\mc{E}_P)$ is finite-dimensional.
Furthermore,  
there is a canonical pairing on~$\mr H^\bu(\mc{E}_P)$ of degree~$-1$ that is invariant for the shifted Lie bracket.

\begin{lemma}\label{hol ab class obs}
Let $G$ be an abelian complex algebraic group.
On the closed complex surface $M$ with holomorphic principal $G$-bundle $P \to M$, 
the cohomology of the classical observables $\obscl(M)$ is isomorphic to~$\symc(\mr H^\bu(\mc{E}_P)^*)$,
which is a finitely-generated free graded-commutative algebra.
\end{lemma}

We remark that this algebra also has a natural 1-shifted Poisson bracket,
arising from the canonical pairing on~$\mr H^\bu(\mc{E}_P)^*$.

\begin{proof}
This result is a special case of results from \cite{Book1};
see especially Chapter 5, Section 6, and Appendix~C.
We gloss the complete argument here.

The graded vector space of global observables is a version of a completed symmetric algebra:
\[
\obscl(M) = \prod_{n \in \NN} \Gamma(M^n, \left(E^!_P\right)^{\boxtimes n})_{S_n},
\]
where 
\begin{itemize}
\item $M^n$ denotes the $n$-fold product of $M$ with itself; 
\item $\boxtimes$ denotes the outer product of vector bundles: 
for $V \to M$ and $W \to N$ vector bundles, $V \boxtimes W \to M \times N$ is the obvious vector bundle;
\item $E^!_P = E_P^* \otimes \mr{Dens}_M$ is the dual vector bundle $E_P^*$ to $E_P$ twisted by the density line.
\end{itemize}
It has a differential determined by $\dbar_P$. 
We filter this $\obscl(M)$ by powers of the augmentation ideal,
and it is complete with respect to this filtration.
We view it as living in a category of pro-cochain complexes.
(Note that we could take the direct sum rather than the direct product, 
if we wanted a symmetric algebra rather than a completed symmetric algebra.
This would avoid working with complete filtered complexes.)

The cohomology of $\obscl(M)$ is thus the product over $\NN$ of the cohomology for each factor.
For each $n$, that cohomology is finite-dimensional and agrees with~$\sym^n(\mr H^\bu(\mc{E}_P)^*)$.
\end{proof}

For $G$ nonabelian, the fields $\mc{E}_P$ have a {\em nontrivial} shifted dg Lie algebra structure,
by mixing the wedge product of Dolbeault forms with the Lie bracket on $\gg$.
This dg Lie algebra $\mc{E}_P[-1]$ has a natural moduli-theoretic interpretation: it can identified with the $(-1)$-shifted tangent complex to the derived stack $T^*[-1]\mr{Higgs}_G(M)$, the shifted cotangent to the stack of $G$-Higgs bundles (as discussed in Section \ref{holo_vacua_section}), at a point $(P,0)$ corresponding to the zero Higgs bundle.

The cohomology $\mr H^\bu(\mc{E}_P)[-1]$ again carries the structure of a graded Lie algebra.

\begin{cor}
On the closed complex surface $M$ with holomorphic principal bundle $P \to M$, 
there is a canonical filtration on the classical observables $\obscl(M)$ by powers of the augmentation ideal. 
In the associated spectral sequence, the first page is isomorphic to~$\symc(\mr H^\bu(\mc{E}_P)^*)$.
\end{cor}

The next interesting differential in the spectral sequence is the Chevalley differential arising from the Lie algebra structure of~$\mr H^\bu(\mc{E}_P)$.

In most cases of geometric interest, we have a stronger result.
For $M$ a closed K\"ahler manifold of complex dimension 2, 
Hodge theory offers a deformation retract of $\mc{E}_P$ onto its cohomology.
In fact, by the main result of \cite{DGMS},
$\mc{E}_P[-1]$ is quasi-isomorphic to its cohomology {\em as a dg Lie algebra},
so that we have the following.

\begin{lemma}
For $M$ a closed K\"ahler manifold of complex dimension 2, 
the classical observables $\obscl(M)$ are quasi-isomorphic to $\clie^\bu_{\mr{Lie}}(\mr H^\bu(\mc{E}_P)[-1])$ as dg commutative algebras.
\end{lemma}

We now turn to describing the global quantum observables on a closed surface.
Since we constructed a quantization on $\CC^2$,
we only know that we have quantizations for manifolds closely related to $\CC^2$.
For instance, our quantization determines a quantization on any complex torus $\CC^2/\Lambda$, 
where $\Lambda$ denotes a lattice $\ZZ^4 \hookrightarrow \CC^2$,
because our quantization is manifestly $\Lambda$-equivariant and hence the quantization descends from $\CC^2$ to the quotient.
Similarly, we obtain a quantization on a Hopf surface $\CC^2 \setminus \{0\}/q^\ZZ$, 
with $q \in \CC^\times\times \CC^\times$,
because our quantization restricts to $\CC^2 \setminus \{0\}$, the punctured affine plane,
and our quantization is equivariant under rescaling by $q$. 
In general, if $U$ is an open subspace of $\CC^2$ preserved by the action by a discrete subgroup $L$ of $\GL_2(\CC) \ltimes \CC^2$ with the property that $U/L$ is a manifold, then we have a quantization for the holomorphic theory on $U/L$.

\begin{remark}
Now that we are discussing the quantization of the algebra of observables, we will restrict attention to the trivial principal $G$-bundle over $M$, in order to ensure that our results from Section \ref{anomaly_section} for the quantization of the theory on $\CC^2$ can be applied (by passing, as discussed above, to a discrete quotient of an open subset of~$\CC^2$).
~\hfill$\Diamond$\end{remark}

\begin{prop} \label{holo_4d_fact_hom_prop}
Let $G$ be a reductive algebraic group.
Consider the trivial holomorphic $G$-bundle on $M$, 
a closed complex surface of the form discussed in the paragraph above, so that the holomorphically twisted theory on $M$ admits a quantization.
Then the cohomology of the fiberwise polynomial quantum observables $\obsq_{\mr{fp}}(M)$ is isomorphic to 
\[
\det \left(\mr H^{\bu,\bu}_{\dbar}(M)\otimes \gg\right)[d_M],
\]
where $d_M$ modulo 2 agrees with the Euler characteristic.
\end{prop}

This result is an application of Lemma~\ref{obs as det}.

\begin{remark} \label{fact_hom_near_a_vacuum_rmk}
 We have a similar result describing the quantum observables on $M$ for the theory near any choice of vacuum $[x] \in [\gg^*/G]$.  By virtue of Lemma \ref{symmetry_breaking_lemma}, we can identify the perturbative theory over the vacuum $[x]$ with the theory over the vacuum $0$ for the gauge group $G_x$, the stabilizer of $x$.  The cohomology of the fiberwise polynomial quantum observables over $[x]$ can therefore be identified with $\det \left(\mr H^{\bu,\bu}_{\dbar}(M)\otimes \gg_x\right)[d_M]$.~\hfill$\Diamond$\end{remark}

We want to mention one last construction that leads to an $A_\infty$-algebra that has the flavor of the higher Kac-Moody algebras of Faonte, Hennion and Kapranov \cite{FHK}.
Consider the holomorphic twist on the manifold $\CC^2 \setminus \{0\}$,
which \cite{FHK} would denote~$\oAA^2$.
Using spherical coordinates, we can view $\oAA^2$ as $\RR_{>0} \times S^3$,
and hence we can compactify the theory along $S^3$ to get a one-dimensional theory.
Our quantization lives on this manifold,
so the observables form a factorization algebra on $\oAA^2$
and their pushforward to $\RR_{> 0}$ along the radial coordinate $\rho \colon \oAA^2 \to \RR_{> 0}$ is a factorization algebra.
We call this pushforward the algebra of {\em spherical operators} of the theory.
This factorization algebra encodes the ``canonical radial quantization'' of our holomorphic theory.

There is a rather precise statement we can make, 
thanks to \cite{GwilliamWilliamsKM},
notably the discussion following Definition 3.12.
Consider $V$ a dg vector space of finite rank,
which determines an elliptic complex $\Omega^{0,\bu}(\oAA^2) \otimes V$ on $\oAA^2$.
There is an associated {\em free} holomorphic cotangent theory,
sometimes known as a $\beta\gamma$ system.
It is shown that there is a dense inclusion
\[
U(\mc{H}_{V,2}) \to \rho_* \obsq
\]
of factorization algebras on $\RR_{> 0}$,
where $\mc{H}_{V,2}$ is the dg Heisenberg algebra of \cite[Definition 3.12]{GwilliamWilliamsKM}.
In other words, for a free cotangent theory, there is a dg Weyl algebra inside the algebra of spherical operators.
It encodes the operators one knows from canonical radial quantization.

It is straightforward to write down the fiberwise polynomial version of this construction,
so that we use completed functions in the ``base'' direction.
In other words, we view the vector space $V$ as replaced by $\widehat{V}$, 
the formal neighborhood of the origin in $V$.
This fiberwise polynomial algebra $U_{\mr{fp}}(\mc{H}_{V,2})$ should be seen as the ring of differential operators on the space $\maps(\oAA^2, \widehat{V})$,
which is a kind of higher loop space.
(For discussion of this interpretation, again see \cite{FHK} and~\cite{GwilliamWilliamsKM}.)
The algebra of spherical operators is a completion of that algebra,
and it is controlled by differential operators on this higher loop space.

Something similar should hold for the holomorphic theory with a nonabelian gauge group.
We replace the formal moduli space $\widehat{V}$ with $B\gg$,
the formal moduli space arising from $\gg$.
Then the higher loop space $\maps(\oAA^2, \widehat{V})$ is replaced by $\maps(\oAA^2, B\gg)$,
which can be seen as a formal neighborhood of the trivial bundle in~$\bun_G(\oAA^2)$.
Loosely speaking, we expect that the algebra of spherical operators is controlled by an algebra of differential operator type on this formal moduli space~$\bun_G(\oAA^2)^{\wedge}_0$.

We now give a precise, but weaker, statement, using the result about free cotangent theories.

\begin{lemma}
Let $G$ be a reductive algebraic group.
Consider the algebra $\rho_* \obsq_{\mr{fp}}$ of spherical operators for the theory.
Let $|\gg|$ denote the underlying vector space of~$\gg$.
Using the anti-diagonal filtration, we take the associated graded and obtain a map
\[
U_{\mr{fp}}(\mc{H}_{|\gg|[-1],2}) \to \gr \rho_* \obsq_{\mr{fp}}
\]
of factorization algebras on $\RR_{> 0}$.
\end{lemma}

In other words, the spherical operators for the nonabelian gauge theory should determine a deformation of this dg Weyl algebra, obtained by computing the differentials in the fiberwise polynomial spectral sequence.
The main subtlety here is that the pushforward is sensitive to the complex geometry of the thickened spheres, notably the ratio of the interior to exterior radii,
whereas the locally constant algebra on the left hand side is insensitive to that ratio.
For this reason, the pushforward is not ``just'' differential operators on the formal moduli space.

We expect that deformation to have the following form.
Let $R_2$ denote the graded-commutative algebra $\mr H^\bu(\oAA^2, \mc{O}_{\oAA^2}^{\mr{alg}})$.
Since algebraic functions include canonically into the Dolbeault complex, there is a map from $R_2$ to the Dolbeault cohomology of $\oAA^2$.
For a reductive algebraic group $G$, the graded Lie algebra $\gg[\eps] \otimes R_2$ thus maps canonically to the cohomology of the local Lie algebra encoding the holomorphic twist we are studying.
Given a choice of invariant pairing on the reductive Lie algebra $\gg$, 
the dg algebra 
\[
\clie^\bu_{\mr{Lie}}(\gg[\eps] \otimes R_2) = \mc{O}(B(\gg \otimes R_2))
\]
has a Poisson bracket determined by the pairing on $\gg$ and the residue pairing on~$R_2$.
There is a canonical deformation to a dg associative algebra,
namely the algebra of differential operators on~$B(\gg \otimes R_2)$.

\begin{conjecture}
Let $G$ be a reductive algebraic group.
Consider $\rho_* \obsq_{\mr{fp}}$, the algebra of spherical operators for the holomorphically twisted theory.  There is a map from the algebra of differential operators on~$B(\gg \otimes R_2)$ to the second page of the spectral sequence for the anti-diagonal filtration on $\rho_* \obsq_{\mr{fp}}$.
\end{conjecture}

\subsection{B-type twists}

For simplicity we will focus our attention to the B-twist associated to the parameter $(1,1,0)$, so that we are working with the usual de Rham differential.  
With some care everything in this section could be generalized to a generic B-twist associated to the parameter $(t,t,0)$ for any non-zero $t$.

\begin{remark}
In fact it would also be possible to consider theories associated to the twist at parameter $(t_1,t_2,0)$ where $t_1$ and $t_2$ are both non-zero, but where the two values may differ.  Such twists only carry a de Rham action of $\SO(2) \times \SO(2)$ rather than $\SO(4)$, so the factorization homology can only be computed on products $\Sigma_1 \times \Sigma_2$ of surfaces.
~\hfill$\Diamond$\end{remark}

Let $G$ be a reductive algebraic group.
For any 4-dimensional smooth manifold $M$, consider a principal $G$-bundle $P \to M$ equipped with a flat connection $\nabla$. 
It induces a flat connection on the adjoint bundle $\ad(P) = P \times^G \gg$, which we also denote by $\nabla$.
The space of fields is then
\[
\mc{E}_P = \Omega^{\bu}(M, \ad(P))[1] \oplus \Omega^{\bu}(M, \ad(P)^*)[2],
\]
which is sections of a graded vector bundle $E_P \to M$,
and we equip this graded vector space with the differentials induced by $\nabla$.
The the cohomology of the space of fields is
\[
\mr H^\bu(\mc{E}_P) = \mr H^{\bu}_{\nabla}(M,\ad(P) )[1] \oplus \mr H^{\bu}_{\nabla}(M, \ad(P)^*)[2].
\]
When $P$ is the trivial bundle $G \times M$ with the associated flat connection,
this reduces to
\[
\mr H^\bu_{\nabla}(\mc{E}_{\mr{triv}}) = \mr H^{\bu}_{\mr{dR}}(M )\otimes \gg[1] \oplus \mr H^{\bu}_{\mr{dR}}(M) \otimes \gg^*[2].
\]
For a closed 4-manifold  (i.e., compact and without boundary), 
the cohomology $\mr H^\bu(\mc{E}_P)$ is finite.
Furthermore,  
there is a canonical pairing on~$\mr H^\bu(\mc{E}_P)$ of degree~-1 that is invariant for the shifted Lie bracket.

\begin{lemma}
Let $G$ be an abelian reductive algebraic group.
On a closed 4-manifold $M$ with principal $G$-bundle $P \to M$ and flat connection $\nabla$, 
the cohomology of the classical observables $\obscl(M)$ is isomorphic to~$\symc(\mr H^\bu(\mc{E}_P)^*)$,
which is a finitely-generated free graded-commutative algebra.
\end{lemma}

We remark that this algebra also has a natural 1-shifted Poisson bracket,
arising from the canonical pairing on~$\mr H^\bu(\mc{E}_P)^*$.

\begin{proof}
This result is a special case of results from \cite{Book1};
see especially Chapter 5, Section 6, and Appendix~C.
The argument is completely parallel to that for Lemma~\ref{hol ab class obs},
everywhere replacing Dolbeault complexes with the de Rham complexes for the flat connection~$\nabla$.
\end{proof}

For $G$ nonabelian, the fields $\mc{E}_P$ have a {\em nontrivial} shifted dg Lie algebra structure,
by mixing the wedge product of differential forms with the Lie bracket on $\gg$.
This dg Lie algebra $\mc{E}_P[-1]$ has a natural moduli-theoretic interpretation.  We can identify it with the $(-1)$-shifted tangent complex to the derived stack $T^*[-1]\mr{Flat}_G(M)$, the shifted cotangent to the stack of flat $G$-bundles on $M$, at the point $(P, \nabla)$ (see also the discussion of stacks of this type in Section \ref{B_vacua_section}.)

The cohomology $\mr H^\bu(\mc{E}_P)[-1]$ carries the structure of a graded Lie algebra.

\begin{cor} 
On the closed 4-manifold $M$ with principal bundle $P \to M$ and flat connection $\nabla$, 
there is a canonical filtration on the classical observables $\obscl(M)$ by powers of the augmentation ideal. 
In the associated spectral sequence, the first page is isomorphic to~$\symc(\mr H^\bu(\mc{E}_P)^*)$.
\end{cor}

The next interesting differential in the spectral sequence is the Chevalley differential arising from the Lie algebra structure of~$\mr H^\bu(\mc{E}_P)$.

In some cases, we have a stronger result.
For $M$ a closed 4-manifold that is formal in the sense of rational homotopy theory,
$\mc{E}_P[-1]$ is quasi-isomorphic to its cohomology {\em as a dg Lie algebra},
so that we have the following.

\begin{lemma}
For $M$ a formal closed 4-manifold, 
the classical observables $\obscl(M)$ are quasi-isomorphic to $\clie^\bu_{\mr{Lie}}(\mr H^\bu(\mc{E}_P)[-1])$ as dg commutative algebras.
\end{lemma}

The quantum situation for the B-twist is nice:
thanks to Corollary~\ref{framed E4}, 
our quantization works on any oriented 4-manifold.
Indeed, by Lemma~\ref{obs as det}, we have the following.

\begin{prop}  \label{B_4d_fact_hom_prop}
Let $M$ be an oriented closed 4-manifold, equipped with the trivial principal $G$-bundle.
Then the cohomology of the fiberwise polynomial quantum observables $\obsq_{\mr{fp}}(M)$ is isomorphic to 
\[
\det \left(\mr H^{\bu}_{\mr{dR}}(M )\otimes \gg\right)[d_M],
\]
where $d_M$ modulo 2 agrees with the Euler characteristic.
\end{prop}

\begin{remark}
By the same argument as in Remark \ref{fact_hom_near_a_vacuum_rmk}, we can identify the cohomology of the fiberwise polynomial quantum observables after symmetry breaking to a choice of vacuum $[x] \in [\gg^*/G]$ with 
\[
\det \left(\mr H^{\bu}_{\mr{dR}}(M )\otimes \gg_x\right)[d_M],
\]
where $\gg_x$ is the centralizer of $x$.
~\hfill$\Diamond$\end{remark}

Now take $M$ to be a closed 3-manifold and consider the theory on $M \times \RR$.
In other words, we compactify our 4-dimensional theory along $M$ to produce a 1-dimensional theory in terms of the $\RR$ coordinate (or ``time direction'').
This theory is a one-dimensional topological $\sigma$-model of AKSZ-type, 
where the target space is, in essence, $T^*\Flat_G(M)$, 
where $\Flat_G(M)$ denotes the derived stack of flat $G$-connections on $M$.  In practice, we take a formal moduli space as the target,
typically the formal neighborhood of the trivial $G$-bundle equipped with the trivial connection.

The observables of this compactified theory should offer a deformation quantization of this symplectic derived stack.
In more explicit terms, let $\pi \colon M \times \RR \to \RR$ denote the projection map,
and consider the pushforward factorization algebra $\pi_* \obsq$ on $\RR$.
It is locally constant, and hence it corresponds to an $\EE_1$-algebra,
thanks to the equivalence of $\infty$-categories between $\EE_1$-algebras and locally constant factorization algebras on $\RR$~\cite{LurieHA, AyalaFrancis}.
The $\EE_1$-algebra associated to the pushforwards has a clear conceptual meaning, as we will now see.

Let $\Flat_G(M)^\wedge_0$ denote the formal neighborhood of the trivial flat bundle inside all flat $G$-bundles on~$M$.

\begin{prop} \label{B_3d_fact_hom_prop}
For a reductive algebraic group $G$ with Lie algebra $\gg$,
consider the fiberwise polynomial observables of the B-twisted theory on $M \times \RR$.
The pushforward observables $\pi_* \obscl_{\mr{fp}}$ and $\pi_*\obscl_{\mr{fp}}$ are locally constant on $\RR$ and hence correspond to $\EE_1$-algebras:
\begin{itemize}
\item the classical observables $\pi_* \obscl_{\mr{fp}}$ correspond to $\mc{O}_{\mr{fp}}(T^*\Flat_G(M)^\wedge_0)$, and
\item the quantum observables $\pi_* \obsq_{\mr{fp}}$ correspond to an algebra of differential operator type on~$\Flat_G(M)^\wedge_0$,
in the sense of Definition~\ref{def DOT}.
\end{itemize}
\end{prop}

Before proving this proposition, we will provide a very concrete characterization of the graded associative algebra associated to $\mr H^\bu \pi_* \obsq$,
the strictly locally constant factorization algebra on~$\RR$.
The reader may find thinking through these very explicit constructions useful for appreciating what sort of thing these methods can produce.

Let us start by considering the case of the abelian gauge group $\GG_m$.
In this case, we are working with the derived moduli of flat line bundles --- a kind of derived version of a Jacobian --- 
and it is easy to show that the tangent complex to the trivial flat bundle is modeled by $\Omega^\bu(M)[1] = B(\Omega^\bu(M))$, the shifted de Rham complex.
There are no nontrivial $L_\infty$ brackets, so the formal neighborhood $\Flat_G(M)^\wedge_0$ of the trivial bundle, 
is $\widehat{\Omega^\bu(M)[1]}$, viewed as a formal moduli space.

Taking cohomology simplifies things, so that we get a particularly concrete consequence of the proposition.
If the reader wants a very explicit discussion of how to produce representatives of observables and how the Weyl algebra arises,
we direct them to  Chapter 4 of~\cite{Book1},
where this kind of situation is treated in depth.

\begin{corollary}\label{abelian B def}
For the abelian group $\GG_m$,
consider the fiberwise polynomial observables of the B-twisted theory on $M \times \RR$.
The pushforward observables $\pi_* \obscl_{\mr{fp}}$ and $\pi_*\obsq_{\mr{fp}}$ are locally constant on $\RR$,
and their cohomologies correspond to graded associative algebras:
\begin{itemize}
\item $\mr H^\bu \pi_* \obscl_{\mr{fp}}$ corresponds to $\mc{O}_{\mr{fp}}(T^* \widehat{(\mr H^\bu(M)[1])})$,  and
\item $\mr H^\bu \pi_* \obsq_{\mr{fp}}$ corresponds to differential operators on the formal moduli space~$\widehat{\mr H^\bu(M)[1]}$.
\end{itemize}
\end{corollary}

For $G$ nonabelian, we obtain a related statement using the spectral sequence of the anti-diagonal filtration.

\begin{corollary}\label{nontriv def}
For a reductive algebraic group $G$ with Lie algebra $\gg$,
consider the fiberwise polynomial observables of the B-twisted theory on $M \times \RR$.
The pushforward observables $\pi_* \obscl_{\mr{fp}}$ and $\pi_*\obscl_{\mr{fp}}$ are locally constant on $\RR$.
Moreover, the spectral sequence of the anti-diagonal filtration produces a sequence of locally constant factorization algebras.

On the first page and ignoring the differential, these factorization algebras correspond to the graded associative algebras:
\begin{itemize}
\item $\mr H^\bu \pi_* \obscl_{\mr{fp}}$ corresponds to $\mc{O}_{\mr{fp}}(T^* B(|\gg| \otimes H^\bu(M)))$,  and
\item $\mr H^\bu \pi_* \obsq_{\mr{fp}}$ corresponds to differential operators on $B(|\gg| \otimes \mr H^\bu(M))$,
\end{itemize}
where $|\gg|$ denotes the vector space underlying~$\gg$.
\end{corollary}

The differential on the first page is determined by the cubic term $I_3^{\mr{cl}}$ of the {\it classical} interaction,
which arises from the Lie bracket on $\Omega^{\bu}(M \times \RR) \otimes \gg[\eps]$.
Hence for the classical observables, the first page of the spectral sequence corresponds to the dg associative algebra 
\[
\mc{O}_{\mr{fp}}(T^* B(\gg \otimes \mr H^\bu(M))) \cong \clie^\bu_{\mr{Lie}}(\gg \otimes \mr H^\bu(M) \ltimes (\gg^*\otimes \mr H^\bu(M))[-2] ).
\]
For the quantum observables, the first page of the spectral sequence corresponds to the dg associative algebra of differential operators on~$B(\gg \otimes \mr H^\bu(M))$.
Later pages would yield further deformations.

\begin{proof}[Proof of Proposition \ref{B_3d_fact_hom_prop}]
We begin with the classical observables.
The pushforward $\pi_* \obscl$ is the factorization algebra on $\RR$ that assigns
\[
\clie^\bu_{\mr{Lie}}(\Omega^\bu(I \times M) \otimes \gg[\eps])
\]
to each open set $I \subset \RR$.
Note that the K\"unneth theorem assures us that the inclusion
\[
\Omega^\bu(I) \otimes \Omega^\bu(M) \otimes \gg[\eps] \hookrightarrow \Omega^\bu(I \times M) \otimes \gg[\eps]
\]
is a quasi-isomorphism of dg Lie algebras for every open $I$.
Hence we have a quasi-isomorphism of factorization algebras
\[
\pi_* \obscl \xto{\simeq} \clie^\bu_{\mr{Lie}}(\Omega^\bu(I) \otimes (\Omega^\bu(M) \otimes \gg[\eps]))
\]
Observe that the dg Lie algebra $\Omega^\bu(M) \otimes \gg$ models $\Flat_G(M)^\wedge_0$,
and so 
\[
\clie^\bu_{\mr{Lie}}(\Omega^\bu(I_0) \otimes (\Omega^\bu(M) \otimes \gg[\eps])) \simeq \mc{O}(T^*\Flat_G(M)^\wedge_0)
\]
when $I_0$ is non-empty and connected (i.e. a single interval).
By taking the fiberwise polynomial observables, 
we find that $\pi_* \obscl_{\mr{fp}}$ maps quasi-isomorphically to a factorization algebra that is locally constant (in the dg sense) and manifestly quasi-isomorphic to $\mc{O}_{\mr{fp}}(T^*\Flat_G(M)^\wedge_0)$, as claimed.

(If one wants, one can go farther and give a map to a strictly locally constant factorization algebra. 
Since every open $I$ is a disjoint union of intervals and hence $\CC^{\pi_0(I)} \hookrightarrow \Omega^\bu(I)$ is a quasi-isomorphic by the Poincar\'e lemma,
we have
\[
\CC^{\oplus \pi_0(I)} \otimes \Omega^\bu(M) \otimes \gg \hookrightarrow \Omega^\bu(I \times M) \otimes \gg
\]
is a quasi-isomorphism.
Thus
\[
\pi_* \obscl_{\mr{fp}}(I) \xto{\simeq} \clie^\bu_{\mr{Lie, fp}}(\CC^{\oplus \pi_0(I)} \otimes \Omega^\bu(M) \otimes \gg[\eps]) \cong \mc{O}_{\mr{fp}}(T^*\Flat_G(M)^\wedge_0)^{\otimes \pi_0(I)}
\]
and this map is functorial in the open~$I$.)

We now turn to the quantum observables.
We have already seen that we can equip these observables with an anti-diagonal filtration.
Taking the associated graded, we obtain the fiberwise polynomial observables of a free cotangent theory, namely the B-twisted theory with abelian Lie algebra $|\gg|$.
Consider, for a moment, the polynomial observables for this free theory (i.e., before completing along the base direction $B|\gg|$ and hence getting the fiberwise polynomial observables).
It is the central result of Chapter 4 of \cite{Book1} that the observables of any free theory are an enveloping factorization algebra of a Heisenberg dg Lie algebra.
In this case, it is the Heisenberg dg Lie algebra arising by extending $(\Omega^\bu_c(\RR \times M) \otimes |g|[\eps])[1]$ via its pairing.
This dg Lie algebra is quasi-isomorphic to the Heisenberg dg Lie algebra for $\Omega^\bu_c(\RR) \otimes  \Omega^\bu(M) \otimes |g|[\eps]$,
by the K\"unneth theorem.
Let us denote it by $\mc{H}_{|\gg|,M}$.
Hence there is a quasi-isomorphism
\[
\clie_\bu^{\mr{Lie}}(\mc{H}_{|\gg|,M}) \xto{\simeq} \pi_* \obsq_{|\gg|}
\]
into the observables for the abelian B-twisted theory.
By \cite{Knudsen}, the left hand side corresponds to the enveloping algebra of the (ordinary, unshifted) Heisenberg Lie algebra $(\Omega^\bu(M) \otimes |g|[\eps])[1] \oplus \hbar \CC$.
This implies that $\pi_* \obsq_{|\gg|}$ models a Weyl algebra (or polynomial differential operators).
To get the claim about fiberwise polynomial observables, we simply complete in the ``base direction,''
which we can do at every stage of the argument just given.
\end{proof}

\subsection{A-type twists}

The situation here is, at first inspection, rather simple.
For a classical A-twisted theory expanded around some point $[x] \in [\gg^*/G]$, the BV complex of fields is acyclic,
so that the space of solutions is just an isolated point.
The observables are quasi-isomorphic to $\CC$, the algebra of functions on a point.
There are no interesting $\bb E_4$-algebra deformations here!

On the other hand, we are working over a base stack $[\gg^*/G]$,
and the constant sheaf on a space encodes interesting topological information.
The key example to bear in mind is that the de Rham complex on a sheaf looks locally trivial, by the Poincar\'e lemma,
but it knows interesting global information,
namely cohomology.

Our situation is exactly parallel.
By construction, at a point $x \in \gg^*$, the BV complex of the A-twisted theory 
\[
(\Omega^{\bu,\bu}(\CC^2) \otimes \gg[\eps][1], \dbar + t_1 \partial_{z_1} + t_2 \partial_{z_2} + \eps \ad_x + u \frac \dd{\dd \eps} \id_{\gg_x})
\]
admits a canonical inclusion from a ``constant subcomplex''
\[
(\gg[\eps][1], \eps \ad_x + u \frac \dd{\dd \eps} \id_{\gg_x})
\]
and this map is a quasi-isomorphism by the Poincar\'e lemma.
Hence, there is a canonical quasi-isomorphism 
\[
\obscl_A(\RR^4) \to \clie^\bu_{\mr{Lie}}(\gg[\eps], \eps \ad_x + u \frac \dd{\dd \eps} \id_{\gg_x})
\]
that we will call the ``Poincar\'e map.''
The right hand side is a model of the de Rham complex of the formal neighborhood of $[x]$ in $[\gg^*/G]$.
In this sense, the global classical observables provide a fat model of the de Rham complex on the stack~$[\gg^*/G]$.

When we quantize, we deform the complex of observables,
and we might hope to deform this Poincar\'e map in a compatible way.
Such a construction would allow us to take a quantum observable (e.g., a Wilson loop) and produce a closed differential form on~$[\gg^*/G]$.

This discussion is so far about factorization algebras on $\RR^4$,
where we have seen the observables are quasi-isomorphic to the unit factorization algebra (i.e., it assigns the base field to any open).
It is straightforward to compute factorization homology of the unit algebra on any  manifold.

\printbibliography

\textsc{University of Massachusetts, Amherst}\\
\textsc{Department of Mathematics and Statistics, 710 N Pleasant St, Amherst, MA 01003}\\
\texttt{celliott@math.umass.edu}\\
\texttt{gwilliam@math.umass.edu}\\
\mbox{}\\
\textsc{University of Edinburgh}\\
\textsc{School of Mathematics, James Clerk Maxwell Building, Peter Guthrie Tait Road, Edinburgh, EH9 3FD}\\
\texttt{brian.williams@ed.ac.uk}

\end{document}